\documentclass[onecolumn,secnumarabic,nobibnotes,nofootinbib,superscriptaddress,aps]{revtex4-2}
\usepackage{color, xcolor, colortbl}
\usepackage{graphicx}
\usepackage{geometry}
\usepackage{amsthm,amsmath,amssymb,amsfonts}
\usepackage{dcolumn}
\usepackage{url}
\usepackage{epstopdf}
\usepackage{enumitem}
\usepackage{algorithm}
\usepackage{algpseudocode}
\usepackage{bm}
\usepackage[caption=false]{subfig}
\usepackage{appendix}
\usepackage{multirow}
\usepackage{braket}
\usepackage[english]{babel}
\usepackage{hyperref}
\usepackage[capitalize]{cleveref}
\usepackage{xpatch}
\usepackage{tikz}
\usepackage{adjustbox}
\usepackage{xspace}
\usetikzlibrary{quantikz}

\newcommand{\polylog}{\operatorname{poly}\log}

\DeclareMathOperator*{\argmin}{arg\,min}

\newcommand{\wt}[1]{\widetilde{#1}}

\newcommand{\innerp}[2]{\left\langle #1 \vert #2 \right\rangle}
\newcommand{\abs}[1]{\left\lvert#1\right\rvert}
\newcommand{\norm}[1]{\left\lVert#1\right\rVert}

\newcommand{\Or}{\mathcal{O}}

\newtheorem{thm}{\protect\theoremname}
\newtheorem{lem}[thm]{\protect\lemmaname}

\newtheorem{prop}[thm]{\protect\propositionname}
\newtheorem{cor}[thm]{\protect\corollaryname}

\providecommand{\definitionname}{Definition}
\providecommand{\assumptionname}{Assumption}
\providecommand{\corollaryname}{Corollary}
\providecommand{\lemmaname}{Lemma}
\providecommand{\propositionname}{Proposition}
\providecommand{\remarkname}{Remark}
\providecommand{\theoremname}{Theorem}
\usepackage{bbm}

\makeatletter
\newenvironment{breakablealgorithm}
  {
   \begin{center}
     \refstepcounter{algorithm}
     \hrule height.8pt depth0pt \kern2pt
     \renewcommand{\caption}[2][\relax]{
       {\raggedright\textbf{\fname@algorithm~\thealgorithm} ##2\par}%
       \ifx\relax##1\relax 
         \addcontentsline{loa}{algorithm}{\protect\numberline{\thealgorithm}##2}%
       \else 
         \addcontentsline{loa}{algorithm}{\protect\numberline{\thealgorithm}##1}%
       \fi
       \kern2pt\hrule\kern2pt
     }
  }{
     \kern2pt\hrule\relax
   \end{center}
  }
\makeatother

\usetikzlibrary{fit}
\tikzset{%
  highlight/.style={rectangle,rounded corners,fill=blue!15,draw,fill opacity=0.3,thick,inner sep=0pt}
}

\newcommand{\DeptMath}{Department of Mathematics, University of California, Berkeley, CA 94720, USA}
\newcommand{\LBLMath}{Applied Mathematics and Computational Research Division, Lawrence Berkeley National Laboratory, Berkeley, CA 94720, USA}

\newcommand{\CIQC}{Challenge Institute of Quantum Computation, University
of California, Berkeley, CA 94720, USA}
\newcommand{\REV}[1]{\textcolor{black}{ #1}}
\begin{document}

\title{Even shorter quantum circuit for phase estimation on early fault-tolerant quantum computers with applications to ground-state energy estimation}
\author{Zhiyan Ding}
\email{zding.m@berkeley.edu}
\affiliation{\DeptMath}
\author{Lin Lin}
\email{linlin@math.berkeley.edu}
\affiliation{\DeptMath}
\affiliation{\LBLMath}
\affiliation{\CIQC}
\date{\today}

\begin{abstract}
We develop a phase estimation method with a distinct feature: its maximal runtime (which determines the circuit depth) is $\delta/\epsilon$, where $\epsilon$ is the target precision, and the preconstant $\delta$ can be arbitrarily close to $0$ as the initial state approaches the target eigenstate. The total cost of the algorithm  satisfies the Heisenberg-limited scaling $\widetilde{\mathcal{O}}(\epsilon^{-1})$.
As a result, our algorithm may significantly reduce the circuit depth for performing phase estimation tasks on early fault-tolerant quantum computers. The key technique is a simple subroutine called quantum complex exponential least squares (QCELS). Our algorithm can be readily applied to reduce the circuit depth for estimating the ground-state energy of a quantum Hamiltonian, when the overlap between the initial state and the ground state is large. If this initial overlap is small, we can combine our method with the Fourier filtering method developed in [Lin, Tong, PRX Quantum 3, 010318, 2022], and the resulting algorithm provably reduces the circuit depth in the presence of a large relative overlap compared to $\epsilon$. The relative overlap condition is similar to a spectral gap assumption, but it is aware of the information in the initial state and is therefore applicable to certain Hamiltonians with small spectral gaps. We observe that the circuit depth can be reduced by around two orders of magnitude in numerical experiments under various settings.
\end{abstract}

\maketitle

\tableofcontents

\vspace{2em}
\section{Introduction}\label{sec:intro}

Phase estimation is one of the most important quantum primitives.
The problem of phase estimation can be equivalently stated as estimating the eigenvalue of a quantum Hamiltonian $H$, under the assumption that we can query $H$ via the Hamiltonian evolution operator $U=e^{-i \tau H}$ for some real number $\tau$. There are two important performance metrics for the phase estimation:  the maximal runtime denoted by $T_{\max}$, 
and the total runtime $T_{\mathrm{total}}$, which is the sum of the runtime multiplied by the number of repetitions from each circuit in the algorithm.
$T_{\max}$ and $T_{\mathrm{total}}$ approximately measures the circuit depth and the total cost of the algorithm, respectively, in a way that is independent of the details in implementing $U$.
If we are also given an eigenvector $\ket{\psi}$ associated with an eigenvalue $e^{-i \tau\lambda}$, the Hadamard test is arguably the simplest algorithm for estimating the phase $\lambda\in[-\pi/\tau,\pi/\tau)$. It uses only one ancilla qubit and a single query to $U$ controlled by the ancilla qubit, i.e., $T_{\max}=\tau$. This makes the Hadamard test ideally suited for early fault-tolerant quantum devices, which is expected to have a very limited number of logical qubits and may have difficulty in handling circuit beyond a certain maximal depth. 
The Hadamard test has many drawbacks too: it requires $\ket{\psi}$ to be an exact eigenstate, which is a stringent condition that cannot be satisfied in most scenarios. It also requires $N_s=\Or(\epsilon^{-2})$ repetitions to estimate $\lambda$ to precision $\epsilon$, and hence the total runtime is $\Or(\epsilon^{-2})$. 

Both problems can be addressed by the quantum phase estimation (QPE) and its many variants~\cite{KitaevShenVyalyi2002,PhysRevA.75.012328,NagajWocjanZhang2009,PhysRevLett.103.220502,BerryHiggins2009,HigginsBerryEtAl2007}.  Generically, estimating the phase to $\epsilon$ accuracy with high success probability requires $T_{\max}$ to be at least $\pi/\epsilon$ for QPE~\cite[Section 5.2.1]{NielsenChuang2000}
\footnote{\REV{The analysis in~\cite[Section 5.2.1]{NielsenChuang2000} assumes the phase $\lambda\in[0,1)$. For a general $H$ whose eignevalues are contained in $[-\|H\|_2,\|H\|_2)$, we first define a scaled Hamiltonian $\widetilde{H}=\frac{1}{2}(H/\norm{H}_2+I)$, and apply QPE
with a forward Hamiltonian evolution (i.e., applying $e^{-i t \wt{H}}$ with $t>0$). In order to estimate $\lambda$ to precision $\epsilon$, we need to estimate $\wt{\lambda}=\frac12(\lambda/\norm{H}_2+1)$ to precision $\wt{\epsilon}=\epsilon/(2\norm{H}_2)$ with high success probability, and the scaled maximal evolution time satisfies $\wt{T}_{\max}\geq 2\pi/\wt{\epsilon}$. Therefore the unscaled maximal evolution time (with respect to $U=e^{-itH}$) is $T_{\max}=\wt{T}_{\max}/(2\norm{H}_2)=2\pi/\epsilon$. This analysis also implies that if we use QPE with both forward and backward Hamiltonian evolution (i.e., applying $U=e^{-i t H}$ with $t\in\mathbb{R}$), then the maximal evolution time can be halved and $T_{\max}\geq \pi/\epsilon$.}}.
Additionally, the total runtime of QPE is $\Or(\epsilon^{-1})$ and achieves the Heisenberg-limited scaling~\cite{Gio2011,Zwi2010,Zwi2012}, which is the optimal scaling permitted by quantum mechanics. The standard version of QPE
(see e.g.,~\cite[Chapter 5]{NielsenChuang2000}) uses at least $\log_2
(\pi(\tau\epsilon)^{-1})$ ancilla qubits and is not suitable for early
fault-tolerant devices, but the semi-classical version of
QPE~\cite{GriffithsNiu1996,BerryHiggins2009,HigginsBerryEtAl2007} can achieve the same task and uses only one ancilla qubit.

To our knowledge, in all existing works on QPE satisfying Heisenberg-limited scaling, the maximal runtime  $\pi/\epsilon$ is non-negotiable, in the sense that the preconstant in front of $\epsilon^{-1}$ cannot be significantly reduced in general. This is because QPE type of methods construct, directly or indirectly, a filtering function that transitions from $1$ to $0$ on an interval of width $\epsilon$. This can be a severe limitation in practice, since estimating $\lambda$ to precision $0.001$ means that $U$ needs to be coherently queried for approximately $3000$ times in the quantum circuit (assuming $\tau=1$). It is therefore desirable to have a phase estimation method that satisfies the following properties:
\begin{itemize}
    \item[(1)] Allow $\ket{\psi}$ to be an inexact eigenstate with one ancilla qubit.
    \item[(2)] Maintain the Heisenberg-limited scaling:  To estimate $\lambda$ to precision $\epsilon$ with probability $1-\eta$, the total cost is $\Or(\epsilon^{-1}\polylog(\epsilon^{-1}\eta^{-1}))$;
    \item[(3)] Reduce the circuit depth: the maximal runtime can be (much) lower than $\pi/\epsilon$, especially when $\ket{\psi}$ is close to be an exact eigenstate of $U$.
\end{itemize}

When $\ket{\psi}$ is an exact eigenstate, the maximal runtime in QPE type methods may be reduced by means of a tradeoff between the circuit depth and the number of repetitions. However, if the initial state is not an exact eigenstate, this strategy is no longer directly applicable. In this paper, we introduce algorithms that can satisfy the properties (2) and (3), without assuming that $\ket{\psi}$ is an exact eigenstate.

\begin{figure}
\centering
\includegraphics[height=0.25\textheight]{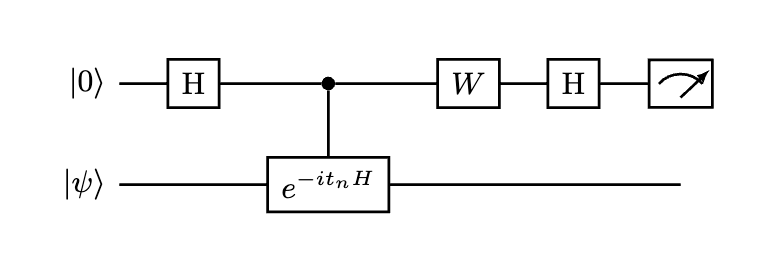}
\caption{Quantum circuit used for collecting the input data. $\mathrm{H}$ is the Hadamard gate, $t_n=n\tau$. Choosing $W=I$ or $W=S^\dagger$ ($S$ is the phase gate) allows us to estimate the real or the imaginary part of $\braket{\psi|\exp(-it_nH)|\psi}$.}
\label{fig:qc}
\end{figure}

\subsection{Main idea}\label{sec:main_idea}

To achieve this, our quantum circuit (\cref{fig:qc}) is the same as the circuit used in the Hadamard test but replaces $U=e^{-i\tau H}$ by $U^n=e^{-i n \tau H}$ for a sequence of integers $n$. This is a simple circuit, uses only one ancilla qubit, and is suitable for early fault-tolerant quantum computers. The most challenging component may be the implementation of the controlled time evolution. Under additional assumptions, the controlled time evolution for certain unitaries may be replaced by uncontrolled time evolution (see e.g.~\cite{Huggins2020,LinTong2022,dong2022ground,LuBanulsCirac2020algorithms,Brien_2021}).

Let $t_n=n\tau$ for $n=0,\ldots,N-1$, and for simplicity we refer to $T_{\max}:=N\tau$ as the maximal running time (the actual maximal running time is $(N-1)\tau$). The circuit provides an estimate of the value of $\braket{\psi|e^{-i t_n H}|\psi}$ by measuring the success probability of the first qubit. Repeated measurements at different $n$ provides a (complex) time series \begin{equation}
\left\{\left(t_n,Z_n\right)\right\}^{N-1}_{n=0}\,,
\label{eqn:input_data}
\end{equation}
where $Z_n$ is a complex-valued random variable such that $\mathbb{E} (Z_n)=\braket{\psi|\exp(-it_nH)|\psi}$. In the intuitive analysis, we may assume $Z_n\approx \braket{\psi|\exp(-it_nH)|\psi}$. We give the detailed construction of $Z_n$ in Section \ref{sec:mqc}.

Without loss of generality, let us denote the target eigenstate by $\ket{\psi_0}$, the target eigenvalue by $\lambda_0$ (in this context, $\lambda_0$ does not need to be the smallest eigenvalue of $H$), and the overlap between the initial state $\ket{\psi}$ and the target eigenstate by $p_0=\abs{\braket{\psi|\psi_0}}^2$. We also assume $\lambda_0\in [-\pi,\pi)$.
If $p_0=1$ (i.e., $\ket{\psi}$ is the target eigenstate) and the number of samples for each $n$ is sufficiently large, we have $Z_n\approx e^{-i t_n \lambda_0}$, which is an exponential function.  If $p_0<1$, we may still fit the input data provided by \cref{eqn:input_data} using a complex exponential $r\exp(-it_n\theta)$, where $r\in\mathbb{C},\theta\in\mathbb{R}$. 

A main step of our method is a subroutine that solves the following nonlinear least squares problem:
\begin{equation}\label{eq:loss_intro}
    (r^*,\theta^*)=\mathrm{argmin}_{r\in\mathbb{C},\theta\in\mathbb{R}} L(r,\theta), \quad
    L(r,\theta)=\frac{1}{N}\sum^{N-1}_{n=0}\left|Z_n-r\exp(-it_n\theta )\right|^2,
\end{equation}
and $\theta^*$ gives the approximation to the phase $\lambda_0$.
Note that once we obtain the data set from the quantum circuit in \cref{fig:qc},
minimizing $L(r,\theta)$ only requires classical computation. This subroutine
is dubbed \textit{quantum complex exponential least squares} (QCELS). The minimization problem can be efficiently solved on classical computers (see \cref{sec:mse}).
An illustrative example of QCELS using the spectrum from TFIM model (see \cref{sec:Ising} for detail) is shown in Fig
\ref{fig:fitting}. In the graph, the initial overlap $p_0=0.8$, the scatter points are the data points $Z_n$
and the curve represents the fitting function $r\exp(-it\theta)$.

\begin{figure}
     \subfloat[\label{fig:fitting}]{
         \centering
         \includegraphics[height=0.25\textheight]{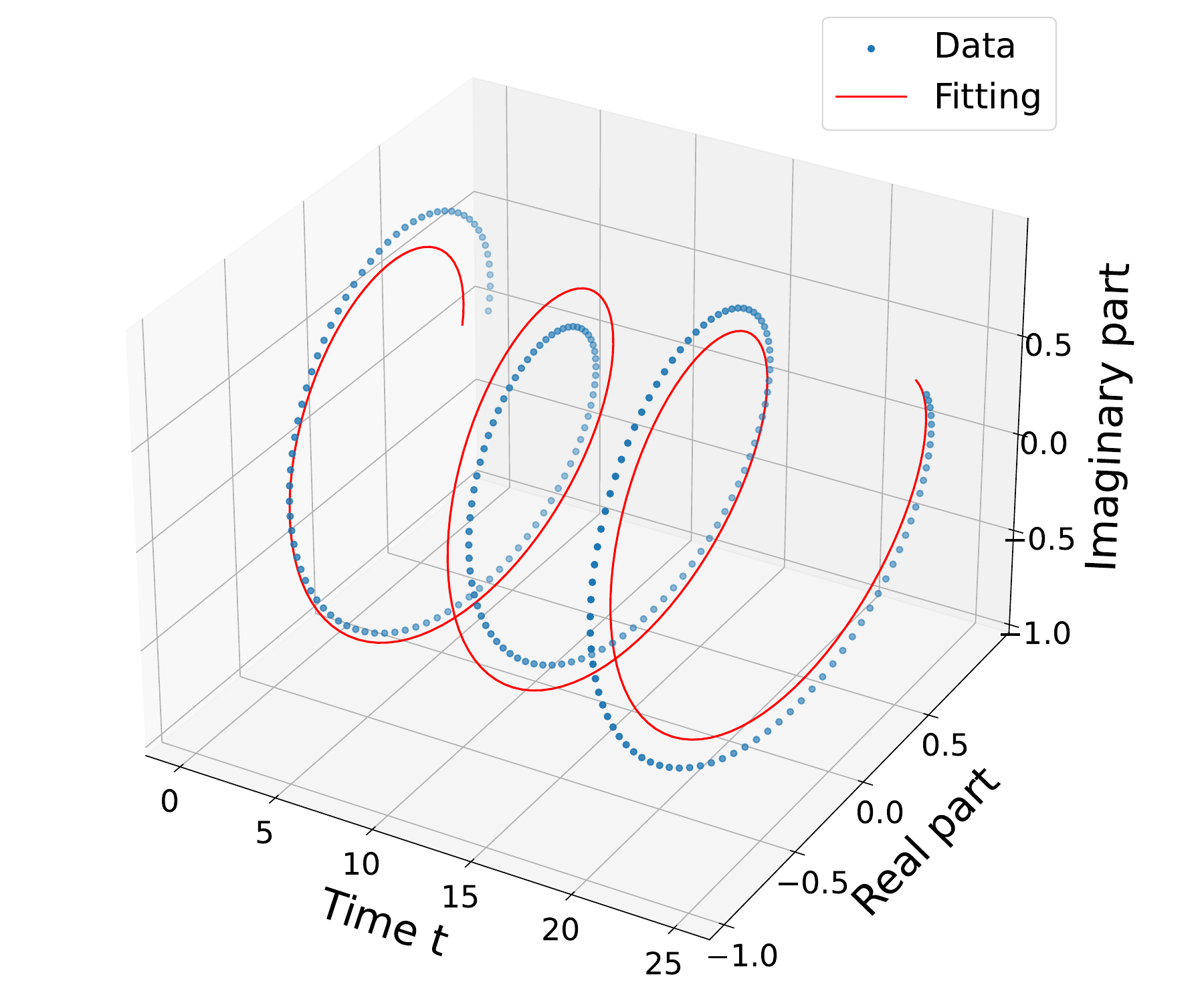}
     }
     \hfill
     \subfloat[\label{fig:theory}]{
         \centering
         \includegraphics[height=0.25\textheight]{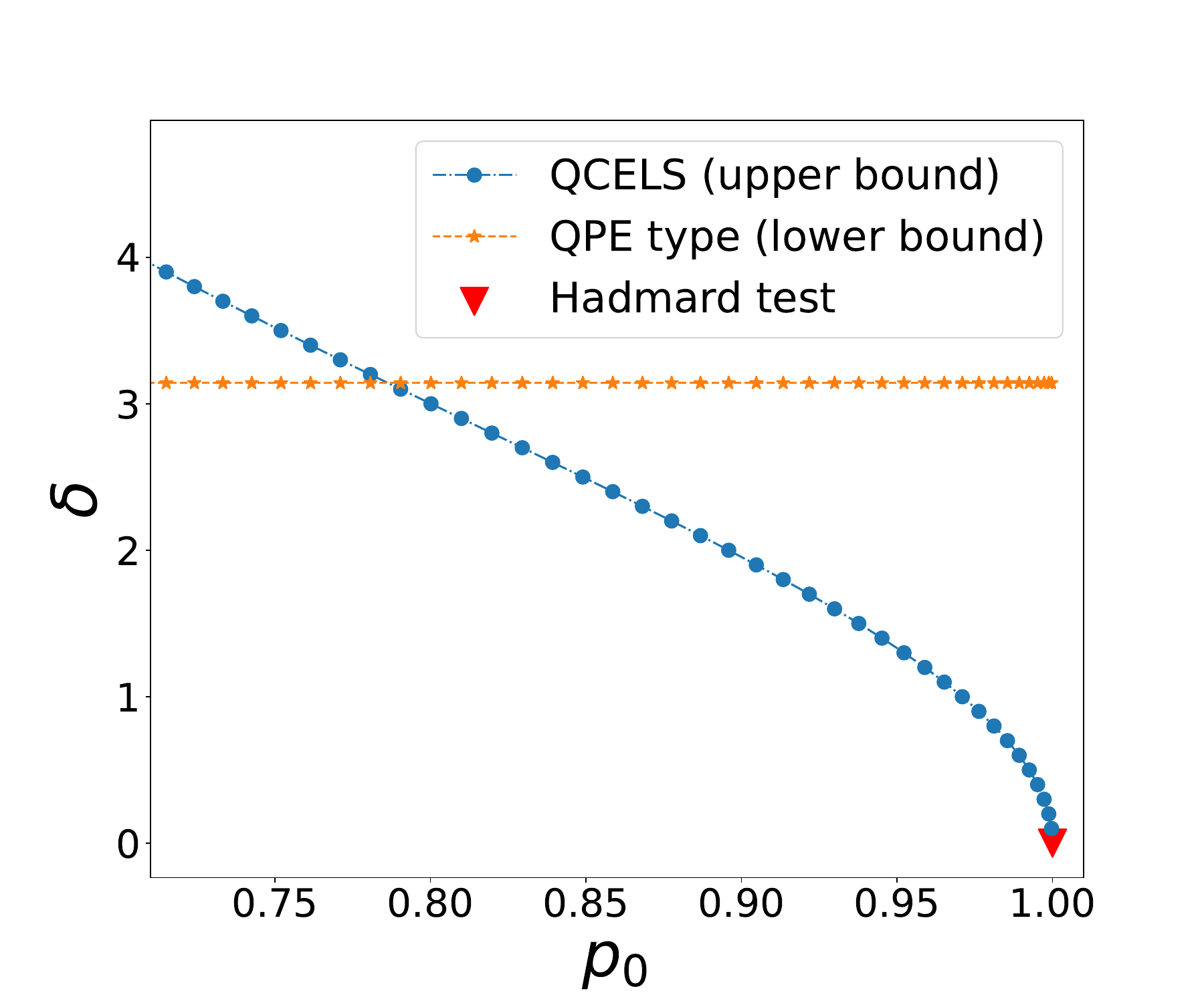}
     }
     \caption{(a) Fitting the noisy input data with $p_0=0.8$ using a complex exponential function. The mismatch between the data and the best fit reflects that the input data is more complex than a single complex exponential function. Despite this mismatch even in the absence of any Monte Carlo sampling error, QCELS is able to accurately estimate the phase under proper conditions. (b) Comparison of the theoretical upper bound of $\delta=T_{\max} \epsilon$  for QCELS ($T_{\max}$ is the maximal runtime) with the lower bound of $\delta$ for QPE type methods when $p_0\ge 0.71$. The Hadamard test is only applicable when $p_0=1$ and in this case $\delta$ can be chosen to be arbitrarily small.   }\label{fig:2}
\end{figure}

If $p_0=1$ and $N=2$, the behavior of QCELS is very similar to the Hadamard test, and
we can estimate $\theta\approx \lambda_0 \bmod [-\pi/\tau,\pi/\tau)$ to any
precision and the circuit depth is independent of $\epsilon$. However, there are some immediate issues with this approach:

\begin{itemize}

\item[(a)] It is not so clear whether \cref{eq:loss_intro} can estimate $\lambda_0$ accurately with a short-depth circuit, even if $p_0$ is close but not equal to $1$. Moreover, this method clearly fails if there exists some eigenstate $\ket{\psi_i}$ such that $|\braket{\psi|\psi_i}|^2>p_0$.
So $p_0$ should be larger than some minimal threshold.

\item[(b)] For each $n$, the number of measurements is at least $1$. If $N=\Theta(\epsilon^{-1})$, the total runtime is at least $N(N-1)/2=\Theta(\epsilon^{-2})$, and the method does not satisfy the Heisenberg-limited scaling.

\end{itemize}

Our main body of work is to address these issues, and to develop an efficient algorithm for post-processing the input time series generated by quantum computers. 

The answer to Question (a) is given by \cref{thm:qcels_simple}. Roughly speaking, when $p_0>0.71$, we may choose a proper $\delta>0$ so that the maximal runtime is $T_{\max}=N\tau=\delta/\epsilon$, and the global minima to \cref{eq:loss_intro} can estimate $\lambda_0$ to precision $\delta/T_{\max}=\epsilon$ ($\bmod\ [-\pi/\tau,\pi/\tau)$). Moreover, when $Z_n$ is sufficiently concentrated around
its expectation and as $p_0\to 1$, $\delta$ can be chosen to be arbitrarily
small. Therefore the maximal runtime (and the circuit depth) can be continuously reduced as the input state approaches an exact eigenstate, and QCELS maintains the desirable behavior of the Hadamard test when $p_0<1$.

To address Question (b), we can start from a small value of $\tau$ which allows us to estimate $\lambda_0$ to precision $\delta/(N\tau)$. If $\delta,N$ are fixed, then this estimate can only reach limited precision. Similar to the binary search strategy for refining the estimate of the eigenvalues~\cite{LinTong2020a,LinTong2022,dong2022ground}, we can refine this estimate by increasing the maximal runtime. Specifically,
we can multiply $\tau$ by constant and repeat the process with fixed $\delta,
N$. We only need to repeat the
process for $J=\log_2(\delta/(N\epsilon))$ times. At the last step, we have $\tau_J=\delta/(N\epsilon)$ and the maximal circuit depth is $T_{\max}=N\tau_J=\delta/\epsilon$.   This procedure is called the \textit{multi-level QCELS} and is described by
\cref{alg:main}. According to \cref{thm:main}, when $p_0\approx 1$, we may choose $\delta=\Theta(\sqrt{1-p_0})\ll 1$ and estimate $\lambda_0$ to precision $\epsilon$. The maximal runtime is $T_{\max}=N\tau=\delta/\epsilon$, and the total cost is
$\widetilde{\Or}(\delta^{-1}\epsilon^{-1})$. 
Both \cref{thm:main} and numerical results verify that \cref{alg:main}
satisfies the desired properties (1)(2)(3) at the beginning of the paper.
In particular, the circuit depth can be continuously adjusted by the parameter
$\delta$ (see a comparison of the theoretical circuit depth of different methods
in Fig \ref{fig:theory}.), and the algorithm satisfies the Heisenberg-limited
scaling for all choices of
$\delta$ within the allowed range determined by $p_0$ and the noise level due to measurements. 

\subsection{Ground-state energy estimation}

As an application, we consider the problem of estimating the ground-state
energy (the algebraically smallest eigenvalue) of an $n$-qubit quantum
Hamiltonian $H$. Here, we assume ground-state energy
$\lambda_0\in [-\pi,\pi)$, and $\ket{\psi_0}$ is the associated eigenvector. In
the absence of additional assumptions, the task can be QMA-hard
~\cite{KitaevShenVyalyi2002,KempaKitaev2005,Aha2009}. Hence we assume that an
initial quantum state $\ket{\psi}=U_I\ket{0^n}$ can be prepared via a unitary
$U_I$, and the overlap $p_0=\abs{\braket{\psi|\psi_0}}^2>0$.  If $p_0\ge 0.71$,
we can readily apply \cref{thm:main} to estimate $\lambda_0$ using a
short-depth circuit. 

If $p_0$ is small, we propose an algorithm combining the multi-level QCELS
algorithm with the Fourier filtering technique developed in
Ref.~\cite{LinTong2022} to estimate $\lambda_0$.  To demonstrate the
efficiency of the algorithm, we assume that there is an interval $I$ containing $\lambda_0$, a slightly larger interval $I'\supset I$ with a positive distance $D$ separating $I$ and $(I')^c$ (see \cref{eqn:def_D}). We introduce a concept called the \emph{relative overlap} of the initial vector $\ket{\psi}$ with the ground state with respect to the intervals $I,I'$: 
\begin{equation}\label{eqn:relative_overlap}
p_{r}(I,I')=\frac{\abs{\braket{\psi|\psi_0}}^2\mathbf {1}_{I}(\lambda_0)}{\sum_{\lambda_k\in I'}\abs{\braket{\psi|\psi_k}}^2}.
\end{equation}
Here the denominator is assumed to be non-vanishing, and $\mathbf {1}_{I}(\cdot)$ is the indicator function on $I$ such that $\mathbf {1}_{I}(\lambda_0)=1$ if $\lambda_0\in I$ and $\mathbf {1}_{I}(\lambda_0)=0$ if $\lambda_0\notin I$.
The most straightforward scenario is that the system has a spectral gap $\Delta=\lambda_1-\lambda_0$. We can then choose $I=[-\pi,\lambda_{\mathrm{prior}}+\Delta/4]$, $I'=[-\pi,\lambda_{\mathrm{prior}}+3\Delta/4]$, and $D=\Delta/2$, where $\lambda_{\mathrm{prior}}$ is a rough estimation of $\lambda_0$ such that $|\lambda_{\mathrm{prior}}-\lambda_0|\leq \Delta/4$. The relative overlap in this case will be $1$. It should be noted that the preceding discussion considers a worst-case scenario. In a real application, even if the spectral gap is very small, it may be feasible to choose suitable values for $I$ and $I'$ that result in a distance $D$ significantly larger than the spectral gap, while still achieving a large relative overlap $p_r(I,I')$. \cref{thm:qcels_small_p0} states that as long as the relative overlap  is larger than $0.71$, we can estimate the ground-state energy to precision $\epsilon$, where the maximal runtime is $T_{\max}=\wt{\Theta}(D^{-1})+\delta/\epsilon$, and the total runtime $T_{\mathrm{total}}$ is approximately $\widetilde{\mathcal{O}}(p^{-2}_0\delta^{-2}(D^{-1}+\delta/\epsilon))$. Hence this algorithm is particularly useful when $D\gg \epsilon$. As the relative overlap approaches $1$, $\delta$ can be chosen to be arbitrarily small. 

\subsection{Related works}\label{sec:related_works}

Based on the generalized uncertainty relation~\cite{Braunstein1996GeneralizedUR}, there exists a uniform complexity lower bound for the problem of phase estimation~\cite{PhysRevLett.96.010401}, i.e., the square of the error is always $\Omega(N_s^{-1}N^{-2})$ in the expectation sense, where $N$ is the query depth (with $\tau=1$) and $N_s$ is the number of repetitions. In our case, to estimate the ground-state energy with precision $\epsilon$, we have $N_s T_{\max}^{2}=\Omega(\epsilon^{-2})$, where $T_{\max}$ is the maximal runtime.
From this perspective, the Hadamard test (with $N_s=\Or(\epsilon^{-2})$ and $T_{\max}=\Or(1)$) and QPE (with $N_s=\Or(1)$ and $T_{\max}=\Or(\epsilon^{-1})$) are at the two ends of the spectrum. It is possible to achieve a measurement-depth trade-off by setting $N_s=\Or(\epsilon^{-2(1-\alpha)}),T_{\max}=\Or(\epsilon^{-\alpha})$
for some $0<\alpha<1$~\cite{Wang2019}. However, the total cost will be at least $N_sT_{\max}=\Or(\epsilon^{\alpha-2})$. Hence when $\alpha<1$, this strategy does not satisfy the Heisenberg-limited scaling. 

Our work is related to the robust phase estimation (RPE) in the context of quantum metrology for single qubit systems, which was first proposed in Ref.~\cite{Higgins_2009}. RPE satisfies the Heisenberg-limited scaling and allows the input state to be an inexact eigenstate, as long as the overlap with the desired eigenstate is larger than a certain constant~\cite{PhysRevA.104.069901}. Due to these advantages, RPE has been applied  in quantum metrology, as well as other systems that can be viewed as effective single qubit systems~\cite{PhysRevLett.118.190502,PhysRevApplied.10.044033,Russo2021EvaluatingED,PhysRevA.104.069901,PhysRevA.102.042613}. Empirical observations also suggest that the maximal runtime of RPE may be  smaller than $\pi/\epsilon$. However, we are not aware of the theoretical analysis on this aspect of the algorithm. 
It is possible to generalize the RPE to perform phase estimation of general $n$-qubit systems.
That being said, QCELS and multi-level QCELS can also be applied for parameter estimation in quantum metrology, with provably short circuit depth.

There are a few other phase estimation algorithms that also use a single ancilla qubit. The efficiency of the algorithms are so far demonstrated numerically. Ref.~\cite{Brien2019} develops a post-processing technique to extract eigenvalues from phase estimation data based on a classical time-series (or frequency) analysis.  \cite{PhysRevA.65.042323} proposes a method that estimates $\left\langle\psi\right|\exp(-itH)\ket{\psi}$ first and then
performs a classical Fourier transform to estimate the eigenvalues. A very different type of algorithms for ground state energy estimation is the variational quantum eigensolver (VQE) \cite{McCleanRomeroBabbushEtAl2016, PeruzzoMcCleanShadboltEtAl2014,Wang2019}, which constructs a variational ansatz $\ket{\psi(\theta)}$ to approximate the lowest eigenvector $\ket{\psi_0}$ and the parameter $\theta$ of the ansatz is adjusted to minimize the energy $\left\langle\psi(\theta)\right|H\ket{\psi(\theta)}$. The advantage of the VQE is that the quantum circuit is very simple because short depth circuits (often without using ancilla qubits) are enough to  estimate  $\left\langle\psi(\theta)\right|H\ket{\psi(\theta)}$. However, the efficiency and accuracy of VQE largely depend on the representation power of
the variational ansatz $\psi(\theta)$, and the solver of the
non-convex optimization problem. Similar to VQE, there are also other algorithms that try to perform phase estimation  using the quantum states generated in the time evolution, such as quantum
imaginary time evolution (QITE) algorithm~\cite{MottaSunTanEtAl2020} and some methods based on classical Krylov
subspace method such as the quantum subspace diagonalization~\cite{Huggins2020,St2020}. However, these methods also lack provable complexity upper bound, and existing theoretical analysis on quantum subspace diagonalization methods~\cite{doi:10.1137/21M145954X} has not been able to reveal the advantage of such methods compared to classical QPE methods.

For ground-state energy estimation, a number of quantum algorithms~\cite{AbramsLloyd1999,GeTuraCirac2019,LinTong2022,dong2022ground,WanBertaCampbell2022,Wang_2022} have been developed for ground-state energy estimation using the Hamiltonian evolution input model. However, the maximal runtime of all existing works satisfying the Heisenberg-limited scaling is at least $C/\epsilon$ for some constant $C=\Omega(1)$ that is independent from the overlap $p_0$. Take the method in Ref.~\cite{LinTong2022} for instance, which uses the same quantum circuit as in~\cref{fig:qc} to generate the input data, and can estimate $\lambda_0$ with Heisenberg-limited scaling for any $p_0>0$. The method uses a Fourier filter to approximate the shifted sign function. To resolve the ground-state energy to precision $\epsilon$, the shifted sign function defined on $[-\pi,\pi)$ should make a transition from $1$ to $0$ within a small interval of size $\epsilon/2$. The maximal runtime is $\wt{\Or}(\epsilon^{-1})$ and the preconstant is larger than $\pi$ (see \cite[Appendix A]{LinTong2022}). A similar mechanism of constructing  filtering functions is used in the near-optimal ground state preparation and ground-state energy estimation algorithm based on the block encoding input model~\cite{LinTong2020a}, the quantum eigenvalue transformation of unitary matrices (QETU) using the Hamiltonian evolution input model~\cite{dong2022ground}, and the statistical approach with a randomized implementation of Hamiltonian evolution~\cite{WanBertaCampbell2022}. 

More recently, Ref.~\cite{Wang_2022} introduces a method that uses the Fourier filtering techniques from \cite{LinTong2022} to generate a rough estimation $\widetilde{\lambda}_0$ for $\lambda_0$ in the first step. Then, it uses a derivative Gaussian filter around $\widetilde{\lambda}_0$ to refine the estimation of $\lambda_0$. The main result~\cite[Corollary 1.3]{Wang_2022} is that if the system has a spectral gap $\Delta$, for any $\alpha\in [0,1]$, the maximal runtime can be chosen to be $\wt{\Or}(\epsilon^{-\alpha}\Delta^{-1+\alpha})$, and the total cost is $\wt{\Or}(\Delta^{1-\alpha}\epsilon^{-2+\alpha})$. When $\alpha=1$, this reduces to the previous result in \cite{LinTong2022}, i.e., both the maximal runtime and the total cost scale as $\wt{\Or}(\epsilon^{-1})$. When $\alpha=0$, the maximal runtime becomes $\wt{\Or}(\Delta^{-1})$ which can be much smaller than $\wt{\Or}(\epsilon^{-1})$ when $\Delta\gg \epsilon$, and this is compensated by an increase in the total cost to $\wt{\Or}(\epsilon^{-2}\Delta)$.
We show in \cref{cor:small_p0} that under the same assumption, we may choose $\delta=\epsilon/\Delta$ in \cref{thm:qcels_small_p0}, and the maximal runtime our method is also $\wt{\Or}(\Delta^{-1})$ and the total runtime is $\wt{\Or}(\epsilon^{-2}\Delta)$. While the maximal runtime allowed by~\cite[Corollary 1.3]{Wang_2022} should be at least $\wt{\Or}(\Delta^{-1})$, our Theorem \ref{thm:qcels_small_p0} allows an even shorter maximal runtime of $\Or(1)$ under proper conditions. For example, in many quantum systems, although the spectral gap $\Delta$ is very small, the relative contribution of the ground state to the initial state is significant in some large interval $\left[\lambda_0,\lambda_0+D\right)$, where $D= \Omega(1)\gg \Delta$. Applying 
the result of Theorem \ref{thm:qcels_small_p0} to this system, the maximal runtime is $\mathcal{O}\left(D^{-1}\right)=\mathcal{O}(1)$.
It may still be difficult to estimate this relative overlap in practice. But unlike the spectral gap, the relative overlap condition is aware of the information in the initial state, and this relaxed condition may significantly increase the applicability range of our algorithm in practice especially for certain Hamiltonians with a small spectral gap (see  numerical examples in Section \ref{sec:Hubbard}).

\subsection{Organization}
The rest of the paper is organized as follows. In Section \ref{sec:mm}, we introduce QCELS and multi-level QCELS by assuming $p_0$ is larger than a certain constant threshold. We also provide an intuitive analysis why QCELS can reduce the maximal runtime when $p_0$ is close to $1$.
In Section \ref{sec:cost}, we analyse the maximal runtime and the total cost of our methods. We then extend the method to any $p_0>0$ in Section \ref{sec:small}. The numerical simulation of our method is provided in Section \ref{sec:ns}, where we mainly compare our method with QPE, followed by discussions and future directions in Section \ref{sec:discuss}.

\section{Main method}\label{sec:mm}
\subsection{Generating input data from quantum circuit}\label{sec:mqc}

In \cref{fig:qc}, we may
\begin{itemize}
    \item[(1)] Set $W=I$, measure the ancilla qubit and define a random variable $X_n$ such that $X_n=1$ if the outcome is $0$ and $X_n=-1$ if the outcome is $1$. Then 
\begin{equation}\label{eqn:X}
\mathbb{E}(X_n)=\mathrm{Re}\left(\left\langle\psi\right|\exp(-in\tau H)\ket{\psi}\right)\,.
\end{equation}

\item[(2)] Set $W=S^\dagger$, measure the ancilla qubit and define a random variable $Y_n$ such that $Y_n=1$ if the outcome is $0$ and $Y_n=-1$ if the outcome is $1$. Then 
\begin{equation}\label{eqn:Y}
\mathbb{E}(Y_n)=\mathrm{Im}\left(\left\langle\psi\right|\exp(-in\tau H)\ket{\psi}\right)\,.
\end{equation}
\end{itemize}

Given two preset parameters $N,N_s>0$ and time step $\tau>0$,  we use the quantum circuit in \cref{fig:qc} to prepare the following data set:
\begin{equation}\label{eqn:dataset}
    \mathcal{D}_{H}=\left\{\left(n\tau,Z_n\right)\right\}^{N-1}_{n=0}\,.
\end{equation}
where $Z_n$ is calculated by running the quantum circuit (\cref{fig:qc}) $N_s$ times. More specifically,
\begin{equation}\label{eqn:Z_n}
Z_n=\frac{1}{N_s}\sum^{N_s}_{k=1}\left(X_{k,n}+iY_{k,n}\right).
\end{equation}
Here $X_{k,n},Y_{k,n}$ are independently generated by the quantum circuit (\cref{fig:qc}) with different $W$ and satisfy \eqref{eqn:X}, \eqref{eqn:Y} respectively. Hence in the limit $N_s\to \infty$, we have
\begin{equation}
Z_n=\left\langle\psi\right|\exp(-in\tau H)\ket{\psi}.
\label{eqn:Zn_expect}
\end{equation}

To prepare the data set in \cref{eqn:dataset}, the maximal simulation time is $T_{\max}=(N-1)\tau$ and the total simulation time is $N(N-1)N_s\tau/2=NN_s T_{\max}/2$. To reduce the complexity of our algorithm, it suffices to find an efficient way to post-processing the data set \eqref{eqn:dataset} so a good approximation to $\lambda_0$ can be constructed with proper choice of  $N,N_s$ and $\tau$.

\subsection{QCELS and its intuitive analysis}\label{sec:mse}

Using the data set \eqref{eqn:dataset}, we define the mean square error (MSE):
\begin{equation}\label{eq:loss}
    L(r,\theta)=\frac{1}{N}\sum^{N-1}_{n=0}\left|Z_n-r\exp(-i\theta n\tau)\right|^2\,,
\end{equation}
where $r\in\mathbb{C}$ and $\theta\in\mathbb{R}$. The approximation to $\lambda_0$ is constructed by minimizing the loss function $L(r,\theta)$: Let 
\begin{equation}\label{eq:loss2}
    (r^*,\theta^*)=\argmin_{r\in\mathbb{C},\theta\in\mathbb{R}} L(r,\theta)\,,
\end{equation}
then $\theta^*$ is an approximation to $\lambda_0$, and this defines the QCELS algorithm. Note that once we obtain the data set from the quantum circuit, minimizing $L(r,\theta)$ only requires classical computation. 

The mean square error $L(r,\theta)$ is a quadratic function with respect to $r$. For a fixed value of $\theta$, minimization with respect to $r$ gives
\begin{equation}
r(\theta)=\frac{1}{N}\sum_{n=0}^{N-1} e^{i \theta n\tau} Z_n,
\end{equation}
and
\begin{equation}
\min_{r\in\mathbb{C}}L(r,\theta)=\frac{1}{N}\sum^{N-1}_{n=0}\abs{Z_n}^2-\frac{1}{N}
\abs{\sum_{n=0}^{N-1}Z_n e^{i\theta n \tau}}^2.
\label{eqn:min_r}
\end{equation}
Therefore minimizing $L(r,\theta)$ is equivalent to maximizing 
\begin{equation}
f(\theta)=
\abs{\sum_{n=0}^{N-1}Z_n e^{i\theta n \tau}}^2,
\label{eqn:max_loss}
\end{equation} 
which is a nonlinear function with respect to $\theta$. According to the definition of $Z_n$ in \eqref{eqn:Z_n} (see Appendix \ref{sec:bound_En} for the rigorous statement)
\[
\left|Z_n-\left\langle\psi\right|\exp(-in\tau H)\ket{\psi}\right|=\Or(N^{-1/2}_s)\,.
\]
Thus, when $N_s\gg 1$ and $p_0\approx1$, intuitively we have
\begin{equation}
Z_n\approx \left\langle\psi\right|\exp(-in\tau H)\ket{\psi}=\exp(i \lambda_0 n\tau )\,,
\end{equation}
which implies
\begin{equation}
f(\theta)\approx
\abs{\frac{\exp(i(\lambda_0-\theta) N\tau)-1}{\exp(i(\lambda_0-\theta) \tau)-1}}^2
= \left|\frac{\sin((\lambda_0-\theta) N\tau/2)}{\sin((\lambda_0-\theta) \tau/2)}\right|^2\,.
\end{equation}
Recall $T_{\max}=N\tau$. When $N$ is large enough, the maximum of $\left|\frac{\sin((\lambda_0-\theta) N\tau/2)}{\sin((\lambda_0-\theta) \tau/2)}\right|$ occurs at $\theta=\lambda_0$ and the closest local maximal $\theta^*$  satisfies $|\theta^*-\lambda_0|\geq \frac{\pi}{T_{\max}}$. Therefore to find the maximal value of $f(\theta)$ on the interval $[-\pi,\pi)$, we may choose a uniform grid of size up to $\lceil T_{\max}\rceil$, and perform gradient ascent from each grid point. By maximizing over the values from all the local maxima, we can robustly find the global maxima of $f$ (and hence the global minima of the loss function $L$). As an illustration, in Figure \ref{fig:compare}, we give an example of the landscape of the loss function and compare the optimization results with different initial guess.

\begin{figure}
     \centering
     \subfloat{
         \centering
         \includegraphics[height=0.3\textheight]{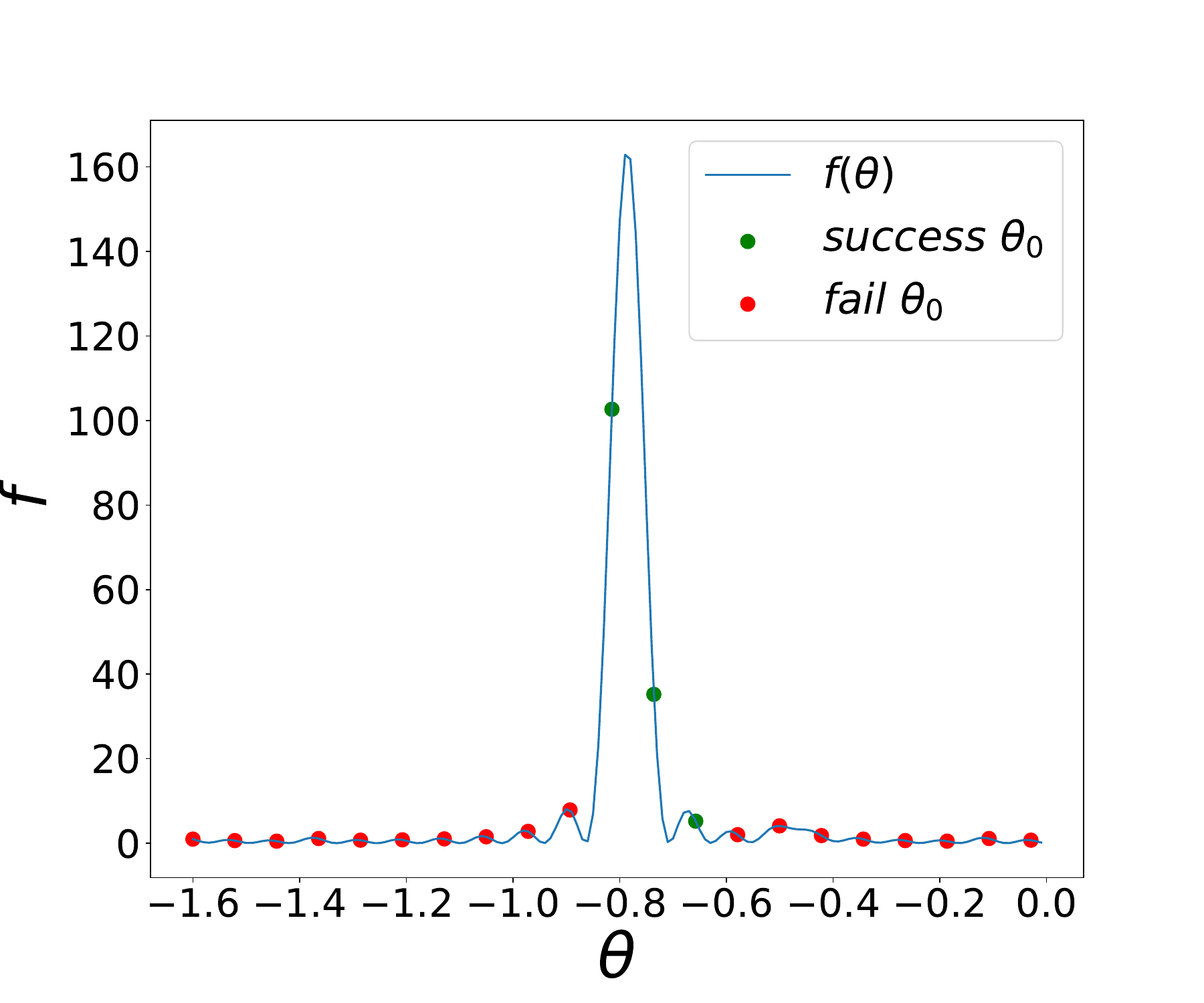}
     }
     \caption{\label{fig:compare} Landscape of the objective function $f(\theta)$ in \cref{eqn:max_loss}, and a number of possible choices of the initial guess $\theta_0$ with $T_{\max}=80$ and eight sites TFIM model (see detail in \cref{sec:Ising}). Here $p_0=0.8$, and the landscape for other values of $p_0$ is similar. The BFGS algorithm is used to maximize the objective function $f(\theta)$ for different initial guess $\theta_0$. The error threshold is set as $0.01$, meaning the optimization problem is successfully solved if $|\theta^*-\mathrm{argmax}_\theta f(\theta)|<0.01$.}
\end{figure}

When $N,N_s\gg 1$ and $p_0$ is sufficiently large, we can show $\theta^*$ is a good approximation to $\lambda_0$ with relatively small depth. The rigorous theoretical results are presented in Section \ref{sec:cost}. Here, we briefly introduce the intuition and leave more details to the next section.

Let $\{(\lambda_m,\ket{\psi_m})\}^{M-1}_{m=0}$ be the set of eigenpairs of the Hamiltonian $H$, and the distance of $\theta^*$ to $\lambda_m$ is 
\begin{equation}
R_m=\left|(\lambda_{m}-\theta^*)\tau\ \mathrm{mod}\ [-\pi,\pi)\right|.
\end{equation}
Our goal is to prove that $R_0$ is small. When $N_s\gg 1$, intuitively we have
\begin{equation}
Z_n\approx \left\langle\psi\right|\exp(-in\tau H)\ket{\psi}=p_0 \exp(-i \lambda_0 n\tau )+\sum^{M-1}_{m=1}p_m\exp(-i \lambda_m n\tau ),
\end{equation}
where $p_m=\left|\innerp{\psi_m}{\psi}\right|^2$ is the overlap between the initial quantum state and the $m$-th eigenvector. Hence in the limit $N_s\to \infty$, 
\begin{equation}
L(r,\theta)= \frac{1}{N}\sum^{N-1}_{n=0}\left|p_0 \exp(-i \lambda_0 n\tau )+\sum^{M-1}_{m=1}p_m\exp(-i \lambda_m n\tau )-r\exp(-i\theta n\tau )\right|^2.
\end{equation}
Similar to the computation above, we find that minimizing the right hand side is equivalent to maximizing the following function:
\begin{equation}
\begin{aligned}
f(\theta)=&\left|\sum^{N-1}_{n=0}\left[p_0 \exp(i (\lambda_0-\theta) n\tau)+\sum^{M-1}_{m=1}p_m\exp(i (\lambda_m-\theta) n\tau)\right]\right|^2\\
=&\left|p_0 \frac{\exp(i (\lambda_0-\theta) N\tau)-1}{\exp(i (\lambda_0-\theta)\tau)-1}+\sum^{M-1}_{m=1}p_m \frac{\exp(i (\lambda_m-\theta) N\tau)-1}{\exp(i (\lambda_m-\theta)\tau)-1}\right|^2\,.
\end{aligned}
\label{eqn:ftheta_limit}
\end{equation}
Therefore
\begin{equation}
f(\lambda_0)=\left|p_0N+\sum^{M-1}_{m=1}p_m \frac{\exp(i (\lambda_m-\lambda_0) N\tau)-1}{\exp(i (\lambda_m-\lambda_0)\tau)-1}\right|^2.
\end{equation}
When the overlap between the initial state and $\ket{\psi_0}$ dominates, i.e., $p_0>\sum^{M-1}_{m=1}p_m$, we can use the first equality in \cref{eqn:ftheta_limit} to obtain
\begin{equation}
\sqrt{f(\lambda_0)}\ge\sum_{n=0}^{N-1} \left(p_0-\sum_{m=1}^{M-1} p_m\right)=N \left(p_0-\sum_{m=1}^{M-1} p_m\right).
\end{equation}
Notice 
\begin{equation}
\left|\frac{\exp(i (\lambda-\theta) N\tau)-1}{\exp(i (\lambda-\theta)\tau)-1}\right|=\left|\frac{\sin((\lambda-\theta)\tau N/2)}{\sin((\lambda-\theta)\tau/2)}\right|\leq \frac{\pi}{|(\lambda-\theta)\tau\ \mathrm{mod}\ [-\pi,\pi)|}\,.
\label{eqn:expinequal_1}
\end{equation}
The second equality of \cref{eqn:ftheta_limit} together with \cref{eqn:expinequal_1} gives
\begin{equation}
\sqrt{f(\theta^*)}\leq\frac{\pi \sum_{m=0}^{M-1} p_m}{\min_{m=0}^{M-1} |(\lambda_m-\theta^*)\tau\ \mathrm{mod}\ [-\pi,\pi)|}= \frac{\pi}{\min_mR_m}.
\end{equation}
As a result, 
\begin{equation}
\left(p_0-\sum^{M-1}_{m=1}p_m\right)N \leq \sqrt{f(\lambda_0)}\leq \sqrt{f(\theta^*)}\leq\frac{\pi}{\min_mR_m}\,.
\label{eqn:expinequal_2}
\end{equation}
\cref{eqn:expinequal_2} implies that there must exist some $m^*$ such that $\left(p_0-\sum^{M-1}_{m=1}p_m\right)R_{m^*}\le\pi/N$ and $\theta^*$ must be close to one of the eigenvalues. Since $p_0>\sum^{M-1}_{m=1}p_m$, it is reasonable to expect this eigenvalue should be $\lambda_0$  ($m^*=0$) and we first have \begin{equation}\label{eqn:priorbound}|(\lambda_{0}-\theta^*)\bmod[-\pi/\tau,\pi/\tau)|\leq \frac{\pi}{T_{\max}\left(p_0-\sum^{M-1}_{m=1}p_m\right)}\,.\end{equation}

When $p_0$ is very close to $1$, we can further improve the bound \eqref{eqn:priorbound}. Since $\theta^*$ is the maximal point,
\begin{equation}
\left(p_0-\sum^{M-1}_{m=1}p_m\right)N \leq \sqrt{f(\lambda_0)}\leq \sqrt{f(\theta^*)}\leq \left|\frac{\sin(NR_0 /2)}{\sin(R_0/2)}\right|+(1-p_0)N\,.
\end{equation} 
where the last inequality comes from the first equality of \eqref{eqn:expinequal_1}. This implies that
\begin{equation}\label{eqn:R0small}
\left|\frac{\sin(NR_0 /2)}{\sin(R_0/2)}\right|\ge (3p_0-2)N.
\end{equation}
Define $\delta=R_0N$, we have $\frac{\sin(N (\delta/2N))}{\sin(\delta/2N)}\ge (3p_0-2)N$. When $\delta<\pi$, using the Taylor expansion, we have
\[
\frac{\sin(N (\delta/2N))}{\sin(\delta/2N)}\approx N\left(1-\frac{\delta^2}{24}\right)\ge (3p_0-2)N\,.
\]
Combining this with \eqref{eqn:R0small}, we have
\begin{equation}
\delta^2\approx 72(1-p_0).
\end{equation}
Therefore as $p_0\to 1$, $\delta=\Theta(\sqrt{1-p_0})\to 0$. Note that $\frac{\sin(N x)}{\sin(x)}$ is monotonically decreasing on $[0,\pi/(2N)]$, if we can prove a loose bound $R_0\in[0,\pi/N]$, then it can be refined to
\begin{equation}
R_0\le \frac{\delta}{N},
\end{equation}
or
\begin{equation}
|(\lambda_{0}-\theta^*)\bmod[-\pi/\tau,\pi/\tau)|\leq \frac{\delta}{T_{\max}}=\epsilon.
\end{equation}
In other words, it suffices to choose the maximal runtime $T_{\max}=\delta/\epsilon$.

The intuitive analysis above summarizes the reason why QCELS can estimate $\lambda_0$ with a short-depth circuit.  
The precise statement is given in \cref{thm:qcels_simple}, and the behavior of the preconstant $\delta$ is demonstrated in \cref{fig:theory}.

\subsection{Multi-level QCELS}\label{sec:malg}

Even though QCELS can reduce the maximal runtime, it does not satisfy the Heisenberg-limited scaling. To see this, note that $T_{\max}=(N-1)\tau=\mathcal{O}(\epsilon^{-1})$. If $\tau$ is a constant, then the total simulation time $T_{\mathrm{total}}$ is $\Omega(\epsilon^{-2})$. We may also attempt to choose $N$ to be a constant and let $\tau=\Or(\epsilon^{-1})$. However, the loss function is a periodic function in $\theta$ with period $2\pi/\tau$.  So we can only obtain $|(\lambda_{0}-\theta^*)\bmod[-\pi/\tau,\pi/\tau)|=\mathcal{O}(\epsilon)$, which is a meaningless estimate since $\pi/\tau=\Or(\epsilon)$. 

In this section, we provide a multi-level QCELS algorithm, which maintains the reduced maximal runtime, and satisfies the Heisenberg-limited scaling. Roughly speaking, we will construct a sequence of data sets $\{\mathcal{D}_{H,j}\}$ using an increasing sequence of $\{\tau_j\}$. The maximal simulation time of the algorithm $T_{\max}=N\tau_J$ and the total simulation time of the algorithm $T_{\mathrm{total}}=\sum^J_{j=1}N(N-1)N_s\tau_j$. The parameters in the algorithm should be chosen properly. The increasing speed of $\tau_j$ should also be chosen properly. If $\tau_j$ increases too slowly, we need more iteration steps which increases the total cost. If $\tau_j$ increases too rapidly, there might exist more than one candidates of the estimation interval for $\lambda_0$ in each iteration. We propose the choice of $\tau_{j+1}=2\tau_j$ (see \cref{eqn:tauj_choice} for the precise choice of $\{\tau_j\}$), and this procedure is similar to Kitaev's algorithm~\cite{KitaevShenVyalyi2002}. Each solution of the optimization problem based on $\{\mathcal{D}_{H,j}\}$ helps us shrink the estimation interval and finally we obtain a small estimation interval for $\lambda_0$. The pseudocode of the multi-level QCELS algorithm is given  in Algorithm \ref{alg:main}.
\begin{breakablealgorithm}
      \caption{Multi-level quantum complex exponential least squares}
  \label{alg:main}
  \begin{algorithmic}[1]
  \State \textbf{Preparation:} Number of data pairs: $N$; number of samples: $N_s$; number of iterations: $J$; sequence of time steps: $\{\tau_j\}^{J}_{j=1}$; Quantum oracle: $\left\{\exp(-i\tau_j H)\right\}^J_{j=1
  }$;
  \State \textbf{Running:}
  \State $\lambda_{\min}\gets-\pi$; $\lambda_{\max}\gets\pi$; \Comment{$[\lambda_{\min},\lambda_{\max}]$ contains $\lambda_0$}
  \State $j\gets 1$;
  \For{$j=1,\ldots,J$}
  \State Generate the data set in \cref{eqn:dataset,eqn:Z_n} using the circuit in \cref{fig:qc} with $t_n=n\tau_j$.
  
  \State Define loss function $L(r,\theta)$ according to \eqref{eq:loss}.
  \State Minimizing the loss function (by maximizing \cref{eqn:max_loss}).
  \[
      (r^*_j,\theta^*_j)\gets\argmin_{r\in\mathbb{C},\theta\in[-\lambda_{\min},\lambda_{\max}]} L(r,\theta)\,,
  \] 
  \State $\lambda_{\min}\gets\theta^*_j-\frac{\pi}{2\tau_j}$;  $\lambda_{\max}\gets\theta^*_j+\frac{\pi}{2\tau_j}$ \Comment{Shrink the search interval by $1/2$}
  \EndFor
  \State \textbf{Output:} $\theta^*$
  \end{algorithmic}
\end{breakablealgorithm}

\section{Complexity analysis with a large overlap}\label{sec:cost}
To analyze the complexity of the  multi-level QCELS method in Algorithm \ref{alg:main}, we need to find an upper bound of the maximal/total simulation time for finding an $\epsilon$-approximation to $\lambda_0$. In this section we assume the initial overlap is large, i.e., $p_0>0.71$. The extension to the small $p_0$ regime will be discussed in next section.

For each $0\leq n\leq N-1$, we define
\begin{equation}\label{eqn:En}
E_n=Z_n-\left\langle\psi\right|\exp(-in\tau H)\ket{\psi}=Z_n-\left(p_0 \exp(-i \lambda_0 n\tau )+\sum^{M-1}_{m=1}p_m\exp(-i \lambda_m n\tau)\right)\,,
\end{equation}
which corresponds to the error that occurs in the expectation estimation \eqref{eqn:Z_n}.  Note that $\{E_n\}$ are independent complex random variables with zero expectation and bounded magnitude. Using classical probability theory, we can give a sharp tail bound for $E_n$ with respect to $N,N_s$, which is important for us to derive a choice of $N,N_s$ in our algorithm. The detailed discussion and tail bounds for $E_n$ can be found in Appendix \ref{sec:bound_En}.

Using a proper tail bound for $E_n$, we can analyze the performance of QCELS in \cref{thm:qcels_simple}. This also corresponds to an iteration in Algorithm \ref{alg:main}.

\begin{thm}[Complexity of QCELS, informal]
\label{thm:qcels_simple}
Let $\theta^*$ be the solution of QCELS in \cref{eq:loss2}. 
Given $p_0>0.71$, $0<\eta<1/2$, $0<\epsilon<1/2$, we can choose 
\begin{equation}\label{eqn:delta_1}
\delta=\Theta(\sqrt{1-p_0}),
\end{equation}
and 
\begin{equation}\label{eqn:NNs_1}
NN_s= \widetilde{\Omega}\left(\delta^{-(2+o(1))}\right),\quad \min\{N,N_s\}=\Omega(1),\quad T=N\tau=\frac{\delta}{\epsilon},
\end{equation}
so that 
\begin{equation}\label{eqn:distance}
\mathbb{P}\left(\left|(\theta^*-\lambda_0)\ \mathrm{mod}\ [-\pi/\tau,\pi/\tau) \right|<\epsilon\right)\geq 1-\eta\,.
\end{equation}
\end{thm}

The precise statement of \cref{thm:qcels_simple} and the proof are given in Appendix \ref{sec:pf_prop}.  To show \cref{eqn:distance}, two parts of the error need to be controlled. First, as discussed before, we should control $E_n$ by increasing the number of samples $N_s$, so that it does not change the loss function too much. This is particularly important as $p_0\to 1$. When the condition \eqref{eqn:NNs_1} is satisfied, the probability of $\abs{E_n}=\Or(\delta^2)=\Or(1-p_0)$ is at least $1-\eta$. 
 The second part of error comes from the pollution from eigenvalues other than $\lambda_0$ when $p_0<1$. As a result, $\delta$ cannot be arbitrarily small and needs to satisfy the relation in \cref{eqn:delta_1} as $p_0\to 1$.

Although \cref{thm:qcels_simple} is a complexity result for one step of QCELS, it also shows that this basic version of QCELS cannot satisfy the Heisenberg-limited scaling: Fix $0<\delta<1$, if we choose the parameter according to condition \eqref{eqn:NNs_1} and $N=\Theta(1)$, we have $\tau=\Theta\left(1/\epsilon\right)$. However, this makes the length of the estimation interval $\left|[-\pi/\tau,\pi/\tau)\right|=\Theta\left(\epsilon\right)$, and the resulting estimation $\theta^*$ is meaningless. To solve this problem, we need to choose $\tau=\Theta(1)$. Then we should set $N=\Theta(1/\epsilon)$ to satisfy the third condition in \eqref{eqn:NNs_1}, this finally makes the total cost $T_{\mathrm{total}}=N(N-1)\tau/2=\Omega(1/\epsilon^2)$, which violates the Heisenberg limit scaling.

\cref{thm:qcels_simple} can be used to describe the maximal runtime of Algorithm \ref{alg:main}. Using \cref{eqn:distance}, we can obtain many candidates of the estimation interval for $\lambda_0$ after solving each minimization problem. On the other hand, we can choose $\tau_j$ properly so that only one of this candidate survives in each iteration. After eliminating other candidates, the estimation in \cref{eqn:distance} can be directly written as $\left|\theta^*-\lambda_0\right|<\frac{\delta}{T}$, which implies $T_{\max}=\frac{\delta}{\epsilon}$ is enough to obtain $\epsilon$ precision in our algorithm. Now, we are ready to introduce the choice of the parameters and the main complexity result of Algorithm \ref{alg:main}:
\begin{thm}[Complexity of multi-level QCELS, informal]
\label{thm:main}
Let $\theta^*$ be the output of Algorithm \ref{alg:main}.
Given $p_0>0.71$, $0<\eta<1/2$, $0<\epsilon<1/2$, we can choose $\delta$ according to \cref{eqn:delta_1},  
\begin{equation}
J=\left\lceil\log_2(1/\epsilon)\right\rceil+1,\quad \tau_j=2^{j-1-\left\lceil\log_2(1/\epsilon)\right\rceil}\frac{\delta}{N\epsilon},\quad\forall 1\leq j\leq J\,.
\label{eqn:tauj_choice}
\end{equation}
Choose $NN_s= \widetilde{\Theta}\left(\delta^{-(2+o(1))}\right)$. Then 
\[
T_{\max}=N\tau_J=\frac{\delta}{\epsilon},\quad T_{\mathrm{total}}=\sum^J_{j=1}N(N-1)N_s\tau_j/2=\widetilde{\Theta}\left(\delta^{-(1+o(1))}\epsilon^{-1}\right)\,,
\]
and
\[
\mathbb{P}\left(\left|(\theta^*-\lambda_0)\bmod [-\pi,\pi)\right|<\epsilon\right)\geq 1-\eta\,.
\]
\end{thm}
The precise statement of \cref{thm:main} and the proof are given in Appendix \ref{sec:pf_thm}. \cref{thm:main} shows that as $p_0\to 1$, the multi-level QCELS algorithm satisfies the Heisenberg-limited scaling, and the maximal runtime can be much smaller $\pi/\epsilon$. 
On the other hand, there is a trade-off between the maximal simulation time and the total simulation time. In particular, $N_sN=\widetilde{\Theta}\left(\delta^{-2}\right)$ diverges as $\delta\rightarrow0$. This implies that, although $T_{\mathrm{total}}$ achieves the Heisenberg-limited scaling, the preconstant may become too large if the circuit depth is forced to be very small.

\section{Ground-state energy estimation with a small initial overlap}\label{sec:small}

When $p_0$ is smaller than the threshold value of $0.71$, our strategy is to find a way to ``increase $p_0$'' in the input data. 
If the system has a spectral gap $\Delta=\lambda_1-\lambda_0\gg \epsilon$, we can then use the algorithm from a previous work \cite{LinTong2022} to construct an eigenvalue filter to effectively filter out the contribution above $\lambda_0+\Delta/2$ in the initial state, using a circuit with maximal runtime $\wt{\Theta}(\Delta^{-1})$. 
The effective value of $p_0$ in the filtered data can be approximately $1$, and the multi-level QCELS algorithm becomes applicable.

The spectral gap is a property of the Hamiltonian.  For many quantum systems of interest, the spectral gap $\Delta$ can be very small. Since QCELS can accurately estimate the eigenvalues starting from an inexact eigenstate, the filtering step does not need to be perfect either if $p_0$ is small.
Consider an interval $I$ containing $\lambda_0$,  a larger interval $I'\supset I$, and we define the distance
\begin{equation}
\label{eqn:def_D}
D=\operatorname{dist}((I')^c,I)=\min_{x_1\notin I',x_2\in I}|x_1-x_2|. 
\end{equation}
Then the \emph{relative overlap} of the initial vector $\ket{\psi}$ with the ground state (as defined in  \cref{eqn:relative_overlap}), denoted by $p_r(I,I')$, plays the role of the effective value of $p_0$. Specifically, if $p_r(I,I') \ge 0.71$, we can effectively filter out the contribution from $(I')^c$ in the initial state using the algorithm in \cite{LinTong2022} with maximal runtime $\wt{\Theta}(D^{-1})$.   
After the filtering operation, the relative overlap  plays the role of $p_0$ in the previous section, and we can apply the multi-level QCELS algorithm to estimate $\lambda_0$ with respect to the filtered data. 

The concept of the relative overlap may allow us to  estimate the ground-state energy for certain small gapped systems with a short-depth circuit, especially when $p_r(I,I')\approx 1$ and $D$ is much larger than the spectral gap $\Delta$. Furthermore, unlike the spectral gap assumption which is a property of the Hamiltonian, the relative overlap is a property of the initial state. This introduces flexibility in the initial state design that may be useful for future explorations. 

\subsection{Algorithm}

The modified algorithm has three steps:

\begin{enumerate}

\item \textbf{Rough estimation of $\lambda_0$:} We obtain two rough estimation intervals $I\subset I'\subset [-\pi,\pi]$ for $\lambda_0$,
meaning that $\lambda_0\in I$ and $p_r(I,I')\approx 1$.

\item \textbf{Eigenvalue filtering to remove high energy contribution:} Define a polynomial  $F_q(x)=\sum^d_{l=-d}\hat{F}_{l,q}\exp(ilx)$ such that: 
\begin{equation}\label{eqn:require_f}
\begin{aligned}
&(1)\,|F_q(x)-1|\leq q,\quad \forall x\in I;\quad (2)\,|F_q(x)|\leq q,\quad \forall  x\in [-\pi,\pi]\setminus I'\,.
\end{aligned}
\end{equation}

We use \cite[Lemma 6]{LinTong2022} to construct $F_q$ that satisfies \eqref{eqn:require_f} with $d=\Theta(D^{-1}\mathrm{polylog}(q^{-1}))$, $\left|\hat{F}_{l,q}\right|=\Theta(|l|^{-1})$, and $\sum^{d}_{l=-d} \left|\hat{F}_{l,q}\right|=\Theta(\log(d))$.

\item \textbf{Refined estimation of $\lambda_0$ with multi-level QCELS:} We can apply Algorithm \ref{alg:main} with the filtered data set (see detail in \cref{alg:main2}) to obtain an accurate estimation of the ground state energy. 
\end{enumerate}

Define
\begin{equation}
\hat{F}_{l,q}=\abs{\hat{F}_{l,q}}e^{i\phi_{l,q}}, \quad \beta_l=\left|\hat{F}_{l,q}\right|/\sum^{d}_{l=-d} \left|\hat{F}_{l,q}\right|.
\end{equation}
The main algorithm is summarized in Algorithm \ref{alg:main2}. Here the ``DataGenerator''  is used to filter out the high energy contribution. According to the construction of $F_{q}$ and $Z_{n,q}$, we have
\[
\begin{aligned}
\left|Z_{n,q}-p_0\exp(-i \lambda_0 n\tau )\right|&\approx \left|p_0 (F_q(\lambda_0)-1)\exp(-i \lambda_0 n\tau )+\sum^{M-1}_{m=1}p_mF_q(\lambda_m)\exp(-i \lambda_m n\tau )\right|\\
&\leq \frac{\left(1-p_r(I,I')\right)p_0}{p_r(I,I')}+q\,.
\end{aligned}
\]
This implies the new data set successfully removes the high energy contribution when $q\ll 1$, $p_r(I,I')\approx 1$ and the solution $\theta^*$ to the optimization problem \eqref{eq:loss_small_p0} should be a good approximation to $\lambda_0$. We also need a sequence of data set $\left\{\left(n\tau_j,Z_{n,q}\right)\right\}^{N-1}_{n=0}$ with an increasing sequence of $\{\tau_j\}^J_{j=1}$ to shrink and keep the correct estimation interval. Similar to \cref{alg:main}, we propose the choice of $\tau_{j+1}=2\tau_j$ (see \cref{thm:qcels_small_p0} for the precise choice of $\tau_j$ and $J$).

\begin{breakablealgorithm}
      \caption{Multi-level QCELS based ground-state energy estimation with  small initial overlap}
  \label{alg:main2}
  \begin{algorithmic}[1]
  \State \textbf{Preparation:} Number of data pairs: $N$; number of samples: $N_s$; number of iterations: $J$; sequence of time steps: $\{\tau_j\}^{J}_{j=1}$;
  \State \textbf{Prepare a rough estimation:} 
  
  \State Generate estimation intervals of $\lambda_0$ such that $\lambda_0\in I\subset I'$.
  \State \textbf{Running:}
  \State $\lambda_{\min}\gets-\pi$; $\lambda_{\max}\gets\pi$; \Comment{$[\lambda_{\min},\lambda_{max}]$ is the estimation interval of $\lambda_0$}
  \State $j\gets 1$;
  \While{$j\leq J$}
  \State Generate a data set $\left\{\left(n\tau_j,Z_{n,q}\right)\right\}^{N-1}_{n=0}$ using DataGenerator($\tau_j,d,N_s$). 
  
  \State Define loss function $L_{j}(r,\theta)$ 
  \begin{equation}\label{eq:loss_small_p0} 
    L_{j}(r,\theta)=\frac{1}{N}\sum^{N-1}_{n=0}\left|Z_{n,q}-r\exp(-i\theta n\tau_j)\right|^2\,.
\end{equation}
  \State Minimizing loss function: $      (r^*_j,\theta^*_j)\gets\mathrm{argmin}_{r\in\mathbb{C},\theta\in[-\lambda_{\min},\lambda_{\max}]} L_j(r,\theta)$.
  \State $\lambda_{\min}\gets\theta^*_j-\frac{\pi}{2\tau_j}$;  $\lambda_{\max}\gets\theta^*_j+\frac{\pi}{2\tau_j}$ \Comment{Shrink the search interval by $1/2$.}
  \State $j\gets j+1$
  \EndWhile
  \State \textbf{Output:} $\theta^*$
    \end{algorithmic}
    \hfill
\begin{algorithmic}[1]
\Function{DataGenerator}{$\tau,d,N_s$} \Comment{Generate the filtered data set}
  \State $k\gets 1$;
  \While{$k\leq N_s$}
  \State Generate a random variable $r\in [-d,d-1,\dots,d]$ with the distribution $\mathbb{P}(r=l)=\beta_l$.
  \State Run the quantum circuit (Figure \ref{fig:qc}) with $t_n=r+n\tau$ and $W=I$ to obtain $\widetilde{X}_{k,n}$.
  \State Run the quantum circuit (Figure \ref{fig:qc}) with $t_n=r+n\tau$ and $W=S^\dagger$ to obtain $\widetilde{Y}_{k,n}$.
  \State $Z_{k,n,q}\gets\left(\widetilde{X}_{k,n}+i\widetilde{Y}_{k,n}\right)\exp(i\phi_{r,q})\left(\sum^{d}_{j=-d} \left|\hat{F}_{j,q}\right|\right)$.
  \State $k\gets k+1$
  \EndWhile
 \State $Z_{n,q}\gets\frac{1}{N_s}\sum^{N_s}_{k=1}Z_{k,n,q}$.
\EndFunction
  \end{algorithmic}
\end{breakablealgorithm}

\subsection{Complexity analysis}
In this section, we analyze the complexity of \cref{alg:main2} to show that it can reduce the circuit depth and maintain the Heisenberg-limited scaling. Define the expectation estimation error:
\begin{equation}\label{eqn:Enq}
E_{n,q}=Z_{n,q}-\left\langle\psi\right|F_q(H)\exp(-in\tau H)\ket{\psi}=\frac{1}{N_s}\sum^{N_s}_{k=1}Z_{k,n,q}-\mathbb{E}(Z_{k,n,q})\,,
\end{equation}
and 
\[
G_{n,q}=Z_{n,q}-p_0\exp(-i \lambda_0 n\tau)=E_{n,q}+p_0(F_q(\lambda_0)-1)\exp(-i\lambda_0n\tau)+\sum^{M-1}_{k=1}p_kF_q(\lambda_k)\exp(-i\lambda_kn\tau)\,.
\]

Using \eqref{eqn:require_f}, we obtain 
\[
|G_{n,q}|\leq |E_{n,q}|+q+\frac{\left(1-p_r(I,I')\right)p_0}{p_r(I,I')}\,.\quad
\]
By choosing a large $N_s$ to reduce the expectation error $|E_{n,q}|$, increasing the quality of the filter to reduce the approximation error $q$, and assuming $p_r(I,I')\approx 1$, we can effectively reduce the error $G_{n,q}$. These choices will let us find a good approximation of $\lambda_0$ by solving the optimization problem. 
In Algorithm \ref{alg:main2}, constructing the loss function contains two steps. First, to construct the eigenvalue filter, the required circuit depth is $d=\widetilde{\Theta}(D^{-1})$. We then combine the eigenvalue filter with the Algorithm \ref{alg:main} to construct the filtered data set. This increases the circuit depth to $T_{\max}=d+\delta/\epsilon=\widetilde{\Theta}(D^{-1})+\delta/\epsilon$. To construct an enough accurate loss function, the number of repetitions is $NN_s=\widetilde{\Theta}\left(p^{-2}_0\delta^{-(2+o(1))}
\right)$. This implies the total evolution time  $T_{\mathrm{total}}=\widetilde{\Theta}\left(p^{-2}_0\delta^{-(2+o(1))}\left(D^{-1}+\delta/\epsilon\right)\right)$.

The result is summarized in the following theorem:
\begin{thm}[Complexity of Algorithm \ref{alg:main2}, informal]\label{thm:qcels_small_p0} Given any failure probability $0<\eta<1$, target precision $0<\epsilon<1/2$, and knowledge of the relative overlap $p_r(I,I')\ge 0.71$, we can set $\delta= \Theta\left(\sqrt{1-p_r(I,I')}\right)$, $d=\widetilde{\Theta}(D^{-1})$, $q=\Theta(p_0\delta^2)$, $NN_s= \widetilde{\Omega}\left(p^{-2}_0\delta^{-(2+o(1))}\right)$, $\min\{N,N_s\}= \widetilde{\Omega}(p^{-2}_0)$, and $J,\tau_j$ according to \cref{thm:main} with $T=N\tau_J=\delta/\epsilon$. Then,
\[
\mathbb{P}\left(\left|(\theta^*-\lambda_0)\ \mathrm{mod}\ [-\pi,\pi) \right|<\epsilon\right)\geq 1-\eta\,,
\]
where $\theta^*$ is the output of \cref{alg:main2}. In particular,  the maximal evolution time is $T_{\max}=d+\delta/\epsilon=\widetilde{\Theta}(D^{-1})+\delta/\epsilon$ and the total evolution time is $T_{\mathrm{total}}=\widetilde{\Theta}\left(p^{-2}_0\delta^{-(2+o(1))}\left(D^{-1}+\delta/\epsilon\right)\right)$.
\end{thm}
The detailed statement and the proof of this theorem can be found in Appendix \ref{sec:complexity_small_p0}. This theorem is an analogue of \cref{thm:main}. As a special case, we assume the spectral gap $\Delta=\lambda_1-\lambda_0$ is much larger than the precision $\epsilon$. We can construct $I=[-\pi,\lambda_{\mathrm{prior}}+\Delta/4]$, $I'=[-\pi,\lambda_{\mathrm{prior}}+3\Delta/4]$, and $D=\Delta/2$, where $\lambda_{\mathrm{prior}}$ is a rough estimation of $\lambda_0$ such that $|\lambda_{\mathrm{prior}}-\lambda_0|\leq \Delta/4$. 
Then \cref{thm:qcels_small_p0} gives the following complexity estimate:

\begin{cor}[Complexity of Algorithm \ref{alg:main2} with a spectral gap]\label{cor:small_p0}Given any $0<\delta<1$,   failure probability $0<\eta<1$,  target precision $0<\epsilon<1/2$, and spectral gap $\Delta=\lambda_1-\lambda_0$, we can set $d=\widetilde{\Theta}(\Delta^{-1})$, $q=\widetilde{\Theta}\left(p_0\delta^{2}\right)$,  $NN_s= \widetilde{\Theta}\left(p^{-2}_0\delta^{-(2+o(1))}
\right)$, $\min\{N,N_s\}= \widetilde{\Theta}(p^{-2}_0)$, and $J,\tau_j$ according to \cref{thm:main}.  Then, 
\[
\mathbb{P}\left(\left|(\theta^*-\lambda_0)\ \mathrm{mod}\ [-\pi,\pi) \right|<\epsilon\right)\geq 1-\eta\,,
\]
In particular, we have
\begin{itemize}
    \item Cost of preparing rough estimation (construct $\lambda_{\mathrm{prior}},I,I'$): The maximal evolution time $T_{\max,1}=\widetilde{\Theta}(\Delta^{-1})$, the total evolution time $T_{\mathrm{total},1}=\widetilde{\Theta}(\Delta^{-1}p^{-2}_0)$.
    \item Cost of constructing the loss function: The maximal evolution time $T_{\max,2}=d+\delta/\epsilon=\widetilde{\Theta}(\Delta^{-1})+\delta/\epsilon$, the total evolution time $T_{\mathrm{total},2}=\widetilde{\Theta}\left(p^{-2}_0\delta^{-(2+\zeta)}\left(\Delta^{-1}+\delta/\epsilon\right)\right)$.
\end{itemize}
\end{cor}
In \cref{cor:small_p0}, if $\epsilon\ll \Delta$, we can choose $\delta=\epsilon/\Delta$, then $T_{\max}=\max\{T_{\max,1},T_{\max,2}\}\approx\widetilde{\Theta}(\Delta^{-1})$ and $T_{\mathrm{total}}=T_{\mathrm{total},1}+T_{\mathrm{total},2}\approx\widetilde{\Theta}(p^{-2}_0\Delta\epsilon^{-2})$. This recovers the results of \cite[Theorem 1.1]{Wang_2022}. Note that with such a choice of $\delta$, the total cost does not satisfy the Heisenberg-limited scaling.

\section{Numerical simulation}\label{sec:ns}
In this section, we numerically demonstrate the efficiency of our method using two different models. In \cref{sec:Ising}, we assume a large initial overlap, and compare the performance of Algorithm \ref{alg:main} with QPE for a transverse-field Ising model. In \cref{sec:Hubbard}, we assume the initial overlap is small, and compare the performance of Algorithm \ref{alg:main2} with QPE for a Hubbard model. The Hamiltonian is constructed using the QuSpin package~\cite{WeinbergBukov2017}. Our numerical experiments are available via github (\url{https://github.com/zhiyanding/QCELS}).   

In our numerical experiments, we normalize the spectrum of original Hamiltonian $H$ so that the eigenvalues belong to $[-\pi/4,\pi/4]$. Given a Hamiltonian $H$, we define the normalized Hamiltonian:
\begin{equation}\label{eqn:normalize_H}
\widetilde{H}=\frac{\pi H}{4\|H\|_2}\,.
\end{equation}
We then use the QCELS based \cref{alg:main} or \cref{alg:main2}, as well as QPE to estimate the smallest eigenvalue of $\widetilde{H}$ and measure the error accordingly.

\subsection{Ising model}\label{sec:Ising}
Consider the one-dimensional transverse field Ising model (TFIM) model defined on $L$ sites with periodic boundary conditions:
\begin{equation}\label{eqn:H_Ising}
H=-\left(\sum^{L-1}_{i=1} Z_{i}Z_{i+1}+Z_{L}Z_1\right) -g\sum^m_{i=1} X_i
\end{equation}
where $g$ is the coupling coefficient, $Z_i,X_i$ are Pauli operators for the $i$-th site and the dimension of $H$ is $2^L$.
We choose $L=8,g=4$. We apply Algorithm \ref{alg:main} (referred to as QCELS for simplicity in this subsection) and QPE to estimate $\lambda_0$ of the normalized Hamiltonian $\widetilde{H}$ (see \cref{eqn:normalize_H}). In the test, we set $p_0=0.6,0.8$ and implement QCELS (with $N=5,N_s=100$) and QPE $10$ times to compare the averaged error. The comparison of the results is shown in \cref{fig:Ising}. The errors of both QPE and QCELS are proportional to the inverse of the maximal evolution time ($T_{\max}$). But the constant factor $\delta=T\epsilon$ of QCELS is much smaller than that of QPE. \cref{fig:Ising} shows that QCELS reduces the maximal evolution time by two orders of magnitude, even  in this case when $p_0=0.6$ is smaller than the theoretical threshold $0.71$. This suggests that the numerical performance of QCELS can be significantly better than the theoretical prediction in \cref{thm:main}.
The error of QPE is observed to scale as $6\pi/T$. Moreover, the total evolution time ($T_{\mathrm{total}})$ of QCELS is also smaller (by nearly an order of magnitude) than that of QPE. 
 
\begin{figure}[htbp]
     \centering
     \subfloat{
         \centering
         \includegraphics[width=0.48\textwidth]{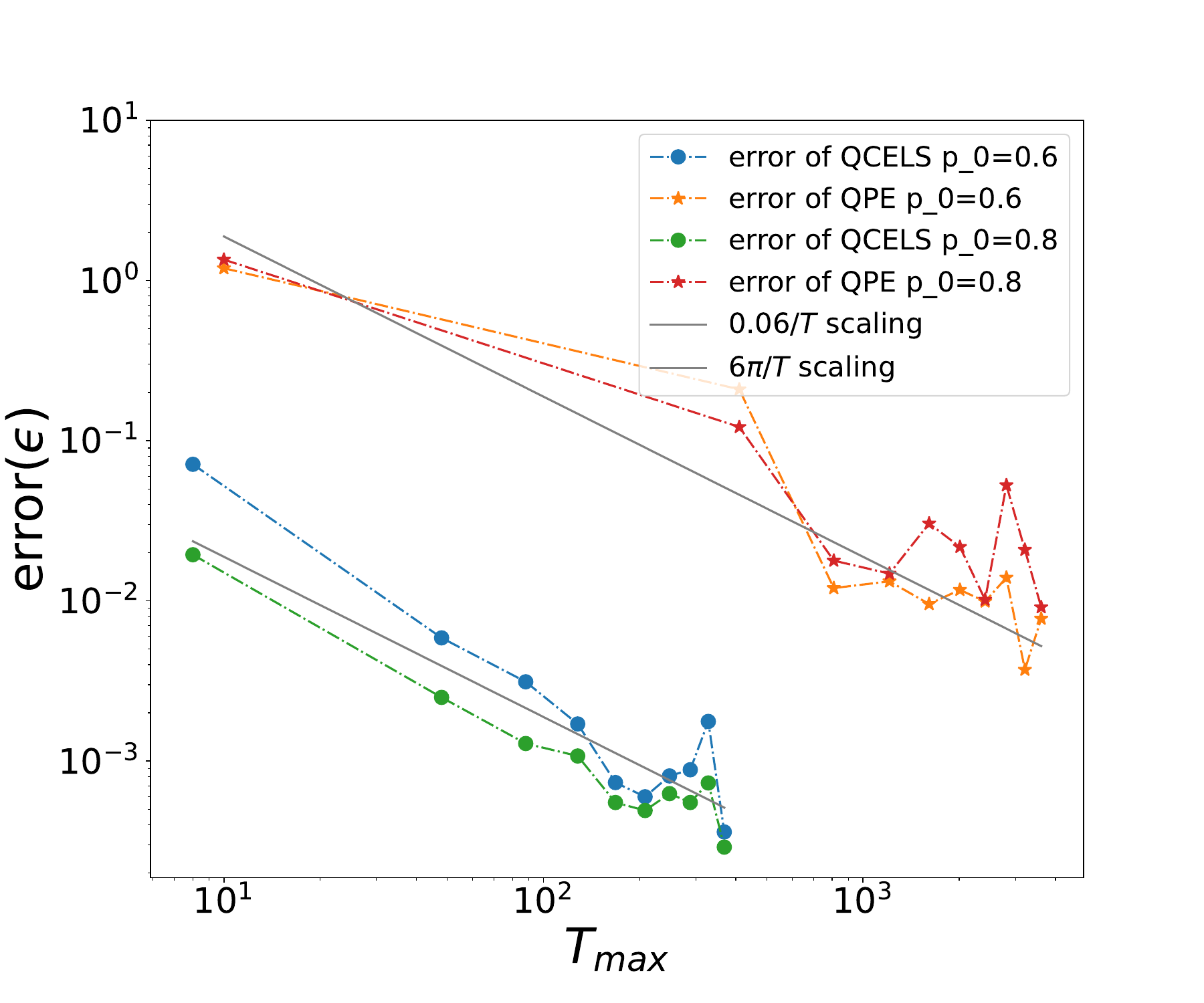}
     }
     \hfill
     \subfloat{
         \centering
         \includegraphics[width=0.48\textwidth]{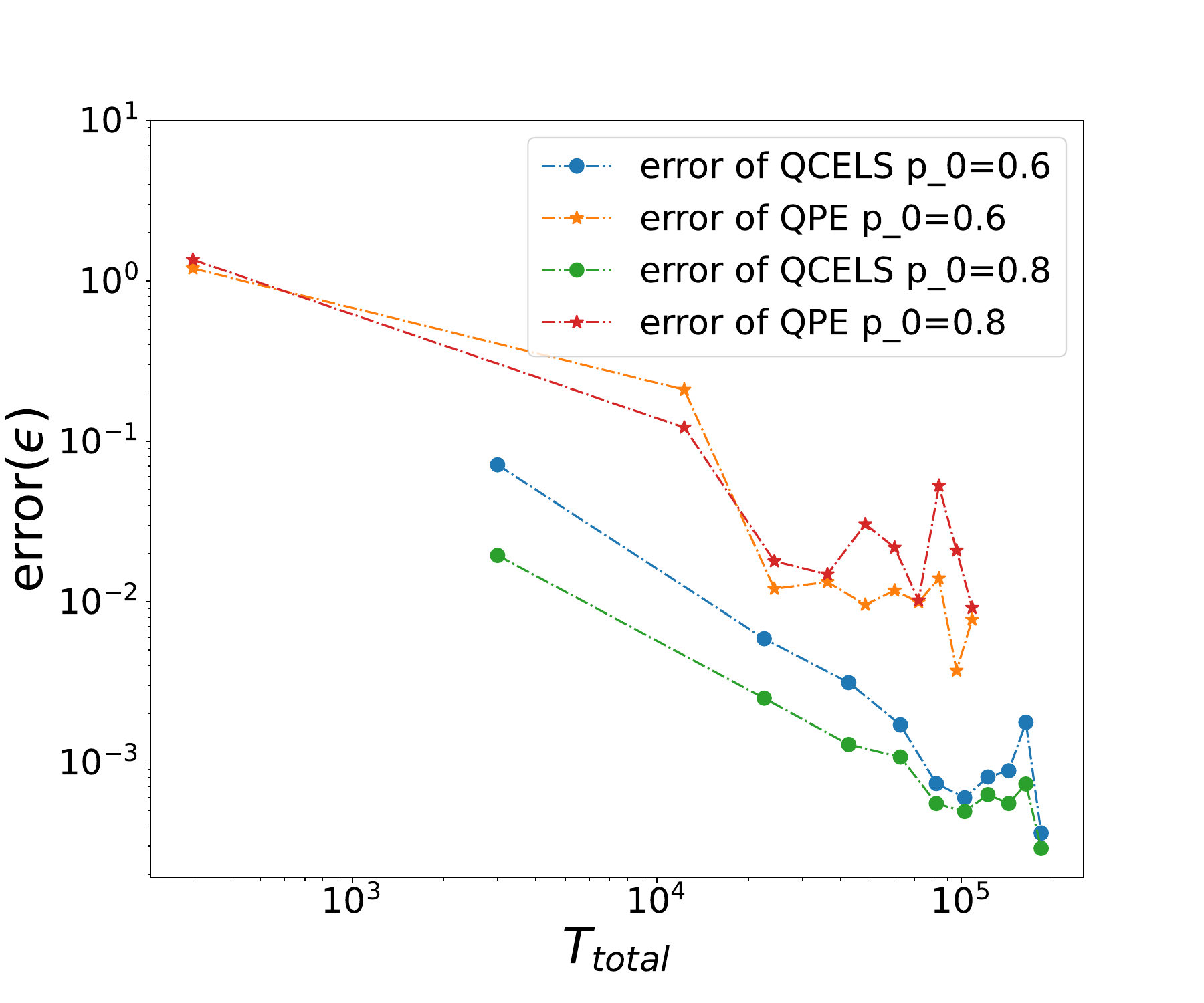}
     }
     \caption{
     \label{fig:Ising} QPE vs QCELS in TFIM model with 8 sites. The initial overlap is large ($p_0=0.6,0.8$). Left: Depth ($T_{\max}$); Right: Cost ($T_{\mathrm{total}}$). For QCELS, we choose $N=5, N_s=100$. $J$ and $\tau_j$ are chosen according to \cref{thm:main}. Both methods have the error scales linearly in $1/T_{\max}$.  The constant factor $\delta=T\epsilon$ of QCELS is much smaller than that of QPE.}
\end{figure}

\subsection{Hubbard model}\label{sec:Hubbard}
Consider the one-dimensional Hubbard model defined on $L$ spinful sites with open boundary conditions
\[
H=-t\sum^{L-1}_{j=1}\sum_{\sigma\in\{\uparrow,\downarrow\}}c^\dagger_{j,\sigma}c_{j+1,\sigma}+U\sum^L_{j=1}\left(n_{j,\uparrow}-\frac{1}{2}\right)\left(n_{j,\downarrow}-\frac{1}{2}\right).
\]
Here $c_{j,\sigma}(c^\dagger_{j,\sigma})$ denotes the fermionic annihilation (creation) operator on the site $j$ with spin $\sigma$. $\left\langle\cdot,\cdot\right\rangle$ denotes sites that are adjacent to each other. $n_{j,\sigma}=c^\dagger_{j,\sigma}c_{j,\sigma}$ is the number operator. 

We choose $L=4,8$, $t=1$, $U=10$. To implement Algorithm \ref{alg:main2} (also referred to as QCELS for simplicity in this subsection) and QPE, we normalize $H$ according to \cref{eqn:normalize_H} and choose a small initial overlap ($p_0=0.1,0.4$).  Following the method in \cref{sec:small}, we first use the algorithm in \cite{LinTong2022} to find a rough estimation $\lambda_{\mathrm{prior}}$ of $\lambda_0$ such that $\left|\lambda_{\mathrm{prior}}-\lambda_0\right|\leq \frac{D}{2}$, where $D$ is chosen properly so that the relative overlap $p_{r}(I,I')>0.75$ with intervals $I=\left[-\pi,\lambda_{\mathrm{prior}}+D/2\right]$ and $I'=\left[-\pi,\lambda_{\mathrm{prior}}+3D/2\right]$. In our test, we set $D=(\lambda_K-\lambda_0)/4$ with $K=\sum^K_{k=1}p_k>p_0/3$. We find that the (normalized) relative gap ($D$) is $0.63$ and $0.26$ for $L=4,8$, respectively. This is significantly larger than the spectral gap, which is $0.018$ and $0.005$ for $L=4,8$, respectively.

After obtaining the rough estimation $\lambda_{\mathrm{prior}}$, we construct the eigenvalue filtering $F_q$ according to \cite[Lemma 6]{LinTong2022} to separate $I,I'$. Noticing $\mathrm{dist}(I,(I')^c)=D$, we set $d=\left\lfloor 15/D\right\rfloor$ to ensure a small enough approximation error $q$. We run QCELS with $N=5$ and $N_s=\left\lfloor 15p_0^{-2}\log(d)\right\rfloor$ and QPE $5$ times to compare the averaged error. The results are shown in \cref{fig:Hubbard_4} (4 sites) and \cref{fig:Hubbard_8} (8 sites). In both figures, it can be seen that the maximal evolution time of QCELS is almost two orders of magnitude smaller than that of QPE. The total cost of the two methods are comparable when $p_0=0.4$, and the total cost of QCELS is larger than that of QPE when $p_0=0.1$, mainly due to the increase of the number of repetitions $N_s$. We note that, for small $p_0$, since we first construct the eigenvalue filter $F_q$, the circuit depth of QCELS is at least $d=\left\lfloor 15/D\right\rfloor$. Thus, it is reasonable to choose $\tau_J\geq d$. This directly ensures a relatively small error ($\epsilon\leq 10^{-2}$) in our case. In other cases when the gap $D$ is large and only low accuracy is need, it may be possible to further reduce the circuit depth. 

\begin{figure}[htbp]
     \centering
     \subfloat{
         \centering
         \includegraphics[width=0.48\textwidth]{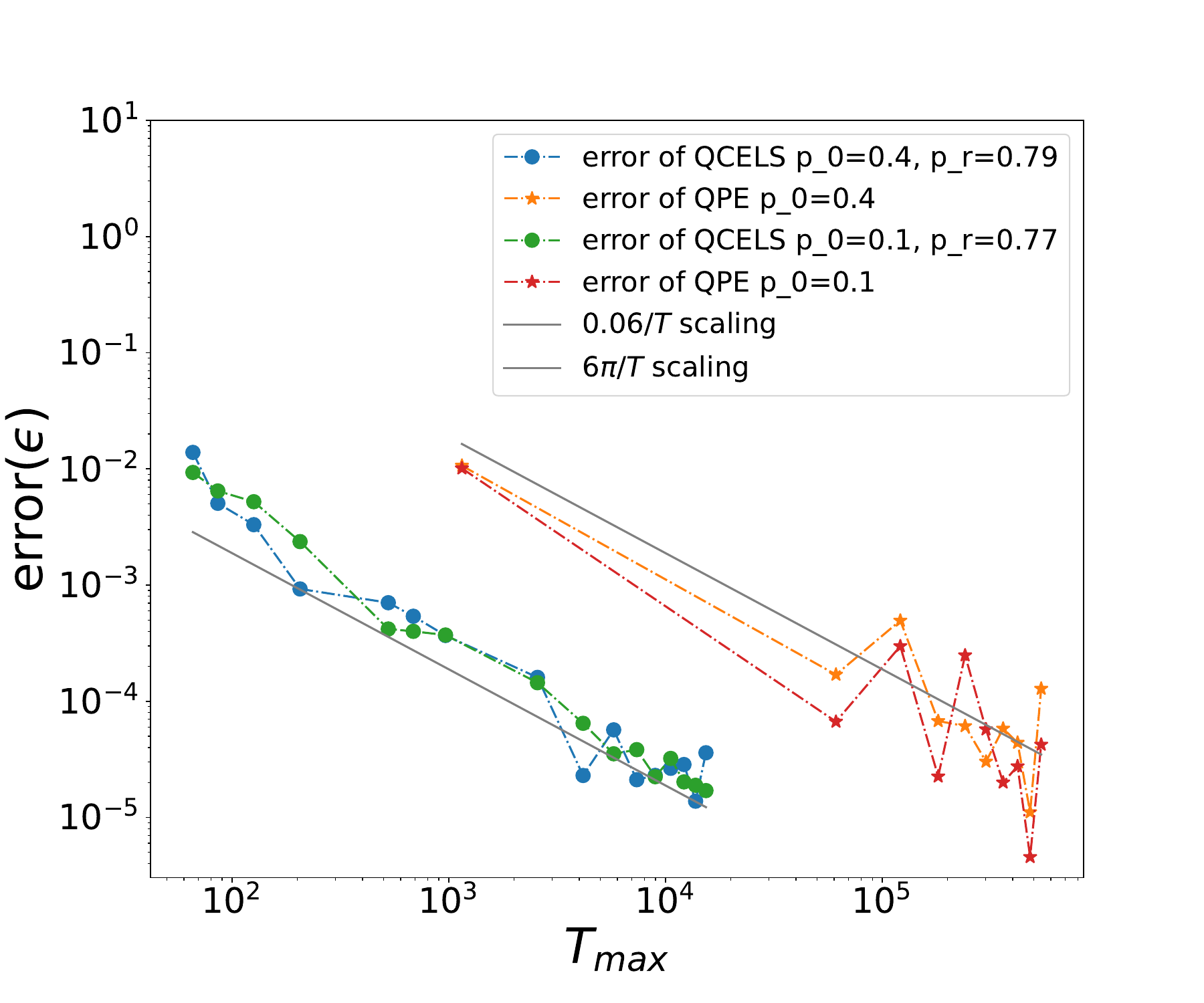}
     }
     \hfill
     \subfloat{
         \centering
         \includegraphics[width=0.48\textwidth]{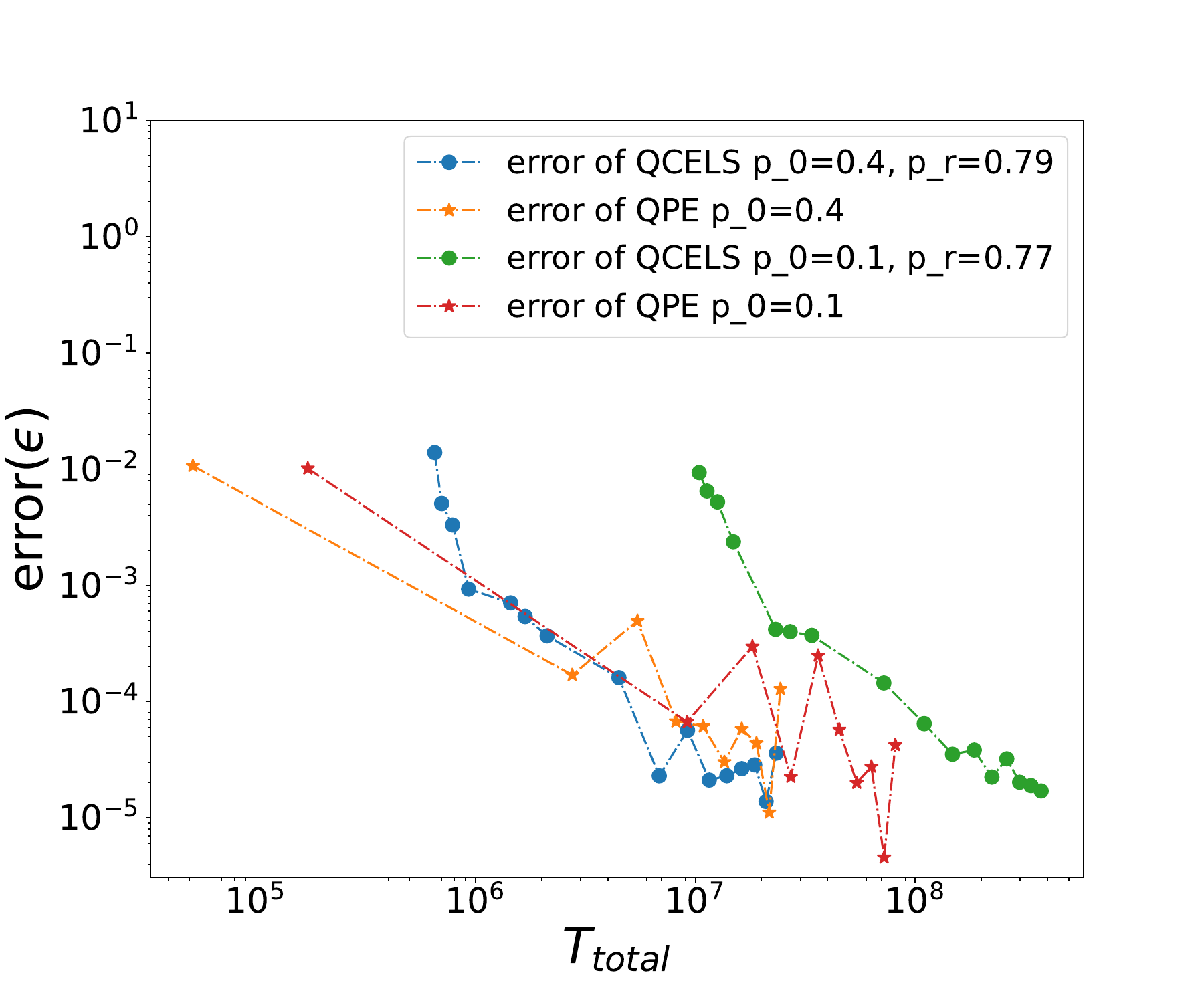}
     }
     \caption{
     \label{fig:Hubbard_4} QPE vs QCELS in Hubbard model with 4 sites. The initial overlap is small ($p_0=0.1,0.4$). Left: Depth ($T_{\max}$); Right: Cost ($T_{\mathrm{total}}$). For QCELS, we choose $N=5$ and $N_s=\left\lfloor 15p_0^{-2}\log(d)\right\rfloor$. $J,\tau_j$ are chosen according to \cref{cor:small_p0}. Compared with QPE, to achieve the same accuracy, QCELS requires a much smaller circuit depth.}
\end{figure}
\begin{figure}[htbp]
     \centering
     \subfloat{
         \centering
         \includegraphics[width=0.48\textwidth]{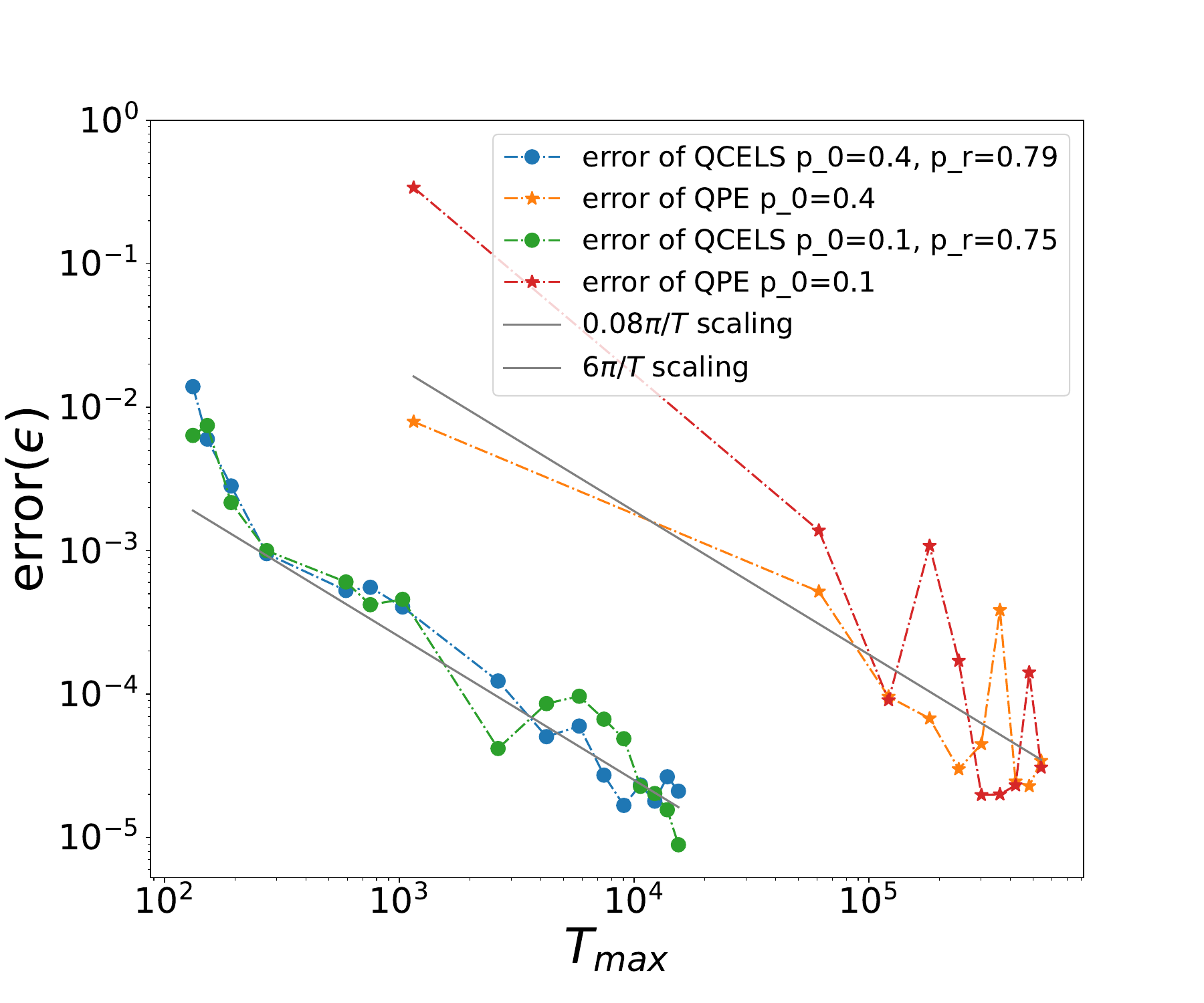}
     }
     \hfill
     \subfloat{
         \centering
         \includegraphics[width=0.48\textwidth]{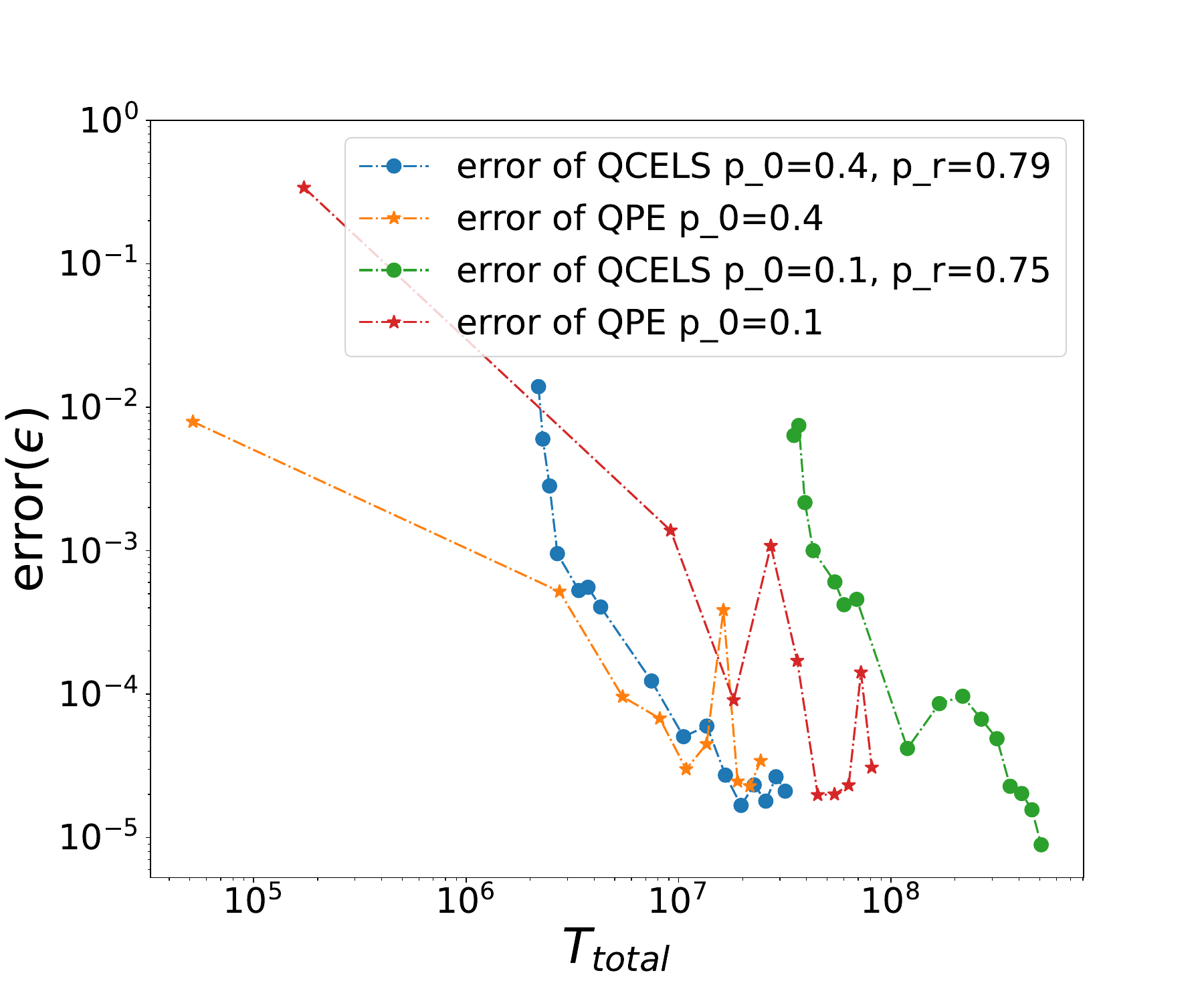}
     }
     \caption{
     \label{fig:Hubbard_8}  QPE vs QCELS in Hubbard model with 8 sites. Left: Depth ($T_{\max}$); Right: Cost ($T_{\mathrm{total}}$). For QCELS, we choose $N=5$ and $N_s=\left\lfloor 15p_0^{-2}\log(d)\right\rfloor$. $J,\tau_j$ are chosen according to \cref{cor:small_p0}. Compared with QPE, to achieve the same accuracy, QCELS has much smaller circuit depth.}
\end{figure}

\section{Discussion}\label{sec:discuss}

Due to the relatively transparent circuit structure and minimal number of required ancilla qubits, the quantum circuit in \cref{fig:qc} is suitable for early fault-tolerant quantum devices, and has received significant attention in performing a variety of tasks on quantum computers. 
Note that all algorithms in this paper (including the filtering algorithm in \cite{LinTong2022}) all use the same circuit, and the only difference is in the post-processing procedure.
This paper finds that the circuit in \cref{fig:qc} is even more powerful than previously thought for phase estimation and ground-state energy estimation, especially when the initial overlap $p_0$, or the relatively overlap $p_r$ is large. The advantage of our method can be theoretically justified when $p_0$ or $p_r$ approaches $1$. Numerical results show that even when $p_0$ or $p_r$ is away from $1$ (e.g. $0.8$), our algorithms can still outperform QPE and reduce the maximal runtime (and hence the circuit depth) by around two orders of magnitude. 

Viewed more broadly, the problem of post-processing the quantum data from the circuit in \cref{fig:qc} is a  signal processing problem using a simple (complex) exponential fitting function. Many methods have been developed in the context of classical signal processing for similar purposes (see e.g.,~\cite{BEYLKIN200517,DAMLE2013251,Hauer_1990,Bunse-Gerstner_1997}). We think that at least two features distinguish the quantum setting from the classical counterpart: (1) it is a priority to reduce the maximal runtime, and (2) each data point in the signal is inherently noisy, and the total number of measurements need to be carefully controlled. While these classical data processing methods can be applied to the phase estimation problem, we are not yet aware of analytical results demonstrating the efficiency of such methods in the quantum setting. Such connections could be an interesting direction to explore in the future. 

When the initial overlap $p_0$ is small, we have combined QCELS with the Fourier filtering algorithm in~\cite{LinTong2022} to effectively amplify this overlap as shown in \cref{alg:main2}. Another natural choice is to use the quantum eigenvalue transformation of unitary matrices (QETU)~\cite{dong2022ground}, which is a more powerful and slightly more complex circuit than that in \cref{fig:qc}, to amplify the overlap with the ground state. While we have demonstrated applications of QCELS-based algorithms to estimate ground-state energies, such algorithms may be useful in a much wider context, such as estimating excited-state energies and other observables~\cite{ZhangWangJohnson2022}. Simultaneous estimation of multiple eigenvalues using the same circuit is another interesting topic, which may open the door for developing efficient algorithms for a broader class of quantum systems with small spectral gaps.

\vspace{1em}
\textbf{Acknowledgment}

This work was supported by the NSF Quantum Leap Challenge Institute (QLCI) program under Grant number OMA-2016245, and by the Google Quantum Research Award. L.L. is a Simons Investigator. The authors thank Yulong Dong, Ethan Epperly, Daniel Stilck Franca, Peter Johnson and Yu Tong for discussions.

\bibliographystyle{abbrvnat}
\bibliography{ref}
\clearpage
\appendix

\renewcommand{\thesubsection}{\thesection.\arabic{subsection}}

\begin{center}
    {\Large \bf {Appendix}}
\end{center}

The appendix is organized as follows:

\begin{enumerate}[leftmargin=.5cm,itemindent=.5cm,labelwidth=\itemindent,labelsep=0cm,align=left]
    \item[Appendix \ref{sec:bound_En}.] We give two lemmas to bound the expectation error defined in \eqref{eqn:En}. These bounds are useful for the later proof.
    
    \item[Appendix \ref{sec:pf_prop}.] We give the proof of  \cref{thm:qcels_simple} in this section.
    
    \item[Appendix \ref{sec:pf_thm}.] We give the proof of  \cref{thm:main} in this section.
    
    \item[Appendix \ref{sec:complexity_small_p0}.] We give the proof of \cref{thm:qcels_small_p0} in this section.
    
    \item[Appendix \ref{sec:newub}-\ref{sec:lbgb}.] We prove some technical estimations that are used in our proof.
\end{enumerate}

\section{Error bound for the expectation estimation}\label{sec:bound_En}
In this section, we give a bound for $E_n$ defined in \eqref{eqn:En} with respect to $N,N_s$. Recall
\[
E_n=Z_n-\left\langle\psi\right|\exp(-in\tau H)\ket{\psi}=Z_n-\left(p_0 \exp(-i \lambda_0 n\tau )+\sum^{M-1}_{m=1}p_m\exp(-i \lambda_m n\tau)\right)
\]
and
\[
Z_n=\frac{1}{N_s}\sum^{N_s}_{k=1}\left(X_{k,n}+iY_{k,n}\right)
\]
where $X_{k,n},Y_{k,n}$ are independently generated by the quantum circuit (Fig \ref{fig:qc}) with different $W$ and satisfy \eqref{eqn:X}, \eqref{eqn:Y} respectively. We also define the average error term for each $\theta$ as
\[
\overline{E}_{\theta}=\frac{1}{N}\sum^{N-1}_{n=0}E_n\exp(i\theta n\tau)\,.
\]

In the following part of this section, we will prove the following bounds for $E_n$ and $\overline{E}_{\theta}$:
\begin{itemize}
    \item (\cref{lem:bound_of_En} equation \eqref{eqn:Enbound}): 
    
    Given $0<\eta<1/2$, when $\min\{N,N_s\}= \Omega(\log(\eta^{-1}))$,
    \begin{equation}\label{eqn:Enbound_firstappendix}
    \mathbb{P}\left(\frac{1}{N}\sum^{N-1}_{n=0}|E_n|\geq 10^{-3}\right)\leq \eta\,.
    \end{equation}
    
    \item (\cref{lem:bound_of_En_new} equation \eqref{eqn:Enbound_new}): 
    
    Given $0<\eta<1/2$ and $0<\rho,\xi<10\pi$, when $NN_s=\Omega\left(\xi^{-2}\mathrm{polylog}(\xi^{-1}\eta^{-1})\right)$,
    \begin{equation}\label{eqn:Enbound_new_firstappendix}
      \mathbb{P}\left(\sup_{\theta\in[\lambda_0-\frac{\rho}{T},\lambda_0+\frac{\rho}{T}]}\left|\overline{E}_{\theta}\right|\geq\xi\right)\leq \eta\,.
    \end{equation}
    
    \item (\cref{lem:bound_of_En_new} equation \eqref{eqn:Enbound_new2}): 
    
    Given $0<\eta<1/2$ and $0<\rho,\xi<10\pi$, when $NN_s=\Omega\left(\xi^{-2}\mathrm{polylog}(\xi^{-1}\eta^{-1})\right)$,
    \begin{equation}\label{eqn:Enbound_new2_firstappendix}
    \mathbb{P}\left(\sup_{\theta\in[\lambda_0-\frac{\rho}{T},\lambda_0+\frac{\rho}{T}]}\left|\overline{E}_{\theta}-\overline{E}_{\lambda_0}\right|\geq \rho\xi\right)\leq \eta\,.
\end{equation}
\end{itemize}

Define $E_{k,n}=X_{k,n}+iY_{k,n}-\left\langle\psi\right|\exp(-in\tau H)\ket{\psi}$. Then, we have $\mathbb{E}(E_{k,n})=0$ and $|E_{k,n}|\leq 2$. 
Using the bound and independence of $\{E_{k,n}\}_{k,n}$, we can first show the following lemma:
\begin{lem}\label{lem:bound_of_En} Given $0<\eta<1$, then
\begin{equation}\label{eqn:Enbound}
    \mathbb{P}\left(\frac{1}{N}\sum^{N-1}_{n=0}|E_n|>\frac{2}{\sqrt{N_s}}+\sqrt{\frac{2\ln(1/\eta)}{N}}\right)\leq \eta\,.
\end{equation}
\end{lem}
\begin{proof}[Proof of \cref{lem:bound_of_En}]
First, since $E_n=\frac{1}{N_s}\sum^{N_s}_{k=1}E_{k,n}$, we have 
\[
\mathbb{E}(E_{n})=0,\quad \mathbb{E}\left(|E_{n}|\right)\leq \left(\mathbb{E}\left(|E_{n}|^2\right)\right)^{1/2}\leq \frac{2}{\sqrt{N_s}},\quad |E_n|\leq 2\,.
\]
Since $\{|E_n|\}$ are bounded by $2$ and independent of each other,  according to Hoeffding's inequality, we have
\[
\mathbb{P}\left(\frac{1}{N}\sum^{N-1}_{n=0}|E_{n}|-\frac{1}{N}\sum^{N-1}_{n=0}\mathbb{E}(|E_{n}|)>\delta\right)\leq \exp\left(-\frac{N\delta^2}{2}\right)\,.
\]
Combing this inequality with $\mathbb{E}\left(|E_{n}|\right)\leq \frac{2}{\sqrt{N_s}}$ and choosing $\delta=\sqrt{\frac{2\ln(1/\eta)}{N}}$, we prove \eqref{eqn:Enbound}.
\end{proof}

We note that the magnitude bound in \cref{lem:bound_of_En} \eqref{eqn:Enbound} is a stronger than what we need in the analysis of the optimization problem \eqref{eq:loss2}. Intuitively, for fixed $\theta$, with high probability, we have $\left|\overline{E}_{\theta}\right|= \Or\left(\frac{1}{\sqrt{NN_s}}\right)$, which is a much better bound than \eqref{eqn:Enbound}. However, we can not directly use sub-Gaussian properties in this case since $\theta$ is not fixed in the optimization process. On the other hand, we expect that when $N,N_s$ are chosen properly, $\theta^*$ should belong to a tiny interval around $\lambda_0$, meaning that it's not necessary to give a uniform bound for $\left|\overline{E}_{\theta}\right|$ all $\theta$. In particular, to control the effect of error term, it suffices to bound $\left|\overline{E}_{\theta}\right|$ when $\theta$ is close to $\lambda_0$. This bound is stated in the following lemma:
\begin{lem}\label{lem:bound_of_En_new} Given $0<\eta<1/2$ and $\rho>0$, then
\begin{equation}\label{eqn:Enbound_new}
    \mathbb{P}\left(\sup_{\theta\in[\lambda_0-\frac{\rho}{T},\lambda_0+\frac{\rho}{T}]}\left|\overline{E}_{\theta}\right|\geq\left(4\sqrt{2}\log^{1/2}\left(\frac{8\sqrt{N_sN}}{\eta}\right)+\rho\right)\frac{1}{\sqrt{N_sN}}\right)\leq \eta
\end{equation}
and
\begin{equation}\label{eqn:Enbound_new2}
    \mathbb{P}\left(\sup_{\theta\in[\lambda_0-\frac{\rho}{T},\lambda_0+\frac{\rho}{T}]}\left|\overline{E}_{\theta}-\overline{E}_{\lambda_0}\right|\geq\left(4\sqrt{2}\log^{1/2}\left(\frac{8\sqrt{N_sN}}{\eta}\right)+1\right)\frac{\rho}{\sqrt{N_sN}}\right)\leq \eta
\end{equation}
\end{lem}
\begin{proof}[Proof of \cref{lem:bound_of_En_new}]
For $\theta_1,\theta_2\in[\lambda_0-\frac{\rho}{T},\lambda_0+\frac{\rho}{T}]$, we have
\[
\begin{aligned}
\left|\overline{E}_{\theta_1}-\overline{E}_{\theta_2}\right|&\leq \left|\frac{1}{N}\sum^{N-1}_{n=0}E_n\left(\exp(i\theta_1 n\tau)-\exp(i\theta_2n\tau\right)\right|\\
& \leq \frac{1}{N}\sum^{N-1}_{n=0}|E_n|\left|2\sin((\theta_1-\theta_2)\tau n/2)\right|\leq |\theta_1-\theta_2|T
\end{aligned}
\]
where we use $|E_n|\leq 2$ in the last inequality. This implies $\overline{E}_\theta$ is a $T$-Lipschitz function in $\theta$.

Now, we first consider the tail bound of $\overline{E}_\theta$ for fixed $\theta\in [\lambda_0-\frac{\rho}{T},\lambda_0+\frac{\rho}{T}]$. 
Write
\[
\overline{E}_{\theta}=\frac{1}{N_sN}\sum^{N-1}_{n=0}\sum^{N_s-1}_{m=1}a_{n,m}+ib_{n,m}\,,
\]
where 
\[
a_{n,m}=X_{m,n}\mathrm{Re}\left(\exp(i\theta n\tau)\right)-Y_{m,n}\mathrm{Im}\left(\exp(i\theta n\tau)\right)\,,
\]
and
\[
b_{n,m}=X_{m,n}\mathrm{Im}\left(\exp(i\theta n\tau)\right)+Y_{m,n}\mathrm{Re}\left(\exp(i\theta n\tau)\right)\,.
\]
It's straightforward to see that $\{a_{n,m}\}$ are independent random variables with zero expectation and $|a_{n,m}|\leq 2$. Then, according to sub-Gaussian theory, for any $\xi>0$, we have
\[
\mathbb{P}\left(\left|\frac{1}{N_sN}\sum^{N-1}_{n=0}\sum^{N_s-1}_{m=1}a_{n,m}\right|\geq\xi\right)\leq 2\exp\left(-\frac{NN_s\xi^2}{8}\right)\,.
\]
Similar bound also hold for $\{b_{n,m}\}$. Thus, we obtain that for any $\xi>0$
\[
\mathbb{P}\left(\left|\overline{E}_{\theta}\right|\geq\xi\right)\leq 4\exp\left(-\frac{NN_s\xi^2}{32}\right)\,.
\]

Given any $\epsilon>0$, we can find a set of $\lfloor\frac{2\rho}{T\epsilon}\rfloor$ points $\{\theta_i\}^{\lfloor\frac{2\rho}{T\epsilon}\rfloor}_{i=1}$ such that for any $\theta\in[\lambda_0-\frac{\rho}{T},\lambda_0+\frac{\rho}{T}]$, there exists $i$ such that $|\theta_i-\theta|\leq \epsilon$. Because $\overline{E}_\theta$ is $T$-Lipschitz, we have
\[
\mathbb{P}\left(\sup_{\theta\in[\lambda_0-\frac{\rho}{T},\lambda_0+\frac{\rho}{T}]}\left|\overline{E}_{\theta}\right|\geq\xi+T\epsilon\right)\leq \mathbb{P}\left(\sup_{1\leq i\leq\lfloor\frac{2\rho}{T\epsilon}\rfloor }\left|\overline{E}_{\theta_i}\right|\geq\xi\right)\leq \frac{8\rho}{T\epsilon}\exp\left(-\frac{NN_s\xi^2}{32}\right)
\]
Choose $\epsilon=\frac{\rho}{T\sqrt{N_sN}}$ and $\xi=4\sqrt{2}\log^{1/2}\left(\frac{8\sqrt{N_sN}}{\eta}\right)\frac{1}{\sqrt{N_sN}}$, we prove \eqref{eqn:Enbound_new}.

Next, to prove \eqref{eqn:Enbound_new2}, we first consider the tail bound of $\overline{E}_\theta-\overline{E}_{\lambda_0}$ for fixed $\theta\in [\lambda_0-\frac{\rho}{T},\lambda_0+\frac{\rho}{T}]$. 
Write
\[
\overline{E}_{\theta}-\overline{E}_{\lambda_0}=\frac{1}{N_sN}\sum^{N-1}_{n=0}\sum^{N_s-1}_{m=1}\overline{a}_{n,m}+i\overline{b}_{n,m}\,,
\]
where 
\[
\overline{a}_{n,m}=X_{m,n}\mathrm{Re}\left(\exp(i\theta n\tau)-\exp(i\lambda_0 n\tau)\right)-Y_{m,n}\mathrm{Im}\left(\exp(i\theta n\tau)-\exp(i\lambda_0 n\tau)\right)\,,
\]
and
\[
\overline{b}_{n,m}=X_{m,n}\mathrm{Im}\left(\exp(i\theta n\tau)-\exp(i\lambda_0 n\tau)\right)+Y_{m,n}\mathrm{Re}\left(\exp(i\theta n\tau)-\exp(i\lambda_0 n\tau)\right)\,.
\]
Note that $\{a_{n,m}\}$ are independent random variables with zero expectation and $|a_{n,m}|\leq \frac{2n\rho}{N}$. Then, according to Hoeffding's inequality, for any $\xi>0$, we have
\[
\mathbb{P}\left(\left|\frac{1}{N_sN}\sum^{N-1}_{n=0}\sum^{N_s-1}_{m=1}\overline{a}_{n,m}\right|\geq\xi\right)\leq 2\exp\left(-\frac{NN_s\xi^2}{8\rho^2}\right)\,.
\]
Similar to before, we finally have
\[
\mathbb{P}\left(\sup_{\theta\in[\lambda_0-\frac{\rho}{T},\lambda_0+\frac{\rho}{T}]}\left|\overline{E}_{\theta}-\overline{E}_{\lambda_0}\right|\geq\xi+T\epsilon\right)\leq \mathbb{P}\left(\sup_{1\leq i\leq\lfloor\frac{2\rho}{T\epsilon}\rfloor }\left|\overline{E}_{\theta_i}\right|\geq\xi\right)\leq \frac{8\rho}{T\epsilon}\exp\left(-\frac{NN_s\xi^2}{32\rho^2}\right)
\]
Choose $\epsilon=\frac{\rho}{T\sqrt{N_sN}}$ and $\xi=4\sqrt{2}\log^{1/2}\left(\frac{8\sqrt{N_sN}}{\eta}\right)\frac{\rho}{\sqrt{N_sN}}$, we prove \eqref{eqn:Enbound_new2}.

\end{proof}
\section{Proof of \texorpdfstring{\cref{thm:qcels_simple}}{Lg}}\label{sec:pf_prop}

For convenience of later discussion, we define a constant
\begin{equation}
\alpha=1+\left(\max_{c\in(0,\pi/2]}\frac{\sin(c)}{\pi+c}\right)\approx 1.217.
\label{eqn:alpha_const}
\end{equation}
In this section, we will prove the following theorem:
\begin{thm}[Complexity of QCELS]\label{thm:1_full}
Let $\theta^*$ be the solution of \cref{eq:loss2}, $\alpha$ be defined in \cref{eqn:alpha_const}, $T=N\tau$, and $p_0>0.71$. Given the depth constant $0<\delta\leq 4$, and the failure probability $0<\eta<1/2$. If $\delta=\Theta(\sqrt{1-p_0})$ and $N,N_s$ satisfy
\begin{equation}\label{eqn:prerq:2_thm}
NN_s= \Omega\left(\delta^{-(2+o(1))}\mathrm{polylog}\left(\log(\zeta^{-1})\eta^{-1}\right)\right),\quad \min\{N,N_s\}= \Omega(\mathrm{polylog}\left(\eta^{-1}\right))\,,\quad  T=\frac{\delta}{\epsilon}
\end{equation}
then 
\[
\mathbb{P}\left(\left|(\theta^*-\lambda_0)\ \mathrm{mod}\ [-\pi/\tau,\pi/\tau) \right|<\frac{\delta}{T}\right)\geq 1-\eta\,.
\]
\end{thm}
We note that \cref{thm:qcels_simple} is a direct corollary of \cref{thm:1_full}. The proof of \cref{thm:1_full} contains two steps. We first use the bound of expectation error in Appendix \ref{sec:bound_En} to show a proposition that induces a rough complexity estimation to an iteration in Algorithm \ref{alg:main}. Then, we use the rough complexity result as an a prior estimation and study the loss function more carefully to obtain a sharper complexity estimation.

\subsection{Rough estimate}

The rough complexity estimation is stated in the following proposition:
\begin{prop}
\label{prop:qcels_simple_full}
Let $\theta^*$ be the solution of \cref{eq:loss2}, $\alpha$ be defined in \cref{eqn:alpha_const}, $T=N\tau$, and $p_0>0.71$. Given the depth constant $0<\delta\leq 4$, and the failure probability $0<\eta<1/2$. If there exists a small enough number $\xi>0$ such that
\begin{equation}\label{eqn:prerq}
\frac{p_0}{\left(1+\alpha\right)p_0-\alpha-\xi}\leq \frac{\delta\cos(\delta/10)}{2\sin(\delta/2)},
\end{equation}
and $N,N_s$ satisfy
\begin{equation}\label{eqn:prerq:2}
NN_s= \Omega\left(\xi^{-2}\mathrm{polylog}\left(\xi^{-1}\eta^{-1}\right)\right),\quad \min\{N,N_s\}= \Omega(\mathrm{polylog}\left(\eta^{-1}\right))\,, 
\end{equation}
then 
\[
\mathbb{P}\left(\left|(\theta^*-\lambda_0)\ \mathrm{mod}\ [-\pi/\tau,\pi/\tau) \right|<\frac{\delta}{T}\right)\geq 1-\eta\,.
\]
\end{prop}

According to the proposition, we can choose $\delta$, $N$, and $N_s$ according to the inequality \eqref{eqn:prerq}. First, we use the parameter $\xi$ in \eqref{eqn:prerq} to represent an upper bound of ``$|E_n|$'' (defined in \cref{eqn:En}) and $N,N_s$ are chosen according to the results in Appendix \ref{sec:bound_En} so that $E_n$ satisfies this upper bound. Second, when $\delta$ is very small, to make sure $p_0$ satisfies \eqref{eqn:prerq}, we need $p_0=1-\mathcal{O}(\delta^2)$, which implies \eqref{eqn:delta_1} in \cref{thm:qcels_simple}. Finally, the lower bound $p_0> 0.71$ comes from \cref{eqn:R_0bound} later in the proof of \cref{prop:qcels_simple_full}. More specifically, to obtain \eqref{eqn:R_0bound}, we need $\frac{p_0}{\left(1+\alpha\right)p_0-\alpha-10^{-3}}\leq 2$, which implies $p_0>0.71$.

We first rewrite the loss function \eqref{eq:loss}. Notice that for any fixed $\theta$,
\[
\begin{aligned}
\max_{r\in\mathbb{C}}L(r,\theta)=&\frac{1}{N}\sum^{N-1}_{n=0}\left|p_0 \exp(i (\theta-\lambda_0) n\tau )+\sum^{M-1}_{m=1}p_m\exp(i (\theta-\lambda_m) n\tau )+E_n\exp(i\theta n\tau)\right|^2\\
&-\left|\frac{1}{N}\sum^{N-1}_{n=0}p_0 \exp(i (\theta-\lambda_0) n\tau )+\sum^{M-1}_{m=1}p_m\exp(i (\theta-\lambda_m) n\tau )+E_n\exp(i\theta n\tau)\right|^2\\
=&\frac{1}{N}\sum^{N-1}_{n=0}\left|p_0 \exp(i \lambda_0 n\tau )+\sum^{M-1}_{m=1}p_m\exp(i \lambda_m n\tau )+E_n\exp(i\theta n\tau)\right|^2\\
&-\left|\frac{1}{N}\sum^{N-1}_{n=0}\left[p_0 \exp(i (\theta-\lambda_0) n\tau )+\sum^{M-1}_{m=1}p_m\exp(i (\theta-\lambda_m) n\tau )+E_n\exp(i\theta n\tau)\right]\right|^2\,.
\end{aligned}
\]
This means that minimizing $L(r,\theta)$ is equivalent to maximizing the magnitude of following function:
\begin{equation}\label{eqn:f}
\begin{aligned}
f(\theta)=&\sum^{N-1}_{n=0}\left[\exp(i (\theta-\lambda_0) n\tau)+\sum^{M-1}_{m=1}\frac{p_m}{p_0}\exp(i (\theta-\lambda_m) n\tau)+\frac{E_n}{p_0}\exp(i\theta n\tau)\right]\\
=&\frac{\exp(i (\theta-\lambda_0) N\tau)-1}{\exp(i (\theta-\lambda_0)\tau)-1}+\sum^{M-1}_{m=1}\frac{p_m}{p_0} \frac{\exp(i (\theta-\lambda_m) N\tau)-1}{\exp(i (\theta-\lambda_m)\tau)-1}+\sum^{N-1}_{n=0}\frac{E_n}{p_0}\exp(i\theta n\tau)\,.
\end{aligned}
\end{equation}
Define 
\[
\overline{E}_{\theta}=\frac{1}{N}\sum^{N-1}_{n=0}E_n\exp(i\theta n\tau)\,.
\] 
Now, we are ready to prove \cref{prop:qcels_simple_full}:
\begin{proof}[Proof of \cref{prop:qcels_simple_full}] Define $R_m=\left|(\lambda_m-\theta^*)\tau\ \mathrm{mod}\ (-\pi,\pi]\right|$ for $0\leq m\leq M-1$. We separate the following proof into three steps. In the first step, we give a lower bound for ``$|f(\lambda_0)|$''. Then we give a loose upper bound for $R_0$ using the fact $|f(\theta^*)|\geq |f(\lambda_0)|$. Finally, we improve the bound to $\frac{\delta}{T}$.

\textbf{Step 1: Lower bound for ``$|f(\lambda_0)|$''.}

Using \cref{eqn:lbgb1} in Appendix \ref{sec:lbgb} , we have 
\begin{equation}\label{eqn:lowerboundgbnew2}
\begin{aligned}
        \lim_{\theta\rightarrow\lambda_0}|f(\theta)|&\geq N-\sum^{M-1}_{k=1}\frac{p_k}{p_0}\left|\frac{\exp(i(\lambda_0-\lambda_k)T)-1}{\exp(i(\lambda_0-\lambda_k)\tau)-1}\right|-\frac{\left|\overline{E}_{\lambda_0}\right|}{p_0}N\\
    &\geq \left(1-(\alpha-1)\frac{1-p_0}{p_0}-\frac{\left|\overline{E}_{\lambda_0}\right|}{p_0}\right)N\,,
\end{aligned}
\end{equation}

\textbf{Step 2: Loose upper bound for $R_0$.}

We claim that for $\alpha$ in \cref{eqn:alpha_const},
\begin{equation}\label{eqn:loosebound_D0}
R_0N\leq \frac{\pi p_0}{\left(\left(1+\alpha\right)p_0-\alpha-\left(\left|\overline{E}_{\theta^*}\right|+\left|\overline{E}_{\lambda_0}\right|\right)\right)}.
\end{equation}

If the claim does not hold, notice
\[
\left|\exp(i(\lambda_0-\theta^*)\tau)-1\right|=\left|2\sin(R_0/2)\right|\geq \frac{2}{\pi}R_0\,.
\]
Combining this with \cref{eqn:gb} in Appendix \ref{sec:newub} , 
\[
\begin{aligned}
|f(\theta^*)|&\leq  \frac{\pi}{R_0}+\sum^{M-1}_{m=1}\frac{p_m}{p_0}\frac{R_m N}{R_m}+\frac{|\overline{E}_{\theta^*}|}{p_0}N\\
&\leq \left(\frac{\pi}{R_0N}+\frac{1-p_0}{p_0}+\frac{|\overline{E}_{\theta^*}|}{p_0}\right)N\\
&< \left(1-\left(\alpha-1\right)\frac{1-p_0}{p_0}-\frac{|\overline{E}_{\lambda_0}|}{p_0}\right)N\leq \lim_{\theta\rightarrow\lambda_0}|f(\theta)|
\end{aligned}
\]
where we use $R_0N>\frac{\pi p_0}{\left(\left(1+\alpha\right)p_0-\alpha-\left(|\overline{E}_{\theta^*}|+|\overline{E}_{\lambda_0}|\right)\right)}$ in the second last inequality. This contradicts to the fact that $|f(\theta^*)|$ is the maximum. Thus, we must have \eqref{eqn:loosebound_D0}.

\textbf{Step 3: Improve upper bound to $\frac{\delta}{T}$ with probability $1-\eta$.}

Define $\beta=\frac{p_0}{\left(1+\alpha\right)p_0-\alpha-\xi}$. First, combining the second inequality of \eqref{eqn:prerq:2} with \cref{lem:bound_of_En} \eqref{eqn:Enbound} (or \eqref{eqn:Enbound_firstappendix}), we have
\[
\mathbb{P}\left(\left|\overline{E}_{\theta^*}\right|+\left|\overline{E}_{\lambda_0}\right|>10^{-3}\right)\leq \eta/2\,.
\]
When $\left|\overline{E}_{\theta^*}\right|+\left|\overline{E}_{\lambda_0}\right|\leq 10^{-3}$, according to \eqref{eqn:loosebound_D0}, we have $R_0\leq \frac{\pi p_0}{\left((1+\alpha)p_0-\alpha-10^{-3}\right)N}$. Also, plugging $\rho=\frac{\pi p_0}{\left((1+\alpha)p_0-\alpha-10^{-3}\right)}$ into \cref{lem:bound_of_En_new} \eqref{eqn:Enbound_new} (or \eqref{eqn:Enbound_new_firstappendix}) and using the first inequality of \eqref{eqn:prerq:2}, we have 
\[
\mathbb{P}\left(\sup_{\theta\in[\lambda_0-\frac{\rho}{T},\lambda_0+\frac{\rho}{T}]}\left|\overline{E}_{\theta}\right|\geq\xi/2\right)\leq \eta/2\,,
\]

Combining the above two inequalities, we have
\[
\mathbb{P}\left(|\overline{E}_{\theta^*}|+|\overline{E}_{\lambda_0}|\geq\xi\right)\leq \eta
\]
Then, to prove \eqref{eqn:distance}, it suffices to show $R_0\leq \frac{\delta}{N}$ when $|\overline{E}_{\theta^*}|+|\overline{E}_{\lambda_0}|\leq \xi$.

When $|\overline{E}_{\theta^*}|+|\overline{E}_{\lambda_0}|\leq\xi$, using \eqref{eqn:loosebound_D0}, $p_0>0.71$, and the fact that $\xi$ is small enough ($\xi<10^{-3}$), we further have
\begin{equation}\label{eqn:R_0bound}
R_0\leq \frac{\pi p_0}{\left((1+\alpha)p_0-\alpha-\xi\right)N}=\frac{\pi\beta}{N}<\frac{2\pi}{N}\,.
\end{equation}
Since $
\left|\frac{\exp(i\theta N)-1}{\exp(i\theta)-1}\right|=\left|\frac{\sin( N\theta/2)}{\sin(\theta/2)}\right|$, from \cref{eqn:f} we also have
\begin{equation}\label{eqn:ftheta_tmp1}
|f(\theta^*)|\leq \left(\frac{\sin(NR_0/2)}{N\sin(R_0/2)}+\frac{1-p_0}{p_0}+\frac{\xi}{2p_0}\right)N\,,
\end{equation}
and
\begin{equation}\label{eqn:ftheta_tmp2}
|f(\theta^*)|\geq \left(1-(\alpha-1)\frac{1-p_0}{p_0}-\frac{\xi}{2p_0}\right)N\geq \lim_{\theta\rightarrow\lambda_0}|f(\theta)|\,.
\end{equation}
Combining \cref{eqn:ftheta_tmp1,eqn:ftheta_tmp2}, we have
\begin{equation}
\frac{\sin(NR_0/2)}{\sin(R_0/2)}\ge N\frac{(1+\alpha)p_0-\alpha-\xi}{p_0}=\frac{N}{\beta}.
\end{equation}
Notice
\[
\frac{\sin(\delta/2)}{\sin(\delta/(2N))}\leq \frac{2\sin(\delta/2)N}{\cos(\delta/(2N))\delta}\leq \frac{N}{\beta}\,,
\] 
where we use $\sin(\delta/(2N))\geq \cos(\delta/(2N))\delta/(2N)$ in the first inequality, and $\beta\leq \frac{\delta\cos(\delta/10)}{2\sin(\delta/2)}$ in the second inequality. Hence
\[
\frac{\sin(NR_0/2)}{\sin(R_0/2)}\ge \frac{\sin( N\delta/(2N))}{\sin(\delta/(2N))}.
\]
Finally, because $\frac{\sin(Nx)}{\sin(x)}$ is monotonically decreasing in $ (0,\pi/N]$, we have
\begin{equation}
R_0\leq \frac{\delta}{N}.
\end{equation}
This concludes the proof.
\end{proof}

\subsection{Refined estimate}

According to \cref{prop:qcels_simple_full} \eqref{eqn:prerq} and \eqref{eqn:prerq:2}, when $\delta\rightarrow0$, we should set $\xi\sim O(\delta^{-2})$ and $NN_s\sim O(\delta^{-4})$ to ensure $T_{\max}=\frac{\delta}{\epsilon}$. Thus, we can not directly prove \cref{thm:1_full} using Proposition \ref{prop:qcels_simple_full}. We need to reduce the scaling of $N,N_s$ with respect to $\delta^{-1}$.

The main idea is to use a different way to bound the expectation error $E_{n}$. To achieve a better bound for the error term, instead of bounding $\overline{E}_{\theta^*}$ and $\overline{E}_{\lambda_0}$ separately, now we can bound the difference of these two error terms using \eqref{eqn:Enbound_new2}. Intuitively, when $\theta^*$ and $\lambda_0$ is close to each other, it is likely that these two error terms will cancel each other when we compare the difference between $|f(\theta^*)|$ and $|f(\lambda_0)|$. This intuition is justified by \cref{lem:bound_of_En_new} \eqref{eqn:Enbound_new2}.  Assume that we have already known $\left|(\theta^*-\lambda_0)\ \mathrm{mod}\ [-\pi/\tau,\pi/\tau] \right|<\frac{\delta}{T}$, then $NN_s\geq\widetilde{\mathcal{O}}\left(\delta^2\xi^{-2}\right)$ is suffices to guarantee $\left|\overline{E}_{\theta,q}-\overline{E}_{\lambda_0,q}\right|\geq \xi$ with high probability. Formally, compared this requirement with the first inequality of \eqref{eqn:prerq:2}, we can reduce the blow up rate to $\Or(\delta^{-2})$ when $\delta\rightarrow0$, which matches the condition in \cref{thm:1_full}. However, the above calculation is assuming $\left|(\theta^*-\lambda_0)\ \mathrm{mod}\ [-\pi/\tau,\pi/\tau] \right|<\frac{\delta}{T}$, which is unknown to us in prior. To overcome this difficulty, we need to use an iteration argument to achieve the desired order. We first show the following lemma to start our iteration:
\begin{lem}\label{lem:2_new}
Given $0<\delta\leq 4$ and $0<\eta<1/2$ and define $T=N\tau$. 
Assume the condition of \cref{prop:qcels_simple_full} is satisfied. If there exists $0<\xi_1<\xi$ such that
\begin{equation}\label{eqn:prerq_1.5_new}
\frac{p_0}{\left(\sqrt{2}\alpha+1\right)p_0-\sqrt{2}\alpha-\sqrt{2}(\xi_1+\xi\delta/2)}< \frac{\delta\cos(\delta/10)}{2\sin(\delta/2)}\,,
\end{equation}
and
\begin{equation}\label{eqn:prerq_3_new}
NN_s= \Omega\left(\delta^2\xi^{-2}_1\mathrm{polylog}(\xi^{-1}\eta^{-1})\right)\,,
\end{equation}
then there exists $\delta_{\mathrm{new}}\in (0,\delta)$ such that
\[
\mathbb{P}\left(\left|(\theta^*-\lambda_0)\ \text{mod}\ [-\pi/\tau,\pi/\tau] \right|<\frac{\delta_{\mathrm{new}}}{T}\right)\geq 1-\eta\,,
\]
where $\delta_{\mathrm{new}}$ is the unique solution to the following equation:
\begin{equation}\label{eqn:final_new}
\frac{p_0}{\left(\sqrt{2}\alpha+1\right)p_0-\sqrt{2}\alpha-\sqrt{2}(\xi_1+\xi\delta_{\mathrm{new}}/2)}=\frac{N\sin(\delta_{\mathrm{new}}/(2N))}{\sin(\delta_{\mathrm{new}}/2)}\,.
\end{equation}
\end{lem}
\begin{proof}[Proof of \cref{lem:2_new}] Define $R_0=\left|(\lambda_0-\theta^*)\tau\ \mathrm{mod}\ (-\pi,\pi]\right|$. Similar to the proof of \cref{prop:qcels_simple_full}, we have
\[
\mathbb{P}\left(R_0<\frac{\delta}{N}\right)\geq 1-\eta/4,\quad \mathbb{P}\left(|\overline{E}_{\theta^*}|+|\overline{E}_{\lambda_0}|\geq\xi\right)\leq \eta/4\,.
\]
Combining \cref{lem:bound_of_En_new} \eqref{eqn:Enbound_new2} (setting $\rho=\delta$) with \eqref{eqn:prerq_3_new} and $p_0>0.71$, we have
\[
    \mathbb{P}\left(\sup_{\theta\in[\lambda_0-\frac{\delta}{T},\lambda_0+\frac{\delta}{T}]}\left|\overline{E}_{\theta}-\overline{E}_{\lambda_0}\right|\geq \xi_1\right)\leq \eta/2\,.
\]
Thus, with probability $1-\eta$, we have
\begin{equation}\label{eqn:assum_new}
R_0<\frac{\delta}{N},\quad |\overline{E}_{\theta^*}|+|\overline{E}_{\lambda_0}|<\xi,\quad \left|\overline{E}_{\theta^*}-\overline{E}_{\lambda_0}\right|< \xi_1\,.
\end{equation}
From \eqref{eqn:assum_new}, it suffices to prove $\left|(\theta^*-\lambda_0)\ \text{mod}\ [-\pi/\tau,\pi/\tau] \right|<\frac{\delta_{\mathrm{new}}}{T}$.

Because $\theta^*$ is the maximal point, using \eqref{eqn:f} and the result in \cref{sec:lbgb} we have
\begin{equation}
\begin{aligned}
&\left|\sum^{N-1}_{n=0}\exp(i (\lambda_0-\theta^*) n\tau)+\sum^{M-1}_{m=1}\frac{p_m}{p_0} \frac{\exp(i (\lambda_m-\theta^*) N\tau)-1}{\exp(i (\lambda_m-\theta^*)\tau)-1}+\sum^{N-1}_{n=0}\frac{E_n}{p_0}\exp(-i\theta ^*n\tau)\right|\\
\geq &\left|N+\sum^{M-1}_{m=1}\frac{p_m}{p_0} \frac{\exp(i (\lambda_m-\lambda_0) N\tau)-1}{\exp(i (\lambda_m-\lambda_0)\tau)-1}+\sum^{N-1}_{n=0}\frac{E_n}{p_0}\exp(-i\lambda_0\tau)\right|\\
\geq &\left|N+\sum^{N-1}_{n=0}\frac{E_n}{p_0}\exp(-i\lambda_0\tau)\right|-(\alpha-1)\frac{1-p_0}{p_0}N\,.
\end{aligned}
\label{eqn:refine_temp1}
\end{equation}
Also, using \eqref{eqn:assum_new}, we have
\begin{equation}
\begin{aligned}
&\left|\sum^{N-1}_{n=0}\exp(i (\lambda_0-\theta^*)n\tau)+\sum^{M-1}_{m=1}\frac{p_m}{p_0} \frac{\exp(i (\lambda_m-\theta^*) N\tau)-1}{\exp(i (\lambda_m-\theta^*)\tau)-1}+\sum^{N-1}_{n=0}\frac{E_n}{p_0}\exp(-i\theta ^*n\tau)\right|\\
\leq &\left|\sum^{N-1}_{n=0}\exp(i (\lambda_0-\theta^*)n\tau)+\sum^{N-1}_{n=0}\frac{E_n}{p_0}\exp(-i\lambda_0n\tau)\right|+\frac{1-p_0}{p_0}N+\frac{\xi_1}{p_0} N\\
=&\left|\frac{\sin(R_0N/2)}{\sin(R_0/2)}\exp\left(iR_0(N-1)/2\right)+\sum^{N-1}_{n=0}\frac{E_n}{p_0}\exp(-i\lambda_0n\tau)\right|+\frac{1-p_0}{p_0}N+\frac{\xi_1}{p_0}N\\
\leq &\left|\frac{\sin(R_0N/2)}{\sin(R_0/2)}+\sum^{N-1}_{n=0}\frac{E_n}{p_0}\exp(-i\lambda_0n\tau)\right|+\frac{1-p_0}{p_0}N+\frac{\xi_1+\xi\delta/2}{p_0}N\\
=&\left|\frac{\sin(R_0N/2)}{N\sin(R_0/2)}+\frac{\overline{E}_{\lambda_0}}{p_0}\right|N+\frac{1-p_0}{p_0}N+\frac{\xi_1+\xi\delta/2}{p_0}N\,,
\end{aligned}
\label{eqn:refine_temp2}
\end{equation}
where we use $\left|\overline{E}_{\theta^*}-\overline{E}_{\lambda_0}\right|< p_0\xi_1$ in the second inequality, $|\exp\left(iR_0(N-1)/2\right)-1|\leq \delta/2$ and $\left|\sum^{N-1}_{n=0}E_n\exp(-i\lambda_0n\tau)\right|\leq N\xi$ in the last inequality.

Let 
\[
\sum^{N-1}_{n=0}E_n\exp(-i\lambda_0\tau)=N\left|\overline{E}_{\lambda_0}\right|\left(\cos(\theta_E)+i\sin(\theta_E)\right)\,,
\]
where $\overline{E}_{\lambda_0}=\frac{1}{N}\sum^{N-1}_{n=0}E_n\exp(-i\lambda_0 n\tau)$. Then $\left|\overline{E}_{\lambda_0}\right|\leq \xi<\frac{p_0}{\pi}<\frac{p_0}{2N}\frac{\sin(R_0N/2)}{\sin(R_0/2)}$ according to the second inequality of \eqref{eqn:assum_new}, and the estimates in \cref{eqn:refine_temp1,eqn:refine_temp2} can be combined to obtain
\begin{equation}\label{eqn:prepare_new}
\begin{aligned}
&\sqrt{\left(\frac{\sin(R_0N/2)}{N\sin(R_0/2)}+\frac{\left|\overline{E}_{\lambda_0}\right|}{p_0}\cos(\theta_E)\right)^2+\left(\frac{\left|\overline{E}_{\lambda_0}\right|}{p_0}\sin(\theta_E)\right)^2}+\frac{\alpha(1-p_0)+(\xi_1+\xi\delta/2)}{p_0}\\
\geq &\sqrt{\left(1+\frac{\left|\overline{E}_{\lambda_0}\right|}{p_0}\cos(\theta_E)\right)^2+\left(\frac{\left|\overline{E}_{\lambda_0}\right|}{p_0}\sin(\theta_E)\right)^2}\,.
\end{aligned}
\end{equation}
Because $\left(\sqrt{(x+a)^2+b^2}\right)'\geq \frac{1}{\sqrt{2}}$ when $x,a,b>0$ and $x+a>b$, we obtain
\[
\begin{aligned}
&\sqrt{\left(1+\frac{\left|\overline{E}_{\lambda_0}\right|}{p_0}\cos(\theta_E)\right)^2+\left(\frac{\left|\overline{E}_{\lambda_0}\right|}{p_0}\sin(\theta_E)\right)^2}\\
&-\sqrt{\left(\frac{\sin(R_0N/2)}{N\sin(R_0/2)}+\frac{\left|\overline{E}_{\lambda_0}\right|}{p_0}\cos(\theta_E)\right)^2+\left(\frac{\left|\overline{E}_{\lambda_0}\right|}{p_0}\sin(\theta_E)\right)^2}\\
\geq & \left(1-\frac{\sin(R_0N/2)}{N\sin(R_0/2)}\right)\frac{1}{\sqrt{2}}\,,
\end{aligned}
\]
where we set $a=\left|\overline{E}_{\lambda_0}\right|\cos(\theta_E)$ and $b=\left|\overline{E}_{\lambda_0}\right|\sin(\theta_E)$. Plugging this back into \eqref{eqn:prepare_new}, we obtain
\begin{equation}\label{eqn:start_new}
\frac{\sin(R_0N/2)}{N\sin(R_0/2)}\geq \frac{\left(\sqrt{2}\alpha+1\right)p_0-\sqrt{2}\alpha-\sqrt{2}(\xi_1+\xi\delta/2)}{p_0}\,,
\end{equation}
According to \eqref{eqn:prerq_1.5_new}, we first have
\[
\frac{\sin(\delta/2)}{N\sin(\delta/(2N))}< \frac{2\sin(\delta/2)}{\delta\cos(\delta/10)}< \frac{\left(\sqrt{2}\alpha+1\right)p_0-\sqrt{2}\alpha-\sqrt{2}(\xi_1+\xi\delta/2)}{p_0}\,.
\]
Thus, there exists $\delta_1<\delta$ such that 
\[
\frac{\sin(\delta_1/2)}{N\sin(\delta_1/(2N))}= \frac{\left(\sqrt{2}\alpha+1\right)p_0-\sqrt{2}\alpha-\sqrt{2}(\xi_1+\xi\delta/2)}{p_0}\,.
\]
Using similar argument as the proof of \cref{prop:qcels_simple_full}, we have $|R_0|\leq \frac{\delta_1}{N}$. Then, similar to the previous argument, we can improve \eqref{eqn:start_new}, meaning that $R_0$ should satisfy the following inequality
\[
\frac{\sin(R_0N/2)}{N\sin(R_0/2)}\geq \frac{\left(\sqrt{2}\alpha+1\right)p_0-\sqrt{2}\alpha-\sqrt{2}(\xi_1+\xi\delta_1/2)}{p_0}
\]
Because
\[
\begin{aligned}
&\frac{\left(\sqrt{2}\alpha+1\right)p_0-\sqrt{2}\alpha-\sqrt{2}(\xi_1+\xi\delta_1/2)}{p_0}\\
>&\frac{\left(\sqrt{2}\alpha+1\right)p_0-\sqrt{2}\alpha-\sqrt{2}(\xi_1+\xi\delta/2)}{p_0}=\frac{\sin(\delta_1/2)}{N\sin(\delta_1/(2N))}\,,
\end{aligned}
\]
there exists $\delta_2<\delta_1$ such that
\[
\frac{\sin(\delta_2/2)}{N\sin(\delta_2/(2N))}= \frac{\left(\sqrt{2}\alpha+1\right)p_0-\sqrt{2}\alpha-\sqrt{2}(\xi_1+\xi\delta_1/2)}{p_0}
\]
and $|R_0|\leq \frac{\delta_2}{N}$. Doing this recurrently, we finally have $|R_0|\leq \frac{\delta_{\mathrm{new}}}{N}$, where 
\[
\frac{\sin(\delta_{\mathrm{new}}/2)}{N\sin(\delta_{\mathrm{new}}/(2N))}= \frac{\left(\sqrt{2}\alpha+1\right)p_0-\sqrt{2}\alpha-\sqrt{2}(\xi_1+\xi\delta_{\mathrm{new}}/2)}{p_0}\,.
\]
Finally, because $h_1(x)=\frac{\sin(x/2)}{N\sin(x/(2N))}$ is a concave function ($x\in[0,\pi]$) with $h_1(0)=1,h'_1(0)=0$ and $h_2(x)=\frac{\left(\sqrt{2}\alpha+1\right)p_0-\sqrt{2}\alpha-\sqrt{2}(\xi_1+\xi x/2)}{p_0}$ is a linearly decreasing function with $0<h_2(0)<1$, $\delta_{\mathrm{new}}$ is the unique solution satisfying \eqref{eqn:final_new}.
\end{proof}

\cref{eqn:prerq_3_new} provides a direction for reducing the scaling of $N,N_s$ with respect to $\delta$. If we ignore \eqref{eqn:prerq}, to guarantee a short depth $\delta_{\mathrm{new}}/\epsilon$, we can choose $\delta\approx \delta_{\mathrm{new}}$, $\xi\sim \delta_{\mathrm{new}}$, and $\xi_1\sim \delta^2_{\mathrm{new}}$, then it suffices to choose $NN_s\sim \Theta(\delta^{-2}_{\mathrm{new}})$. This gives us the desired order that we want. However, the previous argument can not be directly applied because we also need $\xi= \Or(\delta^2)$ and $NN_s= \Omega(\xi^{-2})$ according to \eqref{eqn:prerq} and \eqref{eqn:prerq:2}. These two requirements would bring back the original $\delta^{-4}_{\mathrm{new}}$ dependence in $NN_s$. 

Even the previous argument cannot be applied directly, we can still use \cref{lem:2_new} to  improve the scaling with respect to $\delta$ in the following way: According to \cref{lem:2_new}, for fixed small $\delta_{\mathrm{new}}$, according to \eqref{eqn:prerq_1.5_new}, to make the depth smaller than $\delta_{\mathrm{new}}/\epsilon$, we should choose 
\[
\xi=\mathcal{O}\left(\min\{\delta^2,\delta_{\mathrm{new}}\}\right),\quad \xi_1=\mathcal{O}\left(\delta^2_{\mathrm{new}}\right)
\]
for all $1\leq i\leq M$. Then, according to the first inequality of \eqref{eqn:prerq:2} and \eqref{eqn:prerq_3_new}, we set
\[
NN_s=\widetilde{\Theta}\left(\max_{0\leq i\leq M-1}\left\{\xi^{-2},\delta^2\xi^{-2}_1\right\}\right)
\]
Minimizing this in $\delta,\xi,\xi_1$ with fixed $\delta_1$, a proper choice of these parameters should be 
\[
\xi=\Theta\left(\delta^2\right),\quad \xi_1=\Theta\left(\delta^2_{\mathrm{new}}\right),\quad \delta=\Theta\left(\delta^{\frac{2}{3}}_{\mathrm{new}}\right)\,.
\]
This choice of parameters would reduce the blow up rate of $NN_s$ to $\widetilde{\Theta}\left(\delta^{-\frac{8}{3}}_{\mathrm{new}}\right)$ when $\delta_{\mathrm{new}}\rightarrow0$. Using similar argument as before, we can apply \cref{lem:2_new} repeatedly. In each iteration, we can slightly reduce scaling of $NN_s$, and the final scaling can be $\Or(\delta^{-2+o(1)})$. This iteration process can be carried out using the following lemma:
\begin{lem}\label{lem:3_new}
Given $0<\delta\leq 4$, $0<\eta<1/2$ and define $T=N\tau$. 
Given an integer $M>1$, a decreasing sequence $\{\xi_i\}^M_{i=0}$ with small enough $\xi_0$, a decreasing sequence $\{\delta_i\}^M_{i=0}$. Assume the condition of \cref{prop:qcels_simple_full} is satisfied with $\xi=\xi_0$ and $\delta=\delta_0$. If 
\begin{equation}\label{eqn:prerq_225_new}
NN_s= \Omega\left(\delta^2_i\xi^{-2}_{i+1}\mathrm{polylog}(M\xi^{-1}_{i+1}\eta^{-1})\right)\,,
\end{equation}
and
\begin{equation}\label{eqn:prerq_24_new}
\frac{p_0}{\left(\sqrt{2}\alpha+1\right)p_0-\sqrt{2}\alpha-\sqrt{2}(\xi_{i+1}+\xi_0\delta_{i+1}/2)}\leq  \frac{\delta_{i+1}\cos(\delta_{i+1}/10)}{2\sin(\delta_{i+1}/2)}\,.
\end{equation}
for all $0\leq i\leq M-1$. Then, 
\begin{equation}\label{eqn:distance_new_new}
\mathbb{P}\left(\left|(\theta^*-\lambda_0)\ \bmod\ [-\pi/\tau,\pi/\tau] \right|<\frac{\delta_M}{T}\right)\geq 1-\eta\,.
\end{equation}
\end{lem}
\begin{proof}[Proof of \cref{lem:3_new}] First, according to \cref{lem:2_new} and \cref{prop:qcels_simple_full}, conditions \eqref{eqn:prerq_225_new}, \eqref{eqn:prerq_24_new} ensure
\[
\mathbb{P}\left(\left|(\theta^*-\lambda_0)\ \text{mod}\ [-\pi/\tau,\pi/\tau] \right|<\frac{\delta_1}{T}\right)\geq 1-\eta/2,\quad \mathbb{P}\left(|\overline{E}_{\theta^*}|+|\overline{E}_{\lambda_0}|\geq\xi_0\right)\leq \eta/8\,.
\]
Combining second inequality of \eqref{eqn:prerq_225_new} with \cref{lem:bound_of_En_new} \eqref{eqn:Enbound_new2}, we also have
 \[
    \mathbb{P}\left(\sup_{\theta\in[\lambda_0-\frac{\delta_1}{T},\lambda_0+\frac{\delta_1}{T}]}\left|\overline{E}_{\theta}-\overline{E}_{\lambda_0}\right|\geq \xi_2\right)\leq \frac{\eta}{4(M-1)}\,.
\]
Thus, with probability $1-\left(\
\frac{3\eta}{4}+\frac{\eta}{4(M-1)}\right)$, we have
\[
R_0<\frac{\delta_1}{N},\quad |\overline{E}_{\theta^*}|+|\overline{E}_{\lambda_0}|<\min\{\xi_0,10^{-3}\},\quad \left|\overline{E}_{\theta^*}-\overline{E}_{\lambda_0}\right|< \xi_2\,.
\]
Similar to the proof of \cref{lem:2_new}, we have
\[
\mathbb{P}\left(\left|(\theta^*-\lambda_0)\ \text{mod}\ [-\pi/\tau,\pi/\tau] \right|<\frac{\delta_{\mathrm{new}}}{T}\right)\geq 1-\left(\frac{3\eta}{4}+\frac{\eta}{4(M-1)}\right)\,.
\]
where $\delta_{\mathrm{new}}$ is the unique solution to the following equation:
\[
\frac{p_0}{\left(\sqrt{2}\alpha+1\right)p_0-\sqrt{2}\alpha-\sqrt{2}(\xi_2+\xi_0\delta_{\mathrm{new}}/2)}=\frac{N\sin(\delta_{\mathrm{new}}/(2N))}{\sin(\delta_{\mathrm{new}}/2)}\,.
\]
Because 
\[
\frac{p_0}{\left(\sqrt{2}\alpha+1\right)p_0-\sqrt{2}\alpha-\sqrt{2}(\xi_2+\xi_0\delta_2/2)}= \frac{\delta_2\cos(\delta_2/10)}{2\sin(\delta_2/2)}<\frac{N\sin(\delta_{2}/(2N))}{\sin(\delta_{2}/2)}\,,
\]
we must have $\delta_2>\delta_{\mathrm{new}}$ and
\[
\mathbb{P}\left(\left|(\theta^*-\lambda_0)\ \text{mod}\ [-\pi/\tau,\pi/\tau] \right|<\frac{\delta_2}{T}\right)\geq 1-\left(\frac{3\eta}{4}+\frac{\eta}{4(M-1)}\right)\,.
\]
Combining this with second inequality of \eqref{eqn:prerq_225_new} ($i=2$), with probability $1-\left(\frac{3\eta}{4}+\frac{\eta}{2(M-1)}\right)$, we have
\[
R_0<\frac{\delta_2}{N},\quad |\overline{E}_{\theta^*}|+|\overline{E}_{\lambda_0}|<\xi_0,\quad \left|\overline{E}_{\theta^*}-\overline{E}_{\lambda_0}\right|< \xi_3\,,
\]
which implies
\[
\mathbb{P}\left(\left|(\theta^*-\lambda_0)\ \text{mod}\ [-\pi/\tau,\pi/\tau] \right|<\frac{\delta_3}{T}\right)\geq 1-\left(\frac{3\eta}{4}+\frac{\eta}{2(M-1)}\right)\,.
\]
Doing this repeatedly, we finally obtain \eqref{eqn:distance_new_new}.
\end{proof}

Now, we are ready to prove \cref{thm:1_full}:
\begin{proof}[Proof of \cref{thm:1_full}]
For fixed  decreasing sequence $\{\delta_i\}^M_{i=0}$ with small enough $\delta_0$, according to \eqref{eqn:prerq} and \eqref{eqn:prerq_24_new}, to make the depth smaller than $\delta_M/\epsilon$, we need
\begin{equation}\label{eqn:requirement_p_xi}
(1+\sqrt{2}\alpha)(1-p_0)=\mathcal{O}\left(\min_{i}\left\{\delta^2_i\right\}\right),\quad \xi_0=\mathcal{O}\left(\min\{\delta^2_0,\delta_M\}\right),\quad \xi_i=\mathcal{O}\left(\delta^2_i\right)
\end{equation}
for all $1\leq i\leq M$. The first equation implies that the smallest $\delta_M$ that we can choose is
\[
\delta_M=\Theta(\sqrt{1-p_0})\,.
\]
According to \eqref{eqn:prerq:2} and \eqref{eqn:prerq_225_new}, the last two equations of \eqref{eqn:requirement_p_xi} imply that
\begin{equation}\label{eqn:RQforN_old}
    NN_s=\widetilde{\Omega}\left(\max_{0\leq i\leq M-1}\left\{p^{-2}_0\xi^{-2},p^{-2}_0\delta^2_i\xi^{-2}_{i+1}\right\}\right)
\end{equation}
Minimizing \eqref{eqn:RQforN_old} in $\{\delta_i\}^{M-1}_{i=0},\{\xi_i\}^M_{i=0}$ with fixed $\delta_M$, a proper choice of these parameters should be 
\[
\xi_i=\Theta\left(\delta^2_i\right),\quad \delta_i=\Theta\left(\delta^{\frac{2-(1/2)^i}{2-(1/2)^M}}_M\right)
\]
for $0\leq i\leq M$. Thus, to ensure
\[
\mathbb{P}\left(\left|(\theta^*-\lambda_0)\ \mathrm{mod}\ [-\pi/\tau,\pi/\tau] \right|<\frac{\delta_M}{T}\right)\geq 1-\eta\,,
\]
we can choose $NN_s=\Omega\left(\delta^{-\frac{4}{2-(1/2)^M}}_M\mathrm{polylog}(M\delta^{-1}_M\eta^{-1}p^{-1}_0\Delta^{-1})\right)$. Set $\zeta=\frac{1}{2^M}\frac{1}{2-(1/2)^M}$, we conclude the proof.

\end{proof}

\section{Proof of \texorpdfstring{\cref{thm:main}}{Lg}}\label{sec:pf_thm}
In this section, we prove a rigorous version of \cref{thm:main} as follows:
\begin{thm}[Complexity of multi-level QCELS]
\label{thm:main_full}
Let $\theta^*$ be the output of Algorithm \ref{alg:main}.
Given $p_0>0.71$, $0<\eta<1/2$, $0<\epsilon<1/2$, and small $\zeta>0$. We can choose $\delta$ according to \cref{eqn:delta_1}. Let
\[
J=\left\lceil\log_2(1/\epsilon)\right\rceil+1,\quad \tau_j=2^{j-1-\left\lceil\log_2(1/\epsilon)\right\rceil}\frac{\delta}{N\epsilon},\quad\forall 1\leq j\leq J\,,
\]
and
\[
NN_s=\Theta\left(\delta^{-(2+\zeta)}\mathrm{polylog}\left(\log(\zeta^{-1})\log(\epsilon^{-1})\eta^{-1}\right)\right),\quad \min\{N,N_s\}= \Omega(\mathrm{polylog}\left(\log(\epsilon^{-1})\eta^{-1}\right))\,.
\]
Denote $\theta^*$ as the output of Algorithm \ref{alg:main}, then
\[
\mathbb{P}\left(\left|(\theta^*-\lambda_0)\right|<\epsilon\right)\geq 1-\eta\,.
\]
In particular, we have
\[
T_{\max}=N\tau_J =\frac{\delta}{\epsilon},\quad T_{\mathrm{total}}=\sum^J_{j=1}N(N-1)N_s\tau_j/2=\Theta\left(\frac{\mathrm{polylog}\left(\log(\zeta^{-1})\log(\epsilon^{-1})\eta^{-1}\right)}{\delta^{1+\zeta}\epsilon}\right)\,.
\]
\end{thm}
\begin{proof}[Proof of \cref{thm:main_full}]
To prove \cref{thm:main}, it suffices to prove that for each fixed $j$, 
\begin{equation}\label{eq:induction}
\mathbb{P}\left(\lambda_0\in \left[\theta^*_j-\frac{\delta}{N\tau_j},\theta^*_j+\frac{\delta}{N\tau_j}\right]\right)\leq 1-\frac{j\eta}{J}\,.
\end{equation}
We prove this by induction. First, let $j=1$, according to \cref{thm:1_full},
\[
\mathbb{P}\left(\left|(\theta^*_1-\lambda_0)\ \mathrm{mod}\ [-\pi/\tau_1,\pi/\tau_1] \right|<\frac{\delta}{T}\right)\geq 1-\frac{\eta}{J}\,,
\]
Because $\tau_1=2^{-\left\lceil\log_2(1/\epsilon)\right\rceil}\frac{\delta}{N\epsilon}\leq \frac{\delta}{N}<\frac{1}{4}$, we have $\pi/\tau_1>\pi$. Then,
\[
\mathbb{P}\left(\left|\theta^*_1-\lambda_0\right|<\frac{\delta}{T}\right)\geq 1-\frac{\eta}{J}\,,
\]
Assume \eqref{eq:induction} is true for $j=K-1$, meaning
\[
\mathbb{P}\left(\lambda_0\in \left[\theta^*_{K-1}-\frac{\delta}{N\tau_{K-1}},\theta^*_{K-1}+\frac{\delta}{N\tau_{K-1}}\right]\right)\geq 1-\frac{(K-1)\eta}{J}\,.
\]
Using \cref{thm:1_full} again, we have
\begin{equation}\label{eqn:several_interval}
\mathbb{P}\left(\lambda_0\in \left[\theta^*_{K-1}-\frac{\delta}{N\tau_{K-1}},\theta^*_{K-1}+\frac{\delta}{N\tau_{K-1}}\right]\bigcap C_k\right)\geq 1-\frac{K\eta}{J}
\end{equation}
where $C_k=\bigcup_{c\in\mathbb{Z}}\left[\theta^*_{K}-\frac{\delta}{N\tau_{K}}+\frac{2c\pi}{\tau_K},\theta^*_{K}+\frac{\delta}{N\tau_{K}}+\frac{2c\pi}{\tau_K}\right].$

Noticing $\tau_K=2\tau_{K-1}$, $\frac{\delta}{N}<\frac{\pi}{4}$, we have 
\[
\frac{2\pi}{\tau_K}-\frac{4\delta}{N\tau_K}>\frac{\pi}{\tau_{K}}\,.
\]
Since $\theta^*_K\in \left[\theta^*_{K-1}-\frac{\pi}{\tau_{K}},\theta^*_{K-1}+\frac{\pi}{\tau_{K}}\right]$, we obtain
\[
\begin{aligned}
&\left[\theta^*_{K-1}-\frac{\delta}{N\tau_{K-1}},\theta^*_{K-1}+\frac{\delta}{N\tau_{K-1}}\right]\bigcap C_k\\
=&\left[\theta^*_{K-1}-\frac{\delta}{N\tau_{K-1}},\theta^*_{K-1}+\frac{\delta}{N\tau_{K-1}}\right]\bigcap \left[\theta^*_{K}-\frac{\delta}{N\tau_{K}},\theta^*_{K}+\frac{\delta}{N\tau_{K}}\right]\,.
\end{aligned}
\]
Combining this with \eqref{eqn:several_interval}, we obtain
\[
\mathbb{P}\left(\lambda_0\in \left[\theta^*_{K}-\frac{\delta}{N\tau_{K}},\theta^*_{K}+\frac{\delta}{N\tau_{K}}\right]\right)\geq 1-\frac{K\eta}{J}\,,
\]
which concludes the proof.
\end{proof}
\section{Proof of  \texorpdfstring{\cref{thm:qcels_small_p0}}{Lg}}\label{sec:complexity_small_p0}

In this section, we prove the following rigorous version of \cref{thm:qcels_small_p0}:
\begin{thm}[Complexity of Algorithm \ref{alg:main2}]\label{prop:small_p0_rigorous} Given small $\zeta>0$ and the failure probability $0<\eta<1$. Assume $\delta= \Theta\left(\sqrt{1-p_r(I,I')}\right)$ and $p_r(I,I')$ is close enough to $1$. Let $d=\Theta(D^{-1}\mathrm{polylog}(p^{-1}_0\delta^{-1}))$, $q=\Theta(p_0\delta^2)$, 
\[
J=\left\lceil\log_2(1/\epsilon)\right\rceil+1,\quad \tau_j=2^{j-1-\left\lceil\log_2(1/\epsilon)\right\rceil}\frac{\delta}{N\epsilon},\quad\forall 1\leq j\leq J\,,
\]
and
\[
\begin{aligned}
&NN_s= \Omega\left(p^{-2}_0\delta^{-(2+\zeta)}\mathrm{polylog}(\log(\zeta^{-1})\log(\epsilon^{-1})D^{-1}\eta^{-1}p^{-1}_0)\right)\,,\\
&\min\{N,N_s\}= \Omega(p^{-2}_0\mathrm{polylog}\left(\log(\epsilon^{-1})D^{-1}\eta^{-1}\right))\,,
\end{aligned}
\]
then
\[
\mathbb{P}\left(\left|(\theta^*-\lambda_0)\ \mathrm{mod}\ [-\pi/\tau,\pi/\tau] \right|<\frac{\delta}{T}\right)\geq 1-\eta\,,
\]
where $\theta^*$ is the output of \cref{alg:main2}. In particular, to construct the loss function,  
\[
T_{\max}=d+\delta/\epsilon,\quad T_{\mathrm{total}}=\Theta\left(p^{-2}_0\delta^{-(2+\zeta)}\mathrm{polylog}(\log(\zeta^{-1})\log(\epsilon^{-1})D^{-1}\eta^{-1}p^{-1}_0)\left(d+\delta/\epsilon\right)\right)\,.
\]
\end{thm}
Define 
$\mathcal{F}=\sum^{d}_{l=-d} \left|\hat{F}_{l,q}\right|$. Since $F_q$ in \cref{alg:main2} is chosen according to \cite[Lemma 6]{LinTong2022}, we have $\mathcal{F}=\Theta(\log(D^{-1}\mathrm{polylog}(q^{-1})))$. Then, similar to the proof of \cref{thm:main_full} (in Appendix \ref{sec:pf_thm}), to prove \cref{prop:small_p0_rigorous}, it suffices to show the following theorem which gives us the complexity of one step of the iteration with general choice of $I,I'$, and $F_q$:
\begin{thm}\label{thm:qcels_small_p0_full} Given small $\zeta>0$ and the failure probability $0<\eta<1$. Assume  $\delta=\Theta\left(\sqrt{1-p_r(I,I')}\right)$ and $p_r(I,I')$ is close enough to $1$. Set $q=\Theta(p_0\delta^2)$, \[NN_s= \Omega\left(p^{-2}_0\delta^{-(2+\zeta)}\mathcal{F}^2
\mathrm{polylog}(\log(\zeta^{-1})\eta^{-1}p^{-1}_0\mathcal{F})\right),\quad\min\{N,N_s\}= \Omega(p^{-2}_0\mathcal{F}^2\mathrm{polylog}\left(\eta^{-1}\mathcal{F}\right))\,.\] Then, 
\begin{equation}\label{eqn:distance_small_p_0}
\mathbb{P}\left(\left|(\theta^*-\lambda_0)\ \mathrm{mod}\ [-\pi/\tau,\pi/\tau] \right|<\frac{\delta}{T}\right)\geq 1-\eta\,,
\end{equation}
where $\theta^*$ is defined in \eqref{eq:loss_small_p0} (with $\tau_j=T$).
\end{thm}

Define
\[
\overline{E}_{\theta}=\frac{1}{N}\sum^{N-1}_{n=0}E_{n,q}\exp(i\theta n\tau)\,.
\]
where $E_{n,q}$ is defined in \eqref{eqn:Enq}. Similar to the large $p_0$ setting, we first give a bound for the expectation error $E_{n,q}$ in the following lemma:
\begin{lem}\label{lem:bound_of_En_q} Assume $q<1$. Given $0<\eta<1/2$ and $0<\rho,\xi<10\pi$, then
\begin{itemize}
    \item When $\min\{N,N_s\}= \Omega(p^{-2}_0\mathcal{F}^2\log(\eta^{-1}\mathcal{F}))$,
    \begin{equation}\label{eqn:Enbound_q}
    \mathbb{P}\left(\frac{1}{N}\sum^{N-1}_{n=0}|E_{n,q}|>
    0.001p_0\right)\leq \eta\,.
\end{equation}

\item When $NN_s=\Omega\left(p^{-2}_0\xi^{-2}\mathcal{F}^2\mathrm{polylog}(\xi^{-1}\eta^{-1}p^{-1}_0)\mathcal{F}\right)$,
\begin{equation}\label{eqn:Enbound_new_q}
    \mathbb{P}\left(\sup_{\theta\in[\lambda_0-\frac{\rho}{T},\lambda_0+\frac{\rho}{T}]}\left|\overline{E}_{\theta,q}\right|\geq p_0\xi\right)\leq \eta\,.
\end{equation}

\item When $NN_s=\Omega\left(\rho^2 p^{-2}_0\xi^{-2}\mathcal{F}^2\mathrm{polylog}(\xi^{-1}\eta^{-1}p^{-1}_0\mathcal{F})\right)$,
\begin{equation}\label{eqn:Enbound_new2_q}
    \mathbb{P}\left(\sup_{\theta\in[\lambda_0-\frac{\rho}{T},\lambda_0+\frac{\rho}{T}]}\left|\overline{E}_{\theta,q}-\overline{E}_{\lambda_0,q}\right|\geq p_0\xi\right)\leq \eta\,.
\end{equation}

\end{itemize}
\end{lem}
\begin{proof}[Proof of \cref{lem:bound_of_En_q}] The proof is very similar to that of Lemmas \ref{lem:bound_of_En} and \ref{lem:bound_of_En_new} after noticing $E_{n,q}=\frac{1}{N_s}\sum^{N_s}_{k=1}Z_{k,n,q}-\mathbb{E}(Z_{k,n,q})$ and $|Z_{k,n,q}|\leq \mathcal{F}$. 
\end{proof}

Before proving \cref{thm:qcels_small_p0_full}, we first prove a result that is similar to \cref{prop:qcels_simple_full}.
\begin{lem}\label{lem:small_p0}
Define $T=N\tau$ and assume $p_r(I,I')$ is close enough to $1$. Set $q=\Theta(p_0\delta^2)$. Given $\Theta\left(\sqrt{1-p_r(I,I')}\right)\leq \delta\leq 4$, $0<\eta<1/2$. If there exists a small enough number $\xi>0$ such that
\begin{equation}\label{eqn:pre_small_p0_0}
\frac{1}{1-\xi-\delta^2/200}\leq \frac{\delta\cos(\delta/10)}{2\sin(\delta/2)}\,, 
\end{equation}
and
\begin{equation}\label{eqn:pre_small_p0}
NN_s= \Omega\left(p^{-2}_0\xi^{-2}\mathcal{F}^2\mathrm{polylog}\left(\xi^{-1}\eta^{-1}p^{-1}_0\mathcal{F}\right)\right),\quad \min\{N,N_s\}= \Omega\left(p^{-2}_0\mathcal{F}^2\mathrm{polylog}\left(\eta^{-1}\mathcal{F}\right)\right)\,, 
\end{equation}
then
\begin{equation}\label{eqn:distance_small_p0_lemma1}
\mathbb{P}\left(\left|(\theta^*-\lambda_0)\ \mathrm{mod}\ [-\pi/\tau,\pi/\tau] \right|<\frac{\delta}{T}\right)\geq 1-\eta\,,
\end{equation}
where $\theta^*$ is defined in \eqref{eq:loss_small_p0}.
\end{lem}
\begin{proof}[Proof of \cref{lem:small_p0}.]
$R_0=\left|(\lambda_0-\theta^*)\tau\ \mathrm{mod}\ (-\pi,\pi]\right|$. Recall
\[
G_{n,q}=Z_{n,q}-p_0\exp(i \lambda_0 n\tau)=E_{n,q}+p_0(F_q(\lambda_0)-1)\exp(-i\lambda_0n\tau)+\sum^{M-1}_{k=1}p_kF_q(\lambda_k)\exp(-i\lambda_kn\tau)\,.
\]
and
\begin{equation}\label{eqn:inequality_G}
|G_{n,q}|\leq |E_{n,q}|+q+\frac{\left(1-p_r(I,I')\right)p_0}{p_r(I,I')}\,.
\end{equation}
Notice
\[
\begin{aligned}
    &\min_{r\in\mathbb{C},\theta\in\mathbb{R}}\frac{1}{N}\sum^{N-1}_{n=0}\left|Z_{n,q}-r\exp(-i\theta n\tau)\right|^2\\
    =&\min_{r\in\mathbb{C},\theta\in\mathbb{R}}\frac{1}{N}\sum^{N-1}_{n=0}\left|p_0\exp(-i\lambda_0n\tau)+G_{n,q}-r\exp(-i\theta n\tau)\right|^2\\
    =&\min_{r\in\mathbb{C},\theta\in\mathbb{R}}\frac{1}{N}\sum^{N-1}_{n=0}\left|\exp(-i\lambda_0n\tau)+\frac{G_{n,q}}{p_0}-\frac{r}{p_0}\exp(-i\theta n\tau)\right|^2
\end{aligned}
\]
Similar to the proof of \cref{prop:qcels_simple_full}, this minimization problem is equivalent to maximizing the magnitude of the following function in $\theta$:
\begin{equation}\label{eqn:f2}
\begin{aligned}
f(\theta)=\frac{\exp(i (\theta-\lambda_0) N\tau)-1}{\exp(i (\theta-\lambda_0)\tau)-1}+\sum^{N-1}_{n=1}\frac{G_{n,q}}{p_0}\exp(i\theta n\tau)\,.
\end{aligned}
\end{equation}
First, using \cref{eqn:Enbound_q,eqn:pre_small_p0}, 
\[
\mathbb{P}\left(\frac{|E_{n,q}|}{p_0}> 5*10^{-4}\right)\leq \eta/2.
\]
Combining this with \eqref{eqn:inequality_G}, $q=\Theta(p_0\delta^2)$ and the assumption that $1-p_r(I,I')$ is small enough, we can have 
\[
\mathbb{P}\left(\frac{|G_{n,q}|}{p_0}> 10^{-3}\right)\leq \eta/2\,.
\]

When $\frac{|G_{n,q}|}{p_0}\leq 10^{-3}$, similar to the proof of \cref{prop:qcels_simple_full}, we have $R_0\leq \frac{\pi}{0.9N}$. This implies that, with probability $1-\eta/2$,
\[
\frac{|G_{n,q}|}{p_0}\leq 10^{-3},\quad R_0\leq \frac{\pi}{0.9N}\,.
\]
The second inequality gives us a loose bound for $R_0$. Combining these two inequalities with \eqref{eqn:Enbound_new_q} ($\rho=\pi/0.9$), the first inequality of \eqref{eqn:pre_small_p0}, and $\Theta\left(\sqrt{1-p_r(I,I')}\right)\leq \delta\leq 4$, we further have
\[
    \mathbb{P}\left(\left|\frac{1}{N}\sum^{N-1}_{n=1}\frac{G_{n,q}}{p_0}\exp(i\theta^* n\tau)\right|+\left|\frac{1}{N}\sum^{N-1}_{n=1}\frac{G_{n,q}}{p_0}\exp(i\lambda_0 n\tau)\right|\geq \delta^2/200+\xi\right)\leq \eta\,.
\]

Finally, similar to the proof of \cref{prop:qcels_simple_full}, because
\[
\frac{\sin(\delta/2)}{\sin(\delta/(2N)}\leq \frac{2\sin(\delta/2)N}{\delta\cos(\delta/10)}\leq N(1-\xi-\delta^2/200)\,,
\]
we have $\mathbb{P}\left(R_0\leq\frac{\delta}{N}\right)\geq 1-\eta$, which implies \eqref{eqn:distance_small_p0_lemma1}.
\end{proof}

Similar to the discussion in Appendix \ref{sec:pf_prop}, we can not directly prove \cref{thm:qcels_small_p0_full} using  \cref{lem:small_p0}. We need to to use a different way to bound the expectation error $E_{n,q}$ to reduce the scaling of $N,N_s$ with respect to $\delta$. First, we show the following lemma similar to \cref{lem:2_new}: 

\begin{lem}\label{lem:2}
Define $T=N\tau$ and assume $p_r(I,I')$ is close enough to $1$. Given $\Theta\left(\sqrt{1-p_r(I,I')}\right)\leq \delta\leq 4$, $0<\eta<1/2$, and small enough $\xi>0$. Assume \cref{lem:small_p0} equations \eqref{eqn:pre_small_p0_0}, \eqref{eqn:pre_small_p0} hold and there exists $0<\xi_1<\xi$ such that
\begin{equation}\label{eqn:prerq_new1.5}
\frac{1}{1-\sqrt{2}(\xi_1+\xi\delta/2+\delta^2/200)}< \frac{\delta\cos(\delta/10)}{2\sin(\delta/2)}\,.
\end{equation}
Let $\Theta\left(\sqrt{1-p_r(I,I')}\right)\leq  \delta_{\mathrm{new}}<\delta$  satisfy
\[
\frac{1}{1-\sqrt{2}(\xi_1+\xi\delta_{\mathrm{new}}/2+\delta^2_{\mathrm{new}}/200)}\leq \frac{\delta_{\mathrm{new}}\cos(\delta_{\mathrm{new}}/10)}{2\sin(\delta_{\mathrm{new}}/2)}\,.
\]
If $q=\Theta(p_0\delta^2_{\mathrm{new}})$ and
\begin{equation}\label{eqn:prerq_new3}
NN_s= \Omega\left(p^{-2}_0\delta^2\xi^{-2}_1\mathcal{F}^2\mathrm{polylog}\left(\xi^{-1}_{1}\eta^{-1}p^{-1}_0\mathcal{F}\right)\right)\,.
\end{equation}
Then
\[
\mathbb{P}\left(\left|(\theta^*-\lambda_0)\ \mathrm{mod}\ [-\pi/\tau,\pi/\tau] \right|<\frac{\delta_{\mathrm{new}}}{T}\right)\geq 1-\eta\,,
\]
\end{lem}
\begin{proof}[Proof of \cref{lem:2}] Define $R_0=\left|(\lambda_0-\theta^*)\tau\ \mathrm{mod}\ (-\pi,\pi]\right|$. According to Lemmas \ref{lem:bound_of_En_q} and \ref{lem:small_p0}, we have $\mathbb{P}\left(R_0<\frac{\delta}{N}\right)\geq 1-\eta/4$ and 
\[
\mathbb{P}\left(\left|\overline{E}_{\theta^*}\right|+\left|\overline{E}_{\lambda_0}\right|\geq p_0\xi\right)\leq \eta/4\,.
\]
Combining \cref{lem:bound_of_En_q} \eqref{eqn:Enbound_new2_q} with \eqref{eqn:prerq_new3}, we have
\[
    \mathbb{P}\left(\sup_{\theta\in[\lambda_0-\frac{\delta}{T},\lambda_0+\frac{\delta}{T}]}\left|\overline{E}_{\theta}-\overline{E}_{\lambda_0}\right|\geq p_0\xi_1\right)\leq \eta/2\,.
\]
Thus, with probability $1-\eta$, we have
\begin{equation}\label{eqn:assumption}
R_0<\frac{\delta}{N},\quad \left|\overline{E}_{\theta^*}-\overline{E}_{\lambda_0}\right|< p_0\xi_1\,,
\end{equation}
and
\begin{equation}\label{eqn:assumption_1}
\left|\overline{E}_{\theta}\right|+\left|\overline{E}_{\lambda_0}\right|\leq p_0\xi\,.
\end{equation}
It suffices to prove $\left|(\theta^*-\lambda_0)\ \mathrm{mod}\ [-\pi/\tau,\pi/\tau] \right|<\frac{\delta^\star}{T}$ assuming \eqref{eqn:assumption} and \eqref{eqn:assumption_1}.

Because $\theta^*$ is the maximal point of \eqref{eqn:f2}, we have
\[
\begin{aligned}
&\left|\sum^{M-1}_{m=0}\frac{p_mF_q(\lambda_m)}{p_0} \frac{\exp(i (\theta^*-\lambda_m) N\tau)-1}{\exp(i (\theta^*-\lambda_m)\tau)-1}+\frac{N\overline{E}_{\theta^*}}{p_0}\right|\\
\geq &\left|N+\sum^{N-1}_{n=0}\frac{E_n}{p_0}\exp(i\lambda_0n\tau)\right|-\frac{\delta^2N}{400}\,.
\end{aligned}
\]
where we use $q\leq p_0\delta^2/1600$ and  $\frac{1-p_r(I,I')}{p_r(I,I')}\leq \delta^2/800$ in the last inequality.

Also, using \eqref{eqn:assumption}, we have
\[
\begin{aligned}
&\left|\sum^{M-1}_{m=0}\frac{p_mF_q(\lambda_m)}{p_0} \frac{\exp(i (\theta^*-\lambda_m) N\tau)-1}{\exp(i (\theta^*-\lambda_m)\tau)-1}+\frac{N\overline{E}_{\theta^*}}{p_0}\right|\\
\leq &\left| \sum^{N-1}_{n=0}\exp(i (\theta^*-\lambda_0)n\tau)+\sum^{N-1}_{n=0}\frac{E_n}{p_0}\exp(i\lambda_0n\tau)\right|+\frac{\delta^2N}{400}+\xi_1N\\
=&\left| \frac{\sin(R_0N/2)}{\sin(R_0/2)}\exp\left(iR_0(N-1)/2\right)+\sum^{N-1}_{n=0}\frac{E_n}{p_0}\exp(i\lambda_0n\tau)\right|+\frac{\delta^2N}{400}+\xi_1N\\
\leq &\left| \frac{\sin(R_0N/2)}{\sin(R_0/2)}+\sum^{N-1}_{n=0}\frac{E_n}{p_0}\exp(i\lambda_0n\tau)\right|+(\xi_1+\xi\delta/2+\delta^2/400)N\,,
\end{aligned}
\]
where we use $\left|\overline{E}_{\theta^*}-\overline{E}_{\lambda_0}\right|< p_0\xi_1$ in the first inequality, $|\exp\left(iR_0(N-1)/2\right)-1|\leq \delta/2$ and $\left|\sum^{N-1}_{n=0}\frac{E_n}{p_0}\exp(i\lambda_0n\tau)\right|\leq N\xi$ in the second inequality.

Assume 
\[
\sum^{N-1}_{n=0}E_n\exp(i\lambda_0\tau)=N\left|\overline{E}_{\lambda_0}\right|\left(\cos(\theta_E)+i\sin(\theta_E)\right)\,,
\]
where $\overline{E}_{\lambda_0}=\frac{1}{N}\sum^{N-1}_{n=0}E_n\exp(i\lambda_0 n\tau)$.

Notice $\frac{\left|\overline{E}_{\lambda_0}\right|}{p_0}\leq \xi<\frac{1}{\pi}<\frac{1}{2N}\frac{\sin(R_0N/2)}{\sin(R_0/2)}$ according to the second inequality of \eqref{eqn:assumption}. Then,
\begin{equation}\label{eqn:prepare}
\begin{aligned}
&\sqrt{\left(\frac{\sin(R_0N/2)}{N\sin(R_0/2)}+\frac{\left|\overline{E}_{\lambda_0}\right|}{p_0}\cos(\theta_E)\right)^2+\left(\frac{\left|\overline{E}_{\lambda_0}\right|}{p_0}\sin(\theta_E)\right)^2}+\xi_1+\xi\delta/2+\delta^2/200\\
\geq &\sqrt{\left(1+\frac{\left|\overline{E}_{\lambda_0}\right|}{p_0}\cos(\theta_E)\right)^2+\left(\frac{\left|\overline{E}_{\lambda_0}\right|}{p_0}\sin(\theta_E)\right)^2}\,.
\end{aligned}
\end{equation}
Because $\left(\sqrt{(x+a)^2+b^2}\right)'\geq \frac{1}{\sqrt{2}}$ when $x,a,b>0$ and $x+a>b$, we obtain
\[
\begin{aligned}
&\sqrt{\left(1+\frac{\left|\overline{E}_{\lambda_0}\right|}{p_0}\cos(\theta_E)\right)^2+\left(\frac{\left|\overline{E}_{\lambda_0}\right|}{p_0}\sin(\theta_E)\right)^2}\\
&-\sqrt{\left(\frac{\sin(R_0N/2)}{N\sin(R_0/2)}+\frac{\left|\overline{E}_{\lambda_0}\right|}{p_0}\cos(\theta_E)\right)^2+\left(\frac{\left|\overline{E}_{\lambda_0}\right|}{p_0}\sin(\theta_E)\right)^2}\\
\geq & \left(1-\frac{\sin(R_0N/2)}{N\sin(R_0/2)}\right)\frac{1}{\sqrt{2}}\,,
\end{aligned}
\]
where we see $a=\left|\overline{E}_{\lambda_0}\right|\cos(\theta_E)$ and $b=\left|\overline{E}_{\lambda_0}\right|\sin(\theta_E)$. Plugging this back into \eqref{eqn:prepare}, we obtain
\begin{equation}\label{eqn:start}
\frac{\sin(R_0N/2)}{N\sin(R_0/2)}\geq 1- \sqrt{2}(\xi_1+\xi\delta/2+\delta^2/200)\,,
\end{equation}
According to \eqref{eqn:prerq_new1.5}, we first have
\[
\frac{\sin(\delta/2)}{N\sin(\delta/(2N))}< \frac{2\sin(\delta/2)}{\delta\cos(\delta/10)}< 1- \sqrt{2}(\xi_1+\xi\delta/2+\delta^2/200)\,.
\]
Thus, there exists $\delta_1<\delta$ such that 
\[
\frac{\sin(\delta_1/2)}{N\sin(\delta_1/(2N))}= 1- \sqrt{2}(\xi_1+\xi\delta/2+\delta^2/200)\,.
\]
Using similar argument as the proof of \cref{prop:qcels_simple_full}, we have $|R_0|\leq \frac{\delta_1}{N}$. Then, similar to the previous argument, we can improve \eqref{eqn:start}, meaning that $R_0$ should satisfy the following inequality
\[
\frac{\sin(R_0N/2)}{N\sin(R_0/2)}\geq 1- \sqrt{2}(\xi_1+\xi\delta_1/2+\delta^2_1/200)
\]
Because
\[
1- \sqrt{2}(\xi_1+\xi\delta_1/2+\delta^2_1/200)>1- \sqrt{2}(\xi_1+\xi\delta/2+\delta^2/200)=\frac{\sin(\delta_1/2)}{N\sin(\delta_1/(2N))}\,,
\]
there exists $\delta_2<\delta_1$ such that
\[
\frac{\sin(\delta_2/2)}{N\sin(\delta_2/(2N))}= 1- \sqrt{2}(\xi_1+\xi\delta_1/2+\delta^2_1/200)
\]
and $|R_0|\leq \frac{\delta_2}{N}$. Because $h_1(x)=\frac{\sin(x/2)}{N\sin(x/(2N))}+\frac{\sqrt{2}x^2}{200}$ is a concave function ($x\in[0,\pi]$) with $h_1(0)=1,h'_1(0)=0$ and $h_2(x)=1- \sqrt{2}(\xi_1+\xi x/2)$ is a linearly decreasing function that satisfies $0<h_2(0)<1$. Thus, if $\delta_{\mathrm{new}}$ satisfies
\[
\Theta\left(\sqrt{1-p_r(I,I')}\right)\leq \delta_{\mathrm{new}}<\delta\,,\]
and
\begin{equation}\label{eqn:final}
\frac{\sin(\delta_{\mathrm{new}}/2)}{N\sin(\delta_{\mathrm{new}}/(2N))}< \frac{2\sin(\delta_{\mathrm{new}}/2)}{\delta_{\mathrm{new}}\cos(\delta_{\mathrm{new}}/10)}<1- \sqrt{2}(\xi_1+\xi\delta_{\mathrm{new}}/2+\delta^2_{\mathrm{new}}/200)\,,
\end{equation}
we must have $|R_0|\leq \frac{\delta_{\mathrm{new}}}{N}$.
\end{proof}
Similar to Appendix \ref{sec:pf_prop}, the next step is to do the iteration using the following lemma:
\begin{lem}\label{lem:1_new}
Define $T=N\tau$ and assume $p_r(I,I')$ is close enough to $1$. Given $0<\eta<1/2$, an integer $M>1$, a decreasing sequence $\{\xi_i\}^M_{i=0}$ with small enough $\xi_0$, and a decreasing sequence $\{\delta_i\}^M_{i=0}$ with $\delta_0\leq 4$ and $\delta_M\geq \Theta\left(\sqrt{1-p_r(I,I')}\right)$. Assume $\xi_0,\delta_0,N,N_s$ satisfy the conditions in \cref{lem:small_p0}. If $q=\Theta(p_0\delta^2_M)$,
\begin{equation}\label{eqn:prerq_new225}
NN_s=\Omega\left(p^{-2}_0\delta^2_i\xi^{-2}_{i+1}\mathcal{F}^2\mathrm{polylog}\left(M\xi^{-1}_{i+1}\eta^{-1}p^{-1}_0\mathcal{F}\right)\right)\,,
\end{equation}
and
\begin{equation}\label{eqn:prerq_new24}
\frac{1}{1-\sqrt{2}(\xi_{i+1}+\xi_0\delta_{i+1}/2+\delta^2_{i+1}/200)}\leq  \frac{\delta_{i+1}\cos(\delta_{i+1}/10)}{2\sin(\delta_{i+1}/2)}
\end{equation}
for all $0\leq i\leq M-1$. Then, 
\begin{equation}\label{eqn:distance_new}
\mathbb{P}\left(\left|(\theta^*-\lambda_0)\ \mathrm{mod}\ [-\pi/\tau,\pi/\tau] \right|<\frac{\delta_M}{T}\right)\geq 1-\eta\,.
\end{equation}
\end{lem}
\begin{proof}[Proof of \cref{lem:1_new}] First, according to Lemmas \ref{lem:small_p0} and \ref{lem:2}, we have
\[
\mathbb{P}\left(\left|(\theta^*-\lambda_0)\ \mathrm{mod}\ [-\pi/\tau,\pi/\tau] \right|<\frac{\delta_1}{T}\right)\geq 1-\eta/2,\quad \mathbb{P}\left(|\overline{E}_{\theta^*}|+|\overline{E}_{\lambda_0}|\geq p_0\xi_0\right)\leq \eta/4\,.
\]
Combining \eqref{eqn:prerq_new225} ($i=1$) with \cref{lem:bound_of_En_q} \eqref{eqn:Enbound_new2_q}, we obtain
 \[
    \mathbb{P}\left(\sup_{\theta\in[\lambda_0-\frac{\delta_1}{T},\lambda_0+\frac{\delta_1}{T}]}\left|\overline{E}_{\theta}-\overline{E}_{\lambda_0}\right|\geq p_0\xi_2\right)\leq \frac{\eta}{4(M-1)}\,.
\]
Thus, with probability $1-\left(\
\frac{3\eta}{4}+\frac{\eta}{4(M-1)}\right)$, 
\[
R_0<\frac{\delta_1}{N},\quad |\overline{E}_{\theta^*}|+|\overline{E}_{\lambda_0}|<p_0\xi_0,\quad \left|\overline{E}_{\theta^*}-\overline{E}_{\lambda_0}\right|< p_0\xi_2\,.
\]
Similar to the proof of \cref{lem:2}, these inequalities imply
\[
\mathbb{P}\left(\left|(\theta^*-\lambda_0)\ \mathrm{mod}\ [-\pi/\tau,\pi/\tau] \right|<\frac{\delta_{\mathrm{new}}}{T}\right)\geq 1-\left(\frac{3\eta}{4}+\frac{\eta}{4(M-1)}\right)\,,
\]
where $\delta_{\mathrm{new}}$ is the unique solution to the following equation:
\[
\frac{1}{1-\sqrt{2}(\xi_2+\xi_0\delta_{\mathrm{new}}/2+\delta^2_{\mathrm{new}}/200)}=\frac{N\sin(\delta_{\mathrm{new}}/(2N))}{\sin(\delta_{\mathrm{new}}/2)}\,.
\]
Because 
\[
\frac{1}{1-\sqrt{2}(\xi_2+\xi_0\delta_2/2+\delta^2_{2}/200)}\leq \frac{\delta_2\cos(\delta_2/10)}{2\sin(\delta_2/2)}<\frac{N\sin(\delta_{2}/(2N))}{\sin(\delta_{2}/2)}\,,
\]
we must have $\delta_2>\delta_{\mathrm{new}}$ and
\[
\mathbb{P}\left(\left|(\theta^*-\lambda_0)\ \mathrm{mod}\ [-\pi/\tau,\pi/\tau] \right|<\frac{\delta_2}{T}\right)\geq 1-\left(\frac{3\eta}{4}+\frac{\eta}{4(M-1)}\right)\,.
\]
Combining this with second inequality of \eqref{eqn:prerq_new225} ($i=2$), with probability $1-\left(\frac{3\eta}{4}+\frac{\eta}{2(M-1)}\right)$, we obtain
\[
R_0<\frac{\delta_2}{N},\quad |\overline{E}_{\theta^*}|+|\overline{E}_{\lambda_0}|<p_0\xi_0,\quad \left|\overline{E}_{\theta^*}-\overline{E}_{\lambda_0}\right|< p_0\xi_3\,,
\]
which implies
\[
\mathbb{P}\left(\left|(\theta^*-\lambda_0)\ \mathrm{mod}\ [-\pi/\tau,\pi/\tau] \right|<\frac{\delta_3}{T}\right)\geq 1-\left(\frac{3\eta}{4}+\frac{\eta}{2(M-1)}\right)\,.
\]
Doing this repetitively, we prove \eqref{eqn:distance_new}.
\end{proof}

Now, we are ready to prove \cref{thm:qcels_small_p0_full}.

\begin{proof}[Proof of \cref{thm:qcels_small_p0_full}] According to the conditions of \cref{lem:1_new}, the smallest $\delta_M$ that we can choose is $\Theta\left(\sqrt{1-p_r(I,I')}\right)$. For fixed  decreasing sequence $\{\delta_i\}^M_{i=0}$, according to \eqref{eqn:pre_small_p0_0} and \eqref{eqn:prerq_new24}, to make the depth smaller than $\delta_M/\epsilon$, we should choose
\[
 \xi_0=\mathcal{O}\left(\min\{\delta^2_0,\delta_M\}\right),\quad \xi_i=\mathcal{O}\left(\delta^2_i\right)
\]
for all $1\leq i\leq M$. Then, according to the first inequality of \eqref{eqn:pre_small_p0} and \eqref{eqn:prerq_new225}, we set
\begin{equation}\label{eqn:RQforN}
    NN_s=\widetilde{\Omega}\left(\max_{0\leq i\leq M-1}\left\{p^{-2}_0\xi^{-2},p^{-2}_0\delta^2_i\xi^{-2}_{i+1}\right\}\right)
\end{equation}
Minimizing \eqref{eqn:RQforN} in $\{\delta_i\}^{M-1}_{i=0},\{\xi_i\}^M_{i=0}$ with fixed $\delta_M$, a proper choice of these parameters should be 
\[
\xi_i=\Theta\left(\delta^2_i\right),\quad \delta_i=\Theta\left(\delta^{\frac{2-(1/2)^i}{2-(1/2)^M}}_M\right)
\]
for $0\leq i\leq M$. Thus, to ensure
\[
\mathbb{P}\left(\left|(\theta^*-\lambda_0)\ \mathrm{mod}\ [-\pi/\tau,\pi/\tau] \right|<\frac{\delta_M}{T}\right)\geq 1-\eta\,,
\]
we can choose $NN_s=\Omega\left(\delta^{-\frac{4}{2-(1/2)^M}}_M\mathrm{polylog}(M\delta^{-1}_M\eta^{-1}p^{-1}_0\mathcal{F})\right)$. Set $\zeta=\frac{1}{2^M}\frac{1}{2-(1/2)^M}$, we conclude the proof.
\end{proof}

\section{Additional estimates}
\subsection{Properties of \texorpdfstring{$\left|(\exp(i\theta N)-1)/(\exp(i\theta)-1)\right|$}{Lg} when \texorpdfstring{$\theta\in[0,\pi/N]$}{Lg}}\label{sec:newub}

In this section, we will show that $\left|\frac{\exp(i\theta N)-1}{\exp(i\theta)-1}\right|$ is a decreasing and concave function when $\theta\in[0,\pi/N]$.

First, we have the bound
\begin{equation}\label{eqn:gb}
\left|\frac{\exp(i\theta N)-1}{\exp(i\theta)-1}\right|=\left|\sum_{j=0}^{N-1} e^{ij\theta}\right|\le N.
\end{equation}

To study the other properties of this function, we write it as
\[
\left|\frac{\exp(i\theta N)-1}{\exp(i\theta)-1}\right|=\left|\frac{\sin(\theta N/2)}{\sin(\theta/2)}\right|\,.
\]
Then, it suffices to study the function
\[
g(\theta)=\frac{\sin(\theta N)}{\sin(\theta)}\,,\quad \theta\in [0,\pi/(2N)]\,.
\]

When $N=2$, we obtain $g(\theta)=2\cos(\theta)$ is decreasing and concave when $\theta\in [0,\pi/4]$. Notice
\[
\frac{\sin(\theta N)}{\sin(\theta)}=\cos((N-1)\theta)+\cos(\theta)\frac{\sin((N-1)\theta)}{\sin(\theta)}\,.
\]
Because $\cos((N-1)\theta)$ and $\cos(\theta)$ are decreasing and concave functions when $\theta\in [0,\pi/(2N)]$, using induction argument, $g(\theta)$ is also a decreasing and concave function when $\theta\in [0,\pi/(2N)]$. 

\subsection{Lower bound for  \texorpdfstring{$\lim_{\theta\rightarrow\lambda_0}|f(\theta)|$}{Lg} when \texorpdfstring{$N>2$}{Lg}}\label{sec:lbgb}
First, notice
\[
\lim_{\theta\rightarrow\lambda_0}\frac{\exp(i(\lambda_0-\theta)T)-1}{\exp(i(\lambda_0-\theta)\tau)-1}=N\,.
\]
Then,
\[
\lim_{\theta\rightarrow\lambda_0}\mathrm{Re}(f(\theta))\geq N+\mathrm{Re}\left(\sum^{M-1}_{k=1}\frac{p_k}{p_0}\frac{\exp(i(\lambda_0-\lambda_k)T)-1}{\exp(i(\lambda_0-\lambda_k)\tau)-1}\right)-\frac{\overline{E}}{p_0}N\,,
\]
which implies
\begin{equation}\label{eqn:lbm}
\begin{aligned}
    \lim_{\theta\rightarrow\lambda_0}|f(\theta)|\geq &\lim_{\theta\rightarrow\lambda_0}\mathrm{Re}(f(\theta))\\
\geq &\left(N+\mathrm{Re}\left(\sum^{M-1}_{k=1}\frac{p_k}{p_0}\frac{\exp(i(\lambda_0-\lambda_k)T)-1}{\exp(i(\lambda_0-\lambda_k)\tau)-1}\right)\right)-\frac{\overline{E}}{p_0}N\,.
\end{aligned}
\end{equation}

Let $\alpha=1+\max_{c\in(0,\pi/2]}\frac{\sin(c)}{\pi+c}$ be defined as in \cref{eqn:alpha_const}. Now, we prove the claim: 
\[
g(\theta)=\mathrm{Re}\left(\frac{\exp(i\theta N)-1}{\exp(i\theta)-1}\right)\geq -(\alpha-1) N\,.
\]
Assume not, meaning there exists $\widetilde{\theta}$ such that
\begin{equation}\label{eqn:lbr}
g\left(\widetilde{\theta}\right)<-(\alpha-1) N,
\end{equation}
then similar to the argument in Appendix \ref{sec:newub}, we first have $|\widetilde{\theta}|<\frac{2\pi^2}{3N}<\pi$. Notice
\[
\begin{aligned}
g\left(\widetilde{\theta}\right)&=\frac{[\cos(N\widetilde{\theta})-1][\cos(\widetilde{\theta})-1]+\sin(N\widetilde{\theta})\sin(\widetilde{\theta})}{(\cos(\widetilde{\theta})-1)^2+\sin^2(\widetilde{\theta})}\\
&=\frac{\cos((N-1)\widetilde{\theta}/2)\sin(N\widetilde{\theta}/2)}{\sin(\widetilde{\theta}/2)}\\
&=\cos^2((N-1)\widetilde{\theta}/2)+\frac{1}{2}\frac{\sin((N-1)\widetilde{\theta})\cos(\widetilde{\theta}/2)}{\sin(\widetilde{\theta}/2)}\\
&\geq \frac{1}{2}\frac{\sin((N-1)\widetilde{\theta})\cos(\widetilde{\theta}/2)}{\sin(\widetilde{\theta}/2)}\,.
\end{aligned}
\]
Since $g\left(\widetilde{\theta}\right)<0$, we must have $\frac{\sin((N-1)\widetilde{\theta})}{\sin(\widetilde{\theta}/2)}<0$, which implies $\left|\widetilde{\theta}\right|>\frac{\pi}{N-1}$. Let $c^*=(N-1)|\widetilde{\theta}|-\pi$, then $0<c<\pi/2$ and 
\[
\left|g\left(\widetilde{\theta}\right)\right|\leq \frac{1}{2}\frac{\sin(c)}{\tan((c+\pi)/(2(N-1)))}\leq (N-1)\frac{\sin(c^*)}{\pi+c^*}\leq (N-1)\left(\max_{c\in(0,\pi/2]}\frac{\sin(c)}{\pi+c}\right)\,,
\]
where we use $\tan(\theta)\geq \theta$ when $\theta\in [0,\pi/2]$ in the second inequality. This contradicts to the assumption \eqref{eqn:lbr}. Thus, we must have $g(\theta)\geq -\left(\max_{c\in(0,\pi/2]}\frac{\sin(c)}{\pi+c}\right)N=-(\alpha-1)N$, which finally implies
\begin{equation}\label{eqn:lbgb1}
\lim_{\theta\rightarrow\lambda_0}|f(\theta)|\geq \left(1-\left(\alpha-1\right)\frac{1-p_0}{p_0}-\frac{\overline{E}}{p_0}\right)N
\end{equation}
\end{document}